\documentclass[journal,twocolumn]{IEEEtran}

\usepackage{cite}
\usepackage{algorithm, algorithmic, float}
\usepackage[font=small]{caption}
\usepackage[dvips]{graphicx}
\usepackage{multirow}
\usepackage{colortbl}
\usepackage{tabularx}
\usepackage{color}
\usepackage{setspace}
\usepackage{mathtools}
\usepackage{bm}
\usepackage{amsmath,amsfonts,amssymb,amsthm,epsfig,epstopdf,url,array}
\usepackage{dsfont}

\usepackage{subfig}
\usepackage{float}

\makeatletter
\def\th@remark{%
	\thm@headfont{\bfseries}%
	\normalfont 
	\thm@preskip\topsep \divide\thm@preskip\tw@
	\thm@postskip\thm@preskip
}
\makeatother

\DeclarePairedDelimiter{\floor}{\lfloor}{\rfloor}
\DeclarePairedDelimiter{\ceil}{\lceil}{\rceil}
\newcommand{\overbar}[1]{\mkern 1.5mu\overline{\mkern-1.5mu#1\mkern-1.5mu}\mkern 1.5mu}

\newcommand{\Mod}[1]{\ \mathrm{mod}\ #1}

\newcommand\cmt[1]{\textcolor{black}{#1}}
\newcommand\dmt[1]{\textcolor{black}{#1}}

\ifCLASSINFOpdf
\else
\fi

\IEEEoverridecommandlockouts
\begin{document}
%
\title{Capacity of Clustered Distributed Storage}
%
%
%

\author{Jy-yong~Sohn,~\IEEEmembership{Student Member,~IEEE,} Beongjun~Choi,~\IEEEmembership{Student Member,~IEEE,} Sung~Whan~Yoon,~\IEEEmembership{Member,~IEEE,}
        and~Jaekyun~Moon,~\IEEEmembership{Fellow,~IEEE}
\thanks{The authors are with the School of Electrical Engineering, Korea Advanced Institute of Science and Technology, Daejeon,
	34141, Republic of Korea (e-mail: \{jysohn1108, bbzang10, shyoon8\}@kaist.ac.kr, jmoon@kaist.edu). A part of this paper was presented \cite{sohn2016capacity} at the IEEE Conference on Communications (ICC), Paris, France, May 21-25, 2017.  
	This work is in part supported
	by the National Research Foundation of Korea under Grant No. 2016R1A2B4011298, and in part supported by the ICT R\&D program of MSIP/IITP [2016-0-00563, Research on Adaptive Machine Learning Technology Development for Intelligent Autonomous Digital Companion].}
}
%
%

\markboth{To Appear at IEEE Transactions on Information Theory}%
{Shell \MakeLowercase{\textit{et al.}}: Bare Demo of IEEEtran.cls for IEEE Journals}
%



\maketitle

\begin{abstract}
A new system model reflecting the clustered structure of distributed storage is suggested to investigate 
interplay between storage overhead and repair bandwidth as storage node failures occur.
Large data centers with multiple racks/disks or local networks of storage devices (e.g. sensor network) are good applications of the suggested clustered model.
In realistic scenarios involving clustered storage structures, repairing storage nodes using intact nodes residing in other clusters is more bandwidth-consuming than restoring nodes based on information from intra-cluster nodes. Therefore, it is important to differentiate between intra-cluster repair bandwidth and cross-cluster repair bandwidth in modeling distributed storage. Capacity of the suggested model is obtained as a function of fundamental resources of distributed storage systems, namely, node storage capacity, intra-cluster repair bandwidth and cross-cluster repair bandwidth. The capacity is shown to be asymptotically equivalent to a monotonic decreasing function of number of clusters, as the number of storage nodes increases without bound.
Based on the capacity expression, feasible sets of required resources which enable reliable storage are obtained in a closed-form solution. Specifically, it is shown that the cross-cluster traffic can be minimized to zero (i.e., intra-cluster local repair becomes possible) by allowing extra resources on storage capacity and intra-cluster repair bandwidth, according to the law specified in the closed-form. The network coding schemes with zero cross-cluster traffic are defined as \textit{intra-cluster repairable codes}, which are shown to be a class of the previously developed \textit{locally repairable codes}.
\end{abstract}


%
\IEEEpeerreviewmaketitle

\section{Introduction}
%
%
%
%
\IEEEPARstart{M}{any} 
enterprises, including Google, Facebook, Amazon and Microsoft, use cloud storage systems in order to support massive amounts of data storage requests from clients.
In the emerging Internet-of-Thing (IoT) era, the number of devices which generate data and connect to the network increases exponentially, so that efficient management of data center becomes a formidable challenge.
However, since cloud storage systems are composed of inexpensive commodity disks, failure events occur frequently, degrading the system reliability \cite{ghemawat2003google}. 

In order to ensure reliability of cloud storage, distributed storage systems (DSSs) with erasure coding have been considered to improve tolerance against storage node failures \cite{bhagwan2004total,dabek2004designing,rhea2003pond,shvachko2010hadoop, huang2012erasure, muralidhar2014f4}. 
In such systems, the original file is encoded and distributed into multiple storage nodes. When a node fails, a newcomer node regenerates the failed node by contacting a number of survived nodes. This causes traffic burden across the network, taking up significant repair bandwidth.
Earlier distributed storage systems utilized the 3-replication code: the original file was replicated three times, and the replicas were stored in three distinct nodes. The 3-replication coded systems require the minimum repair bandwidth, but incur high storage overhead. Reed-Solomon (RS) codes are also used (\textit{e.g.} HDFS-RAID in Facebook \cite{borthakur2010hdfs}), which allow minimum storage overhead; however, RS-coded systems suffer from high repair bandwidth. 

The pioneering work of \cite{dimakis2010network} on distributed storage systems focused on the relationship between two required resources, the storage capacity $\alpha$ of each node and the repair bandwidth $\gamma$, when the system aims to reliably store a file $\mathcal{M}$ under node failure events. The optimal $(\alpha, \gamma)$ pairs are shown to have a fundamental trade-off relationship, to satisfy the maximum-distance-separable (MDS) property (i.e., any $k$ out of $n$ storage nodes can be accessed to recover the original file) of the system.   
Moreover, the authors of \cite{dimakis2010network} obtained 
capacity $\mathcal{C}$, the maximum amount of reliably storable data, as a function of $\alpha$ and $\gamma$. 
The authors related the failure-repair process of a DSS with the multi-casting problem in network information theory, and exploited the fact that a cut-set bound is achievable by network coding \cite{ahlswede2000network}. Since the theoretical results of \cite{dimakis2010network}, explicit network coding schemes \cite{rashmi2009explicit, rashmi2011optimal,shah2012interference}
which achieve the optimal $(\alpha, \gamma)$ pairs have also been suggested. These results are based on the assumption of homogeneous systems, i.e., each node has the same storage capacity and repair bandwidth. 

However, in real data centers, storage nodes are dispersed into multiple clusters (in the form of disks or racks) \cite{ford2010availability,huang2012erasure,muralidhar2014f4}, allowing high reliability against both node and rack failure events. 
In this clustered system, repairing a failed node gives rise to both intra-cluster and cross-cluster repair traffic.
While the current data centers have abundant intra-rack communication bandwidth, cross-rack communication is typically limited. 
According to \cite{rashmi2013solution}, nearly a $180$TB of cross-rack repair bandwidth is required everyday in the Facebook warehouse, limiting cross-rack communication for foreground map-reduce jobs. Moreover, surveys \cite{ahmad2014shufflewatcher, benson2010network, vahdat2010scale} on network traffic within data centers show that cross-rack communication is \textit{oversubscribed}; the available cross-rack communication bandwidth is typically $5-20$ times lower than the  intra-rack bandwidth in practical systems.
Thus, a new system model which reflects the imbalance between intra- and cross-cluster repair bandwidths is required.

\subsection{Main Contributions}
This paper suggests a new system model for \textit{clustered DSS} to reflect the clustered nature of real distributed storage systems wherein an imbalance exists between intra- and cross-cluster repair burdens.
This model can be applied to not only large data centers, but also local networks of storage devices 
such as the sensor networks or home clouds which are expected to be prevalent in the IoT era.
This model is also more general in the sense that when the intra- and cross-cluster repair bandwidths are \cmt{set to be equal}, the resulting structure reduces to the original DSS model of \cite{dimakis2010network}. 
\cmt{This paper only considers recovering a single node failure at a time, as in \cite{dimakis2010network}.}
The main contributions of this paper can be seen as twofold: one is the derivation of a closed-form expression for capacity, and the other is the analysis on feasible sets of system resources which enable reliable storage.

\subsubsection{Closed-form Expression for Capacity}
\cmt{Under the setting of functional repair,} storage capacity $\mathcal{C}$ of the clustered DSS is obtained as a function of node storage capacity $\alpha$, intra-cluster repair bandwidth $\beta_I$ and cross-cluster repair bandwidth $\beta_c$. 
The existence of the cluster structure manifested as the imbalance between intra/cross-cluster traffics makes the capacity analysis challenging;
Dimakis' proof in \cite{dimakis2010network} cannot be directly extended to handle the problem at hand.
We show that symmetric repair (obtaining the same amount of information from each helper node) is optimal in the sense of maximizing capacity given the storage node size and total repair bandwidth, as also shown in \cite{ernvall2013capacity} for the case of varying repair bandwidth across the nodes. 
However, we stress that in most practical scenarios, the need is greater for reducing cross-cluster communication burden, and we show that this is possible by trading with reduced overall storage capacity and/or increasing intra-repair bandwidth.
Based on the derived capacity expression, we analyzed how the storage capacity $\mathcal{C}$ changes as a function of $L$, the number of clusters. It is shown that the capacity is asymptotically equivalent to $\underline{C}$, some monotonic decreasing function of $L$.


\subsubsection{Analysis on Feasible $(\alpha, \beta_I, \beta_c)$ Points}
Given the need for reliably storing file $\mathcal{M}$, the set of required resource pairs, node storage capacity $\alpha$, intra-cluster repair bandwidth $\beta_I$ and cross-cluster repair bandwidth $\beta_c$, which enables $\mathcal{C}(\alpha, \beta_I, \beta_c) \geq \mathcal{M}$
is obtained in a closed-form solution.
In the analysis, we introduce $\epsilon = \beta_c/\beta_I$, a useful parameter \cmt{which measures the ratio of the cross-cluster repair burden (per node) to the intra-cluster burden. This parameter represents how scarce the available cross-cluster bandwidth is, compared to the abundant intra-cluster bandwidth.}
We here stress that the special case of $\epsilon=0$ corresponds to the scenario where repair is done only locally via intra-cluster communication, i.e., when a node fails the repair process requires intra-cluster traffic only without any cross-cluster traffic. Thus, the analysis on the $\epsilon=0$ case provides a guidance on the network coding for data centers for the scenarios where the available cross-cluster (cross-rack) bandwidth is very scarce.

Similar to the non-clustered case of \cite{dimakis2010network}, the required node storage capacity and the required repair bandwidth show a trade-off relationship. In the trade-off curve, two extremal points - the minimum-bandwidth-regenerating (MBR) point and the minimum-storage-regenerating (MSR) point - have been further analyzed for various $\epsilon$ values.
Moreover, from the analysis on the trade-off curve, it is shown that the minimum storage overhead $\alpha = \mathcal{M}/k$ is achievable if and only if $\epsilon \geq \frac{1}{n-k}$. This implies that in order to reliably store file $\mathcal{M}$ with minimum storage $\alpha = \mathcal{M}/k$, sufficiently large cross-cluster repair bandwidth satisfying $\epsilon \geq \frac{1}{n-k}$ is required.  
Finally, for the scenarios with the abundant intra-cluster repair bandwidth, the minimum required cross-cluster repair bandwidth $\beta_c$ to reliably store file $\mathcal{M}$ is obtained as a function of node storage capacity $\alpha$.



\subsection{Related Works}
Several researchers analyzed practical distributed storage systems
with a goal in mind to reflect the non-homogeneous nature of storage nodes \cite{ernvall2013capacity, yu2015tradeoff, akhlaghi2010fundamental, shah2010flexible, gaston2013realistic, prakash2016generalization, prakash2017storage, choi2017secure, hu2017optimal, calis2016architecture,  ye2017explicit, tamo2016optimal}. A heterogeneous model was considered in \cite{ernvall2013capacity, yu2015tradeoff} where the storage capacity and the repair bandwidth for newcomer nodes are generally non-uniform. Upper/lower capacity bounds for the heterogeneous DSS are obtained in \cite{ernvall2013capacity}. 
An asymmetric repair process is considered in \cite{akhlaghi2010fundamental}, coining the terms, cheap and expensive nodes, based on the amount of data that can be transfered to any newcomer. The authors of \cite{shah2010flexible} considered a flexible distributed storage system where the amount of information from helper nodes may be non-uniform, as long as the total repair bandwidth is bounded from above.
The view points taken in these works are different from ours in that we adopt a notion of cluster and introduce imbalance between intra- and cross-cluster repair burdens.

\cmt{Recently, some researchers considered the clustered structure of data centers \cite{sohn2016capacity, gaston2013realistic, prakash2016generalization, prakash2017storage, choi2017secure, hu2017optimal, calis2016architecture,  ye2017explicit, tamo2016optimal}. 
Some recent works \cite{gaston2013realistic, prakash2016generalization, prakash2017storage, choi2017secure} provided new system models for clustered DSS and shed light on fundamental aspects of the suggested system. In \cite{gaston2013realistic}, the idea of \cite{akhlaghi2010fundamental} is developed to a two-rack system, by setting the communication burden within a rack much lower than the burden across different racks, similar to our analysis. 
However, the authors of \cite{gaston2013realistic} only considered systems with two racks, while the current paper considers a general setting of $L$ racks (clusters), and provides mathematical analysis on how the number of clusters (i.e., the dispersion of nodes) affects the capacity of clustered distributed storage. 
Similar to the present paper, the authors of \cite{prakash2016generalization, prakash2017storage} obtained the capacity of clustered distributed storage, and provided capacity-achieving regenerating coding schemes. However, the coding schemes considered in \cite{prakash2016generalization, prakash2017storage} do not satisfy the MDS property, and the capacity expression is obtained for limited scenarios when intra-cluster repair bandwidth $\beta_I$ is set to its maximum value. In contrast, the current paper provides the capacity expression for general values of $\beta_I, \beta_c$ parameters, and analyzes the behavior of capacity as a function of the ratio $\epsilon = \beta_c/\beta_I$ between intra- and cross-cluster repair bandwidths. 
\dmt{Moreover, unlike in the previous work, the capacity-achieving coding schemes (whose existence is shown) here satisfy the MDS property.}
In \cite{choi2017secure}, the security issue in clustered distributed storage systems is considered, and the maximum amount of securely storable data in the existence of passive eavesdroppers is obtained.}

\cmt{There also have been some recent works \cite{hu2017optimal, calis2016architecture,  ye2017explicit, tamo2016optimal} on network code design appropriate for clustered distributed storage. 
Motivated by the limited available cross-rack repair bandwidth in real data centers, the work of \cite{hu2017optimal} provides a network coding scheme which minimizes the cross-rack bandwidth in clustered distributed storage systems. However, the suggested coding scheme is applicable for some limited $(n,k,L)$ parameters and the minimum storage overhead ($\alpha = \mathcal{M}/k$) setting. On the other hand, the current paper provides the capacity expression for general $(n,k,L,\alpha)$ setting, and proves the existence of capacity-achieving coding scheme.
The authors of \cite{calis2016architecture} proposed coding schemes tolerant to rack failure events in multi-rack storage systems, but have not addressed the imbalance between intra- and cross-cluster repair burdens in node failure events, which is an important aspect of the current paper.
The authors of \cite{ye2017explicit} considers the scenario of having grouped (clustered) storage nodes where nodes in the same group are more accessible to each other, compared to the nodes in other groups. However, the focus is different to the present paper: \cite{ye2017explicit} focuses on the code construction which has the minimum amount of accessed data (called the optimal access property), while the scope of the present paper is on finding the optimal trade-off between the node storage capacity and the repair bandwidth, as a function of the ratio $\epsilon$ of intra- and cross-cluster communication burdens. Finally, a locally repairable code which can repair arbitrary node within each group is suggested in \cite{tamo2016optimal}; this code can suppress the inter-group repair bandwidth to zero. However, the coding scheme is suggested for the $\epsilon=0$ case only, while the present paper provides the capacity expression for general $0 \leq \epsilon \leq 1$, and proves the existence of an optimal coding scheme.}

Compared to the conference version \cite{sohn2016capacity} of the current work, this paper provides the formal proofs for the capacity expression, and obtains the feasible $(\alpha, \gamma)$ region for $0 \leq \epsilon \leq 1$ setting\footnote{\cmt{The reason why the present paper considers this regime is provided in Section \ref{Section:assumptions}.}} (only $\epsilon=0$ is considered in \cite{sohn2016capacity}). The present paper also shows the behavior of capacity as a function of $L$, the number of clusters, and provides the sufficient and necessary conditions on $\epsilon = \beta_c/\beta_I$, to achieve the minimum storage overhead $\alpha = \mathcal{M}/k$. Finally, the asymptotic behaviors of the MBR/MSR points are investigated in this paper, and the connection between what we call the intra-cluster repairable codes and the existing locally repairable codes \cite{papailiopoulos2014locally} is revealed.


\subsection{Organization}
This paper is organized as follows. Section \ref{Section:Background} reviews preliminary materials about distributed storage systems and the information flow graph, an efficient tool for analyzing DSS. Section \ref{Section:Capacity_of_Clustered_DSS} proposes a new system model for the clustered DSS, and derives a closed-form expression for the storage capacity of the clustered DSS. The behavior of the capacity curves is also analyzed in this section. 
Based on the capacity expression, Section \ref{Section:analysis on feasible points} provides results on the feasible resource pairs which enable reliable storage of a given file.
Further research topics on clustered DSS are discussed in Section \ref{Section:Future_Works}, and Section \ref{Section:Conclusion} draws conclusions.

\section{Background}\label{Section:Background}

\subsection{Distributed Storage System}
Distributed storage systems can maintain reliability by means of erasure coding \cite{dimakis2011survey}. The original data file is spread into $n$ potentially unreliable nodes, each with storage size $\alpha$. When a node fails, it is regenerated by contacting $d < n$ helper nodes and obtaining a particular amount of data, $\beta$, from each helper node. The amount of communication burden imposed by one failure event is called the repair bandwidth, denoted as $\gamma = d \beta$.
When the client requests a retrieval of the original file, assuming all failed nodes have been repaired, access to any $k<n$ out of $n$ nodes must guarantee a file recovery. The ability to recover the original data using any $k<n$ out of $n$ nodes is called the maximal-distance-separable (MDS) property.   
Distributed storage systems can be used in many applications such as large data centers, peer-to-peer storage systems and wireless sensor networks \cite{dimakis2010network}. 

\subsection{Information Flow Graph}\label{Section:Info_flow_graph}
Information flow graph is a useful tool to analyze the amount of information flow from source to data collector in a DSS, as utilized in \cite{dimakis2010network}. 
It is a directed graph consisting of three types of nodes: data source $\mathrm{S}$, data collector $\mathrm{DC}$, and storage nodes $x^i$ as shown in Fig. \ref{Fig:information_flow_graph}.
Storage node $x^{i}$ can be viewed as consisting of input-node $x_{in}^i$ and output-node $x_{out}^i$, which are responsible for the incoming and outgoing edges, respectively. $x_{in}^i$ and $x_{out}^i$ are connected by a directed edge with capacity identical to the storage size $\alpha$ of node $x^i$. 

Data from source $S$ is stored into $n$ nodes. This process is represented by $n$ edges going from $\mathrm{S}$ to $\{x^i\}_{i=1}^n$, where each edge capacity is set to infinity.
A failure/repair process in a DSS can be described as follows.
When a node $x^{j}$ fails, a new node $x^{n+1}$ joins the graph by connecting edges from $d$ survived nodes, where each edge has capacity $\beta$.
After all repairs are done, data collector $\mathrm{DC}$ chooses arbitrary $k$ nodes to retrieve data, as illustrated by the edges connected 
from $k$ survived nodes with infinite edge capacity.
Fig. \ref{Fig:information_flow_graph} gives an example of information flow graph representing a distributed storage system with $n=4, k=3, d=3$.


\begin{figure}[!t]
	\centering
	\includegraphics[height=35mm]{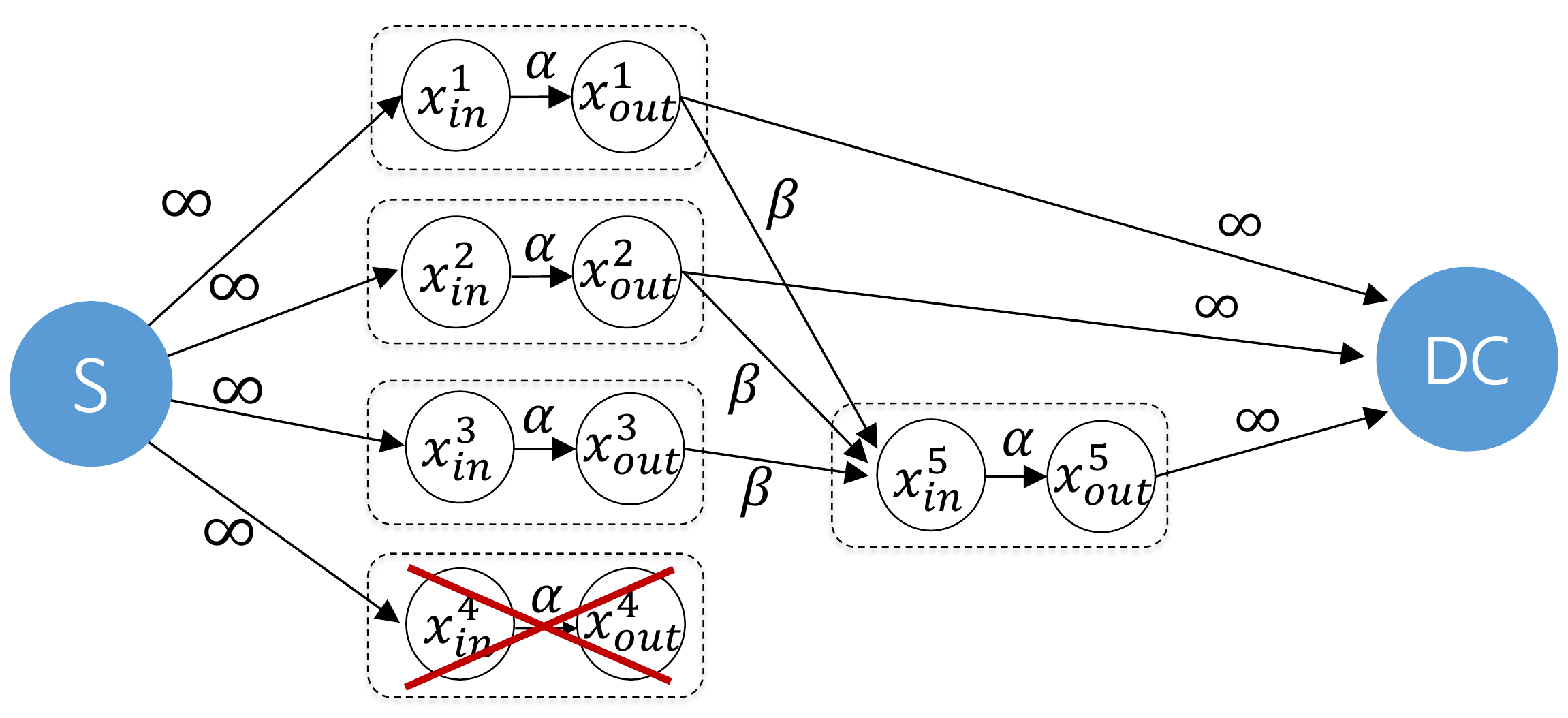}
	\caption{Information flow graph ($n=4, k=3, d=3$)}
	\label{Fig:information_flow_graph}
\end{figure}

%

\begin{figure}[!t]
	\centering
	\includegraphics[height=25mm]{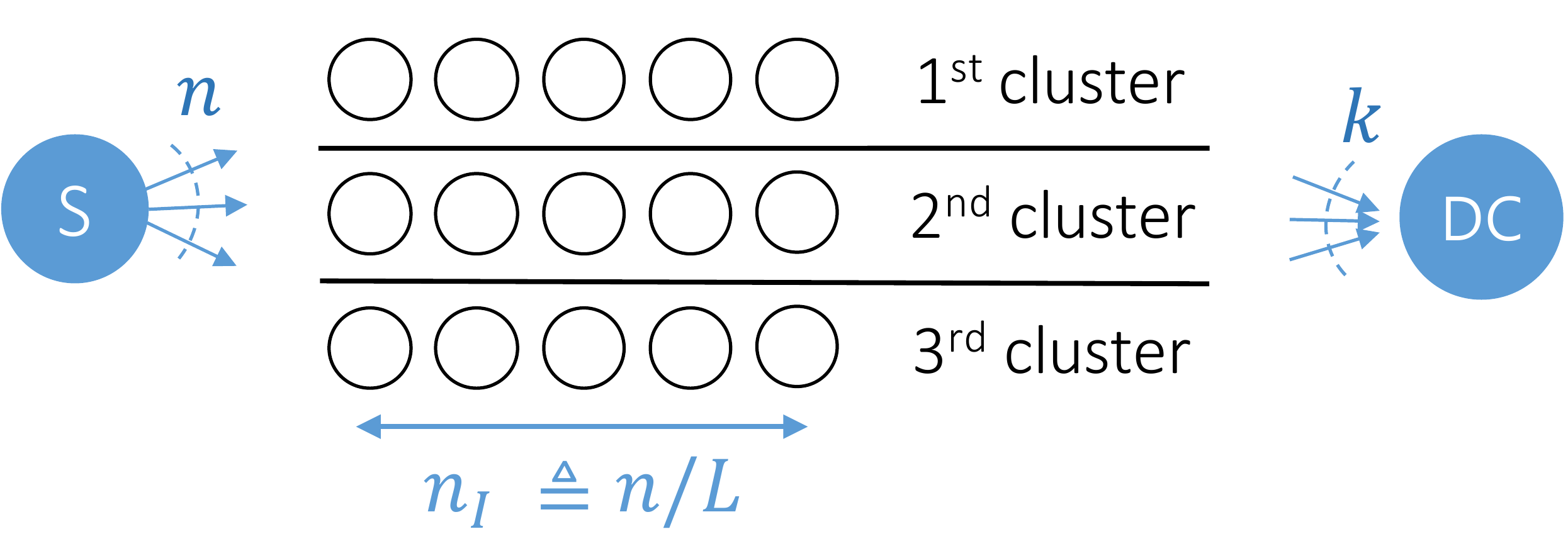}
	\caption{Clustered distributed storage system ($n=15,L=3$)}
	\label{Fig:Layered_DSS}
\end{figure}

\subsection{Notation used in the paper}\label{Subsection: notation}
This paper requires many notations related to graphs, because it deals with information flow graphs. Here we provide the definition of each notation used in the paper.
For the given system parameters, 
we denote $\mathcal{G}$ as the set of all possible information flow graphs. A graph $G\in \mathcal{G}$ is denoted as $G=(V,E)$ where $V$ is the set of vertices and $E$ is the set of edges in the graph. 
For a given graph $G \in \mathcal{G}$, we call a set $c \subset E$ of edges as \textit{cut-set} \cite{bang2008digraphs} if it satisfies the following: every directed path from $\mathrm{S}$ to $\mathrm{DC}$ includes at least one edge in $c$. 
An arbitrary cut-set $c$ is usually denoted as $c=(U, \overbar{U})$ where $U \subset V$ and $\overbar{U} = V\setminus U$ (the complement of $U$) satisfy the following: the set of edges from $U$ to $\overbar{U}$ is the given cut-set $c$.
The set of all cut-sets available in $G$ is denoted as $C(G)$. 
For a graph $G\in \mathcal{G}$ and a cut-set $c \in C(G)$, we denote the sum of edge capacities for edges in $c$ as $w(G,c)$, which is called the \textit{cut-value} of $c$.

A vector is denoted as $\textbf{v}$ using the 
bold notation. For a vector $\textbf{v}$, the transpose of the vector is denoted as $\textbf{v}^T$. 
A set is denoted as $X = \{x_1, x_2, \cdots, x_k\}$, while a sequence $x_1, x_2, \cdots, x_N$ is denoted as $(x_n)_{n=1}^{N}$, or simply $(x_n)$. 
For given sequences $(a_n)$ and $(b_n)$, we use the term ``$a_n$ is asymptotically equivalent to $b_n$" \cite{erdélyi1956asymptotic} if and only if
\begin{equation}
\lim\limits_{n \rightarrow \infty} \frac{a_n}{b_n} = 1.
\end{equation}
We utilize a useful notation:
\begin{equation*}
\mathds{1}_{i=j} = 
\begin{cases}
1, & \text{ if } i=j \\
0, & \text{ otherwise.}
\end{cases}
\end{equation*}
For a positive integer $n$, we use $[n]$ as a simplified notation for the set $\{1,2,\cdots, n\}$. For a non-positive integer $n$, we define  $[n] = \emptyset$. 
Each storage node is represented as either $x^t = (x_{in}^t, x_{out}^t)$ as defined in Section \ref{Section:Info_flow_graph},
or $N(i,j)$ as defined in (\ref{Eqn:set_of_nodes}).
Finally, important parameters used in this paper are summarized in Table \ref{Table:Params}.



\begin{table}[t]
	\caption{Parameters used in this paper }
	\centering
	\label{Table:Params}
	\begin{tabular}{|c|c|}
		\hline
		$n$ & number of storage nodes \tabularnewline
		\hline
		$k$ & number of DC-contacting nodes \tabularnewline
		\hline
		$L$ & number of clusters \tabularnewline
		\hline
		$n_I = n/L$ & number of nodes in a cluster \tabularnewline
		\hline
		$d_I = n_I - 1$ & number of intra-cluster helper nodes \tabularnewline
		\hline
		$d_c = n- n_I$ & number of cross-cluster helper nodes \tabularnewline
		\hline
		$\mathbb{F}_q$ & \cmt{base field which contains each symbol} \tabularnewline
		\hline
		$\alpha$ & storage capacity of each node \tabularnewline
		\hline
		$\beta_I $ & intra-cluster repair bandwidth (per node) \tabularnewline
		\hline
		$\beta_c $ & cross-cluster repair bandwidth (per node) \tabularnewline
		\hline
		$\gamma_I = d_I\beta_I$ & intra-cluster repair bandwidth \tabularnewline
		\hline
		$\gamma_c = d_c\beta_c$ & cross-cluster repair bandwidth \tabularnewline
		\hline
		$\gamma = \gamma_I + \gamma_c$ & repair bandwidth \tabularnewline
		\hline
		$\epsilon = \beta_c/\beta_I$ & ratio of $\beta_c$ to $\beta_I$ ($0 \leq \epsilon \leq 1$) \tabularnewline
		\hline
		$\xi = \gamma_c/\gamma$ & ratio of $\gamma_c$ to $\gamma$ ($ 0 \leq \xi < 1$) \tabularnewline
		\hline
		$R=k/n$ & ratio of $k$ to $n$ ($0 < R \leq 1$) \tabularnewline
		\hline	
	\end{tabular}
\end{table}

\newtheorem{theorem}{Theorem}
\newtheorem{lemma}{Lemma}
\newtheorem{corollary}{Corollary}
\newtheorem{definition}{Definition}
\newtheorem{prop}{Proposition}
\theoremstyle{remark}
\newtheorem{remark}{Remark}

\section{Capacity of Clustered DSS}\label{Section:Capacity_of_Clustered_DSS}

\subsection{Clustered Distributed Storage System}\label{Section:Clustered_DSS}

A distributed storage system with multiple clusters is shown in Fig. \ref{Fig:Layered_DSS}. 
Data from source $\mathrm{S}$ is stored at $n$ nodes which are grouped into $L$ clusters. The number of nodes in each cluster is fixed and denoted as $n_I = n/L$. The storage size of each node is denoted as $\alpha$. 
When a node fails, a newcomer node is regenerated by contacting $d_I$ helper nodes within the same cluster, and $d_c$ helper nodes from other clusters.
\cmt{This paper considers functional repair\cite{dimakis2011survey} in the regeneration process; the newcomer node may store different content from that of the failed node, while maintaining the MDS property of the code.}
The amount of data a newcomer node receives within the same cluster is 
$\gamma_I = d_I \beta_I$ (each node equally contributes to $\beta_I$), and that from other clusters is $\gamma_c = d_c \beta_c$  (each node equally contributes to $\beta_c$).
Fig. \ref{Fig:Repair Process in clustered DSS} illustrates an example of information flow graph representing the repair process in a clustered DSS.

\begin{figure}[!t]
	\centering
	\includegraphics[height=30mm]{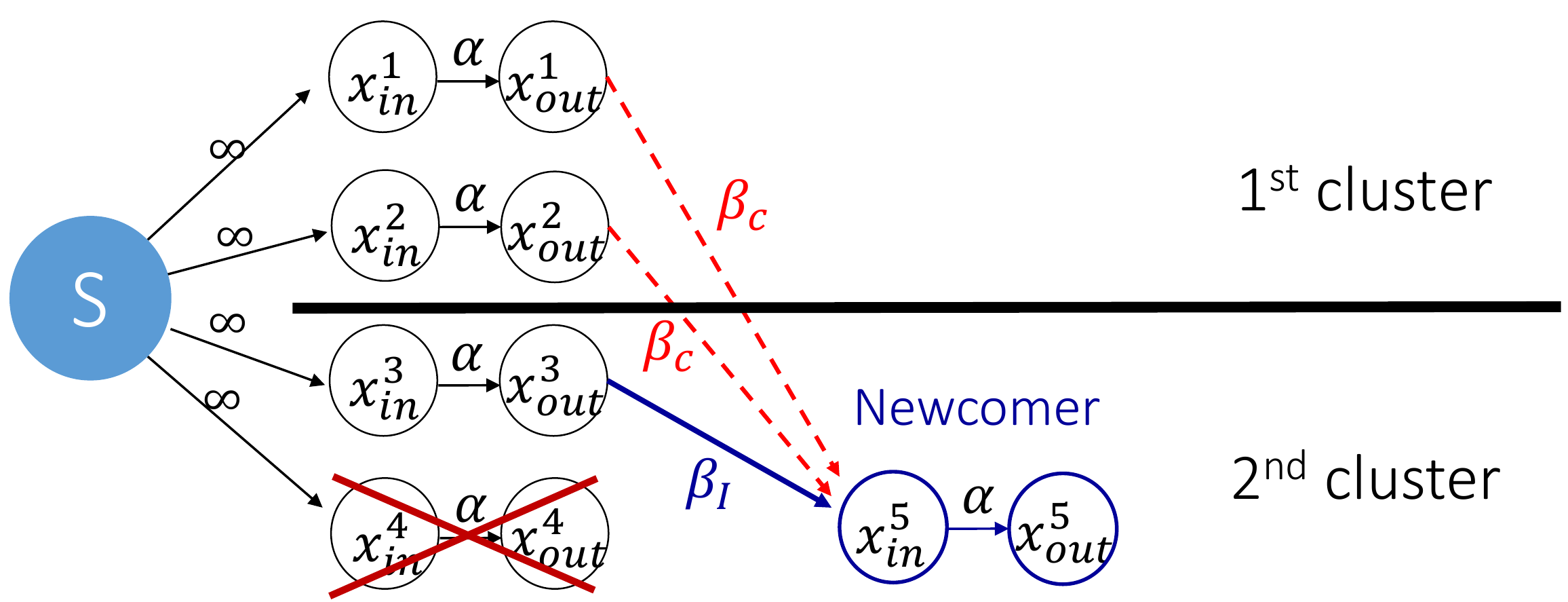}
	\caption{Repair process in clustered DSS ($n=4, L=2, d_I =1, d_c = 2$) }
	\label{Fig:Repair Process in clustered DSS}
\end{figure}


\subsection{Assumptions for the System}\label{Section:assumptions}
We assume that $d_c$ and $d_I$ have the maximum possible values ($d_c = n-n_I, d_I = n_I - 1$), since this is the capacity-maximizing choice, as formally stated in the following proposition. The proof of the proposition is in Appendix \ref{Section:max_helper_node_assumption}.

\begin{prop}\label{Prop:max_helper_nodes}
	Consider a clustered distributed storage system with given $\gamma$ and $\gamma_c$. Then, setting both $d_I$ and $d_c$ to their maximum values maximizes storage capacity. 
\end{prop}

Note that the authors of \cite{dimakis2010network} already showed that in the non-clustered scenario with given repair bandwidth $\gamma$, maximizing the number of helper nodes $d$ is the capacity-maximizing choice. Here, we are saying that a similar property also holds for clustered scenario considered in the present paper.
Under the setting of the maximum number of helper nodes, the overall repair bandwidth for a failure event is denoted as

\begin{equation}\label{Eqn:gamma}
\gamma = \gamma_I + \gamma_c = (n_I-1)\beta_I + (n-n_I)\beta_c
\end{equation}
Data collector $\mathrm{DC}$ contacts any $k$ out of $n$ nodes in the clustered DSS. 
Given that the typical intra-cluster communication bandwidth is larger than the cross-cluster bandwidth in real systems, 
we assume \begin{equation*}
\beta_I \geq \beta_c
\end{equation*} throughout the present paper; \cmt{this assumption limits our interest to $\epsilon \leq 1$.}
Moreover, motivated by the security issue, we assume that 
 a file cannot be retrieved entirely by contacting any single cluster having $n_I$ nodes. Thus, the number $k$ of nodes contacted by the data collector satisfies
\begin{align}\label{Eqn:k_constraint}
k &> n_I.
\end{align}
%
We also assume 
\begin{equation}\label{Eqn:L_constraint}
L \geq 2 
\end{equation}
which holds for most real DSSs. Usually, all storage nodes cannot be squeezed in a single cluster, i.e., $L=1$ rarely happens in practical systems, to prevent losing everything when the cluster is destroyed. Note that many storage systems \cite{huang2012erasure, muralidhar2014f4, rashmi2013solution} including those of Facebook uses $L=n$, i.e., every storage node reside in different racks (clusters), \dmt{to tolerate the rack failure events}. 
Finally, according to \cite{rashmi2013solution}, nearly $98 \%$ of data recoveries in real systems deal with single node recovery. In other words, the portion of simultaneous multiple nodes failure events is small. Therefore, the present paper focuses single node failure events.

%

\subsection{The closed-form solution for Capacity}

Consider a clustered DSS with fixed $n, k, L$ values. In this model, we want to find the set of \textit{feasible} parameters ($\alpha, \beta_I, \beta_c$) which enables storing data of size $\mathcal{M}$. 
In order to find the feasible set, min-cut analysis on the information flow graph is required, similar to \cite{dimakis2010network}.
Depending on the failure-repair process and $k$ nodes contacted by $\mathrm{DC}$, various information flow graphs can be obtained. 

Let $\mathcal{G}$ be the set of all possible flow graphs. 
Consider a graph $G^* \in \mathcal{G}$ with minimum min-cut, the construction of which is specified in Appendix \ref{Section:Proof of Thm 1}.
Based on the max-flow min-cut theorem in \cite{ahlswede2000network}, 
the maximum information flow from source to data collector for arbitrary $G \in \mathcal{G}$ is greater than or equal to 
\begin{equation*}\label{Eqn:capacity expression}
\mathcal{C}(\alpha, \beta_I, \beta_c) \coloneqq \text{min-cut of }G^*,
\end{equation*}
which is called the \textit{capacity} of the system. 
In order to send data $\mathcal{M}$ from the source to the data collector, $\mathcal{C} \geq \mathcal{M}$ should be satisfied. 
Moreover, if $\mathcal{C} \geq \mathcal{M}$ is satisfied, there exists a linear network coding scheme \cite{ahlswede2000network} to store a file with size $\mathcal{M}$. 
Therefore, the set of $(\alpha, \beta_I, \beta_c)$ points which satisfies $\mathcal{C} \geq \mathcal{M}$ is \textit{feasible} in the sense of reliably storing the original file of size $\mathcal{M}$.
Now, we state our main result in the form of a theorem which offers a closed-form solution for the capacity $\mathcal{C}(\alpha,  \beta_I, \beta_c)$ of the clustered DSS.
Note that setting $\beta_I = \beta_c$ reduces to capacity of the non-clustered DSS obtained in \cite{dimakis2010network}.


\begin{theorem} \label{Thm:Capacity of clustered DSS}
	The capacity of the clustered distributed storage system with parameters $(n,k,L,\alpha, \beta_I, \beta_c)$ is 
%

	\begin{equation}\label{Eqn:Capacity of clustered DSS_rev}
	\mathcal{C}(\alpha, \beta_I, \beta_c) = \sum_{i=1}^{n_I} \sum_{j=1}^{g_i} \min \{\alpha, \rho_i\beta_I + (n-\rho_i - j - \sum_{m=1}^{i-1}g_m) \beta_c \},
	\end{equation}
	where
	\begin{align}
	\rho_i &= n_I - i, \label{Eqn:rho_i}\\
	g_m &=
	\begin{cases}
	\floor{\frac{k}{n_I}} + 1, & m \leq (k \Mod{n_I}) \\
	\floor{\frac{k}{n_I}}, & otherwise.\label{Eqn:g_m}
	\end{cases}
	\end{align}
\end{theorem}

The proof is in Appendix \ref{Section:Proof of Thm 1}.
Note that the parameters used in the statement of Theorem \ref{Thm:Capacity of clustered DSS} have the following property, the proof of which is in Appendix \ref{Section:proof_of_omega_bound_gamma}.

\begin{prop}\label{Prop:omega_i_bounded_by_gamma}
	For every $(i,j)$ with $i \in [n_I], j \in [g_i]$, we have 
	\begin{equation}\label{Eqn:omega_is_bounded_by_gamma}
		\rho_i \beta_I + (n-\rho_i - j - \sum_{m=1}^{i-1}g_m)\beta_c \leq \gamma.
	\end{equation}
Moreover,
\begin{equation}\label{Eqn:sum of g is k}
\sum_{m=1}^{n_I} g_m = k
\end{equation}
holds.
\end{prop}


\subsection{Relationship between $\mathcal{C}$ and $\epsilon = \beta_c/\beta_I$}\label{Section:C_versus_kappa}

In this subsection, we analyze the capacity of a clustered DSS as a function of an important parameter 
\begin{equation}\label{Eqn:epsilon}
\epsilon \coloneqq  \beta_c/\beta_I,
\end{equation}
the cross-cluster repair burden per intra-cluster repair burden.
In Fig. \ref{Fig:capacity_versus_kappa_plot}, capacity is plotted as a function of $\epsilon$. From (\ref{Eqn:gamma}), the total repair bandwidth can be expressed as 
\begin{align}\label{Eqn:capacity versus kappa}
\gamma &= \gamma_I + \gamma_c = (n_I-1)\beta_I + (n-n_I)\beta_c \nonumber\\
&= \Big( n_I - 1 + (n-n_I)\epsilon \Big) \beta_I.
\end{align}
Using this expression, the capacity is expressed as 
\begin{equation}\label{Eqn:capacity_kappa} 
\mathcal{C}(\epsilon) = \sum_{i=1}^{n_I} \sum_{j=1}^{g_i} \min \Big\{\alpha, \frac{(n-\rho_i - j - \sum_{m=1}^{i-1}g_m)\epsilon + \rho_i }{(n-n_I)\epsilon + n_I - 1}   \gamma \Big\}.
\end{equation} 
For fair comparison on various $\epsilon$ values, 
capacity is calculated for a fixed ($n, k, L, \alpha, \gamma$) set.
The capacity is an increasing function of $\epsilon$ as shown in Fig. \ref{Fig:capacity_versus_kappa_plot}. 
%
This implies that for given resources $\alpha$ and $\gamma$, allowing a larger $\beta_c$ (until it reaches $\beta_I$) is always beneficial, in terms of storing a larger file. For example, under the setting in Fig. \ref{Fig:capacity_versus_kappa_plot}, allowing $\beta_c = \beta_I$ (i.e., $\epsilon = 1$) can store $\mathcal{M}=48$, while setting $\beta_c = 0$ (i.e., $\epsilon = 0$) cannot achieve the same level of storage.
This result is consistent with the previous work on asymmetric repair in \cite{ernvall2013capacity}, which proved that the symmetric repair maximizes capacity. 
Therefore, when the total communication amount $\gamma$ is fixed, a loss of storage capacity is the cost we need to pay in order to reduce the communication burden $\beta_c$ across different clusters.

\begin{figure}[t]
	\centering
	\includegraphics[width=90mm]{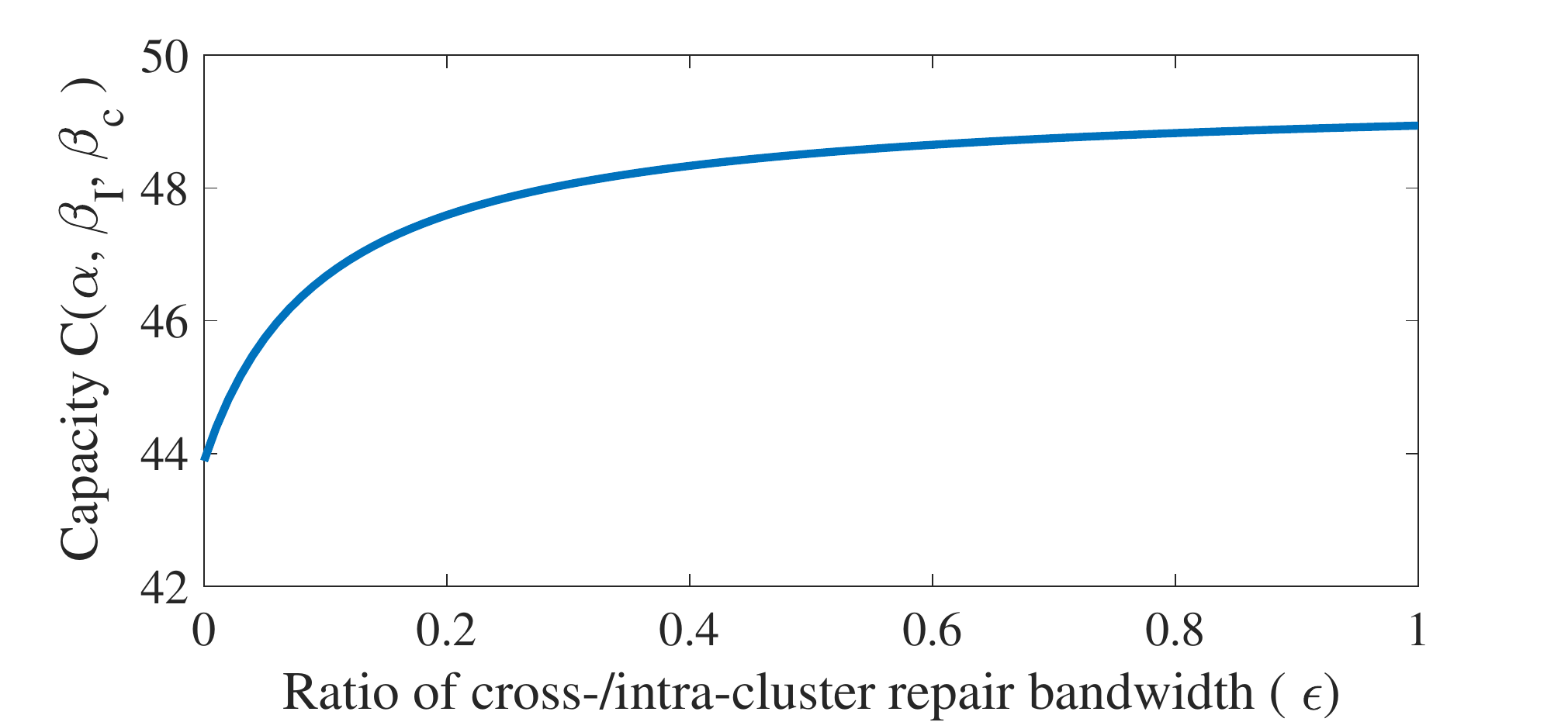}
	\caption{Capacity as a function of $\epsilon$, under the setting of $n=100, k=85, L=10, \alpha = 1, \gamma = 1$}
	\label{Fig:capacity_versus_kappa_plot}
\end{figure}

%
%

\subsection{Relationship between $\mathcal{C}$ and $L$}\label{Section:C_versus_L}

In this subsection, we analyze the capacity of a clustered DSS as a function of $L$, the number of clusters.  For fair comparison, 
	($n, k, \alpha, \gamma$) values are fixed for calculating capacity. 
	In Fig. \ref{Fig:capacity_versus_L_plot}, capacity curves for two scenarios are plotted over a range of $L$ values.
	First, the solid line corresponds to the 
	scenario when the system has abundant cross-rack bandwidth resources $\gamma_c$. In this ideal scenario which does not suffer from the over-subscription problem, the system can store $\mathcal{M}=80$ irrespective of the dispersion of nodes. 

However, consider a practical situation where \cmt{the available cross-rack bandwidth is scarce compared to the intra-rack bandwidth; for example, $\xi=\gamma_c/\gamma=1/5$}. 
The dashed line in Fig. \ref{Fig:capacity_versus_L_plot} corresponds to this scenario where the system has not enough cross-rack bandwidth resources.
In this practical scenario, reducing $L$ (i.e., gathering the storage nodes into a smaller number of clusters) increases capacity. However, note that sufficient dispersion of data into a fair number of clusters is typically desired, in order to guarantee the reliability of storage in rack-failure events. Finding the optimal number $L^*$ of clusters in this trade-off relationship remains as an important topic for future research.

\begin{figure}[t]
	\centering
	\includegraphics[width=85mm]{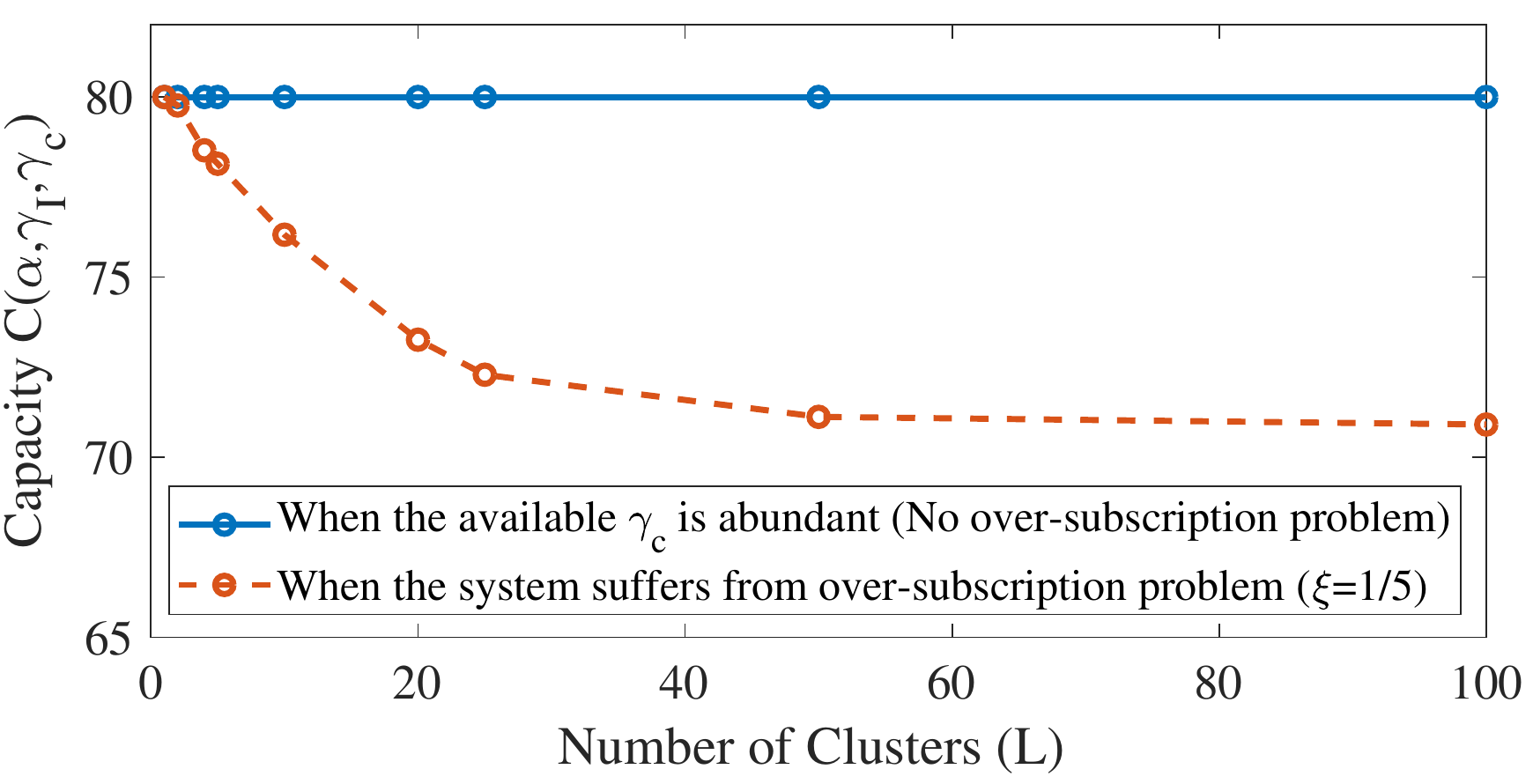}
	\caption{Capacity as a function of $L$, under the setting of $n=100, k=80,\alpha = 1, \gamma = 10$}
	\label{Fig:capacity_versus_L_plot}
\end{figure}

\begin{figure}[t]
	\centering
	\includegraphics[width=65mm]{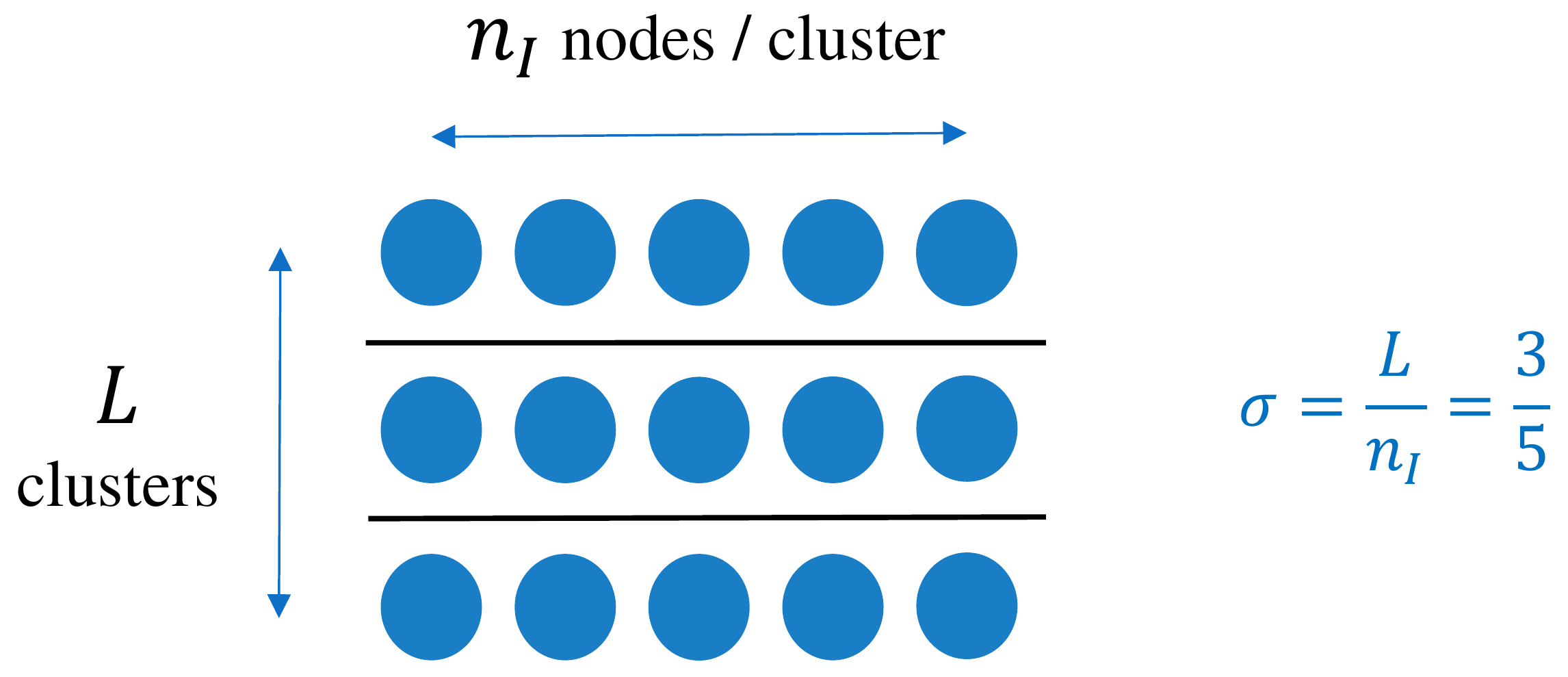}
	\caption{An example of DSS with dispersion ratio $\sigma = 5/3$, when the parameters are set to $L=3, n_I = 5, n=n_IL = 15$}
	\label{Fig:sigma}
\end{figure}

In Fig. \ref{Fig:capacity_versus_L_plot}, the capacity is a monotonic decreasing function of $L$ when the system suffers from an over-subscription problem. However, in general $(n,k,\alpha,\gamma)$ parameter settings, capacity is not always a monotonic decreasing function of $L$. Theorem \ref{Thm:cap_dec_ftn_L} illustrates the behavior of capacity as $L$ varies, focusing on the special case of 
$\gamma = \alpha.$
Before formally stating the next main result, we need to define  \begin{equation}\label{Eqn:sigma}
\sigma \coloneqq \frac{L}{n_I} = \frac{L^2}{n},
\end{equation}
the \textit{dispersion factor} of a clustered storage system, as illustrated in Fig. \ref{Fig:sigma}.
In the two-dimensional representation of a clustered distributed storage system, $L$ represents the number of rows (clusters), while $n_I = n/L$ represents the number of columns (nodes in each cluster). The dispersion factor $\sigma$ is the ratio of the number of rows to the number of columns. If $L$ increases for a fixed $n_I$, then $\sigma$ grows and the nodes become more dispersed into multiple clusters. 

Now we state our second main result, which is about the behavior of $\mathcal{C}$ versus $L$. 

\begin{theorem}\label{Thm:cap_dec_ftn_L}
\cmt{Consider the $\gamma = \alpha$ case when $\sigma, \gamma, R$ and $\xi$ are fixed.} In the asymptotic regime of large $n$, capacity $\mathcal{C}(\alpha, \beta_I, \beta_c)$ is asymptotically equivalent to   
\begin{equation}
\underline{C} = \frac{k}{2} \left(\gamma + \frac{n-k}{n(1-1/L)} \gamma_c \right), \label{Eqn:cap_lower}
\end{equation}
a monotonically decreasing function of $L$. 
This can also be stated as
\begin{equation}
\mathcal{C} \sim \underline{C}\label{Eqn:Cap_asymp_equiv}
\end{equation}
as $n \rightarrow \infty$ for a fixed $\sigma$.
\end{theorem}

\cmt{Note that under the setting of Theorem \ref{Thm:cap_dec_ftn_L}, we have}
\begin{align*}
n_I&= \frac{n}{L} = \frac{n}{\sqrt{n\sigma}} = \Theta(\sqrt{n}), \nonumber\\
k&= nR = \Theta(n), \quad
\alpha= \gamma = \text{constant}, \nonumber\\
\beta_I &= \frac{\gamma_I}{n_I-1} = \Theta(\frac{1}{\sqrt{n}}), \quad \beta_c = \frac{\gamma_c}{n-n_I} = \Theta(\frac{1}{n})
\end{align*}
\cmt{in the asymptotic regime of large $n$.}
The proof of Theorem \ref{Thm:cap_dec_ftn_L} is based on the following two lemmas. 

\begin{lemma}\label{Lemma:cap_upper_lower}
	\cmt{In the case of $\gamma = \alpha$}, capacity $ \mathcal{C}(\alpha,\beta_I,\beta_c)$ is upper/lower bounded as 
	\begin{equation} \label{Eqn:cap_upper_lower_bound}
	\underline{C} \leq \mathcal{C}(\alpha,\beta_I,\beta_c) \leq \underline{C} + \delta  
	\end{equation}
	where
	\begin{align}
	\delta &=  n_I^2(\beta_I-\beta_c)/8 \label{Eqn:cap_delta_val}
	\end{align}
	and $\underline{C}$ is defined in (\ref{Eqn:cap_lower}). 
\end{lemma}

\begin{lemma}\label{Lemma:Scale}
	\cmt{In the case of $\gamma = \alpha$}, $\underline{C}$ and $\delta$  defined in (\ref{Eqn:cap_lower}), (\ref{Eqn:cap_delta_val}) satisfy 
\begin{align}
\underline{C} &= \Theta(n), \nonumber\\
\delta &= O(n_I) \nonumber
\end{align}
\cmt{when $\gamma, R$ and $\xi$ are fixed.}
\end{lemma}
The proofs of these lemmas are in Appendix \ref{Section:Proofs_of_Lemmas}. 
Here we provide the proof of Theorem \ref{Thm:cap_dec_ftn_L} by using Lemmas \ref{Lemma:cap_upper_lower} and \ref{Lemma:Scale}. 
\begin{proof}[proof (of Theorem \ref{Thm:cap_dec_ftn_L})] 
From Lemma \ref{Lemma:Scale},
\begin{equation} \label{Eqn:delta_over_C}
\delta / \underline{C} = O(1/L) = O(\sqrt{1/n\sigma}) \rightarrow 0 
\end{equation}
as $n\rightarrow \infty$ for a fixed $\sigma$.
Moreover, dividing (\ref{Eqn:cap_upper_lower_bound}) by $\underline{C}$ results in 
\begin{equation}\label{Eqn:cap_bound_revisited}
1 \leq \mathcal{C}/\underline{C} \leq 1 + \delta/\underline{C}.
\end{equation} 
Putting (\ref{Eqn:delta_over_C}) into (\ref{Eqn:cap_bound_revisited}) completes the proof.
\end{proof}


%
%
%
%

%

\section{Discussion on Feasible ($\alpha, \beta_I, \beta_c$)}\label{Section:analysis on feasible points}

In the previous section, we obtained the capacity of the clustered DSS. 
This section analyzes the feasible ($\alpha, \beta_I, \beta_c$) points which satisfy $\mathcal{C}(\alpha, \beta_I, \beta_c) \geq \mathcal{M}$ for a given file size $\mathcal{M}$.
Using 
\begin{equation*}
\gamma = \gamma_I + \gamma_c = (n_I-1)\beta_I + (n-n_I)\beta_c
\end{equation*}
in (\ref{Eqn:gamma}), the behavior of the feasible set of $(\alpha, \gamma)$ points can be observed.
Essentially, the feasible points demonstrate a trade-off relationship. Two extreme points \textendash \ minimum storage regenerating (MSR) point and minimum bandwidth regenerating (MBR) point \textendash \ of the trade-off has been analyzed. Moreover, a special family of network codes which satisfy $\epsilon=0$, which we call the \textit{intra-cluster repairable codes}, is compared with the \textit{locally repairable codes} considered in \cite{papailiopoulos2014locally}. Finally, the set of feasible $(\alpha, \beta_c)$ points is discussed, when the system allows maximum $\beta_I$.

%
%

\subsection{Set of Feasible $(\alpha, \gamma)$ Points}\label{Section:nonzero_gammac}

We provide a closed-form solution for the set of feasible $(\alpha, \gamma)$ points which enable reliable storage of data $\mathcal{M}$. Based on the range of $\epsilon$ defined in (\ref{Eqn:epsilon}), the set of feasible points show different behaviors as stated in Corollary \ref{Corollary:Feasible Points_large_epsilon}. 




\begin{corollary}\label{Corollary:Feasible Points_large_epsilon}
	Consider a clustered DSS 
	for storing data $\mathcal{M}$, when $\epsilon = \beta_c/\beta_I$ satisfies $0 \leq  \epsilon \leq 1$. 
	For any $\gamma \geq \gamma^*(\alpha)$, the data $\mathcal{M}$ can be reliably stored, i.e., $\mathcal{C} \geq \mathcal{M}$, while it is impossible to reliably store data $\mathcal{M}$ when $\gamma < \gamma^*(\alpha)$. The threshold function $\gamma^*(\alpha)$ can be obtained as:
	
	\begin{enumerate}
		\item \text{ if }$\frac{1}{n-k} \leq \epsilon \leq 1$ \begin{equation} \label{Eqn:Feasible Points Result}
		\gamma^*(\alpha) = 
		\begin{cases}
		\infty, & \ \alpha \in (0,\frac{M}{k})  \\
		\frac{M-t\alpha}{s_t}, &  
		\ \alpha \in [\frac{M}{t+1+s_{t+1}y_{t+1}} ,\frac{M}{t+s_{t}y_{t}} )
		,\\
		& \ \ \ \ \ \ \ \ \ \ (t = k-1, k-2, \cdots, 1)\\
		\frac{M}{s_0 }, & \ \alpha \in [\frac{M}{s_0}, \infty) 
		\end{cases}
		\end{equation}

		\item  \text{Otherwise (if } $0 \leq \epsilon < \frac{1}{n-k}$)			
		\begin{equation} \label{Eqn:Feasible Points Result_intermediate_epsilon}
		\gamma^*(\alpha) = 
		\begin{cases}
		\infty, & \ \alpha \in (0,\frac{M}{\tau + \sum_{i=\tau+1}^{k}z_i})  \\
		\frac{\mathcal{M} - \tau \alpha}{s_{\tau}}, & \ \alpha \in [\frac{M}{\tau + \sum_{i=\tau+1}^{k}z_i},\frac{\mathcal{M}}{\tau + s_{\tau}y_{\tau}})  \\
		\frac{M-t\alpha}{s_t}, &  
		\ \alpha \in [\frac{M}{t+1+s_{t+1}y_{t+1}} ,\frac{M}{t+s_{t}y_{t}} )
		,\\
		& \ \ \ \ \ \ \ \ \ \ (t = \tau-1, \tau-2, \cdots, 1)\\
		\frac{M}{s_0 }, & \ \alpha \in [\frac{M}{s_0}, \infty) 
		\end{cases}
		\end{equation}
	\end{enumerate}
	where 
	\begin{align}
	\tau & = \max \{ t \in \{0,1,\cdots, k-1\} : z_t \geq 1 \}\label{Eqn:tau}\\
	s_t &=
	\begin{cases}
	 \frac{\sum_{i=t+1}^{k} z_i}{(n_I-1) + \epsilon (n-n_I)}, &  t =0, 1, \cdots, k-1 \label{Eqn:s_t}\\
	0, & t=k
	\end{cases}\\
	y_t &= \frac{(n_I-1)+\epsilon (n-n_I)}{z_t}\label{Eqn:y_t},\\
	z_t &= 
	\begin{cases}
	(n-n_I-t+h_t)\epsilon + (n_I-h_t), & t \in [k] \\
	\infty, & t=0 \label{Eqn:z_t}
	\end{cases}\\
	h_t &= \min \{s \in [n_I]: \sum_{l=1}^s g_l \geq t \}, \label{Eqn:h_t}
	\end{align}
	and $\{g_l\}_{l=1}^{n_I}$ is defined in (\ref{Eqn:g_m}).
\end{corollary}

\begin{figure}[t]
	\centering
	\includegraphics[height=55mm]{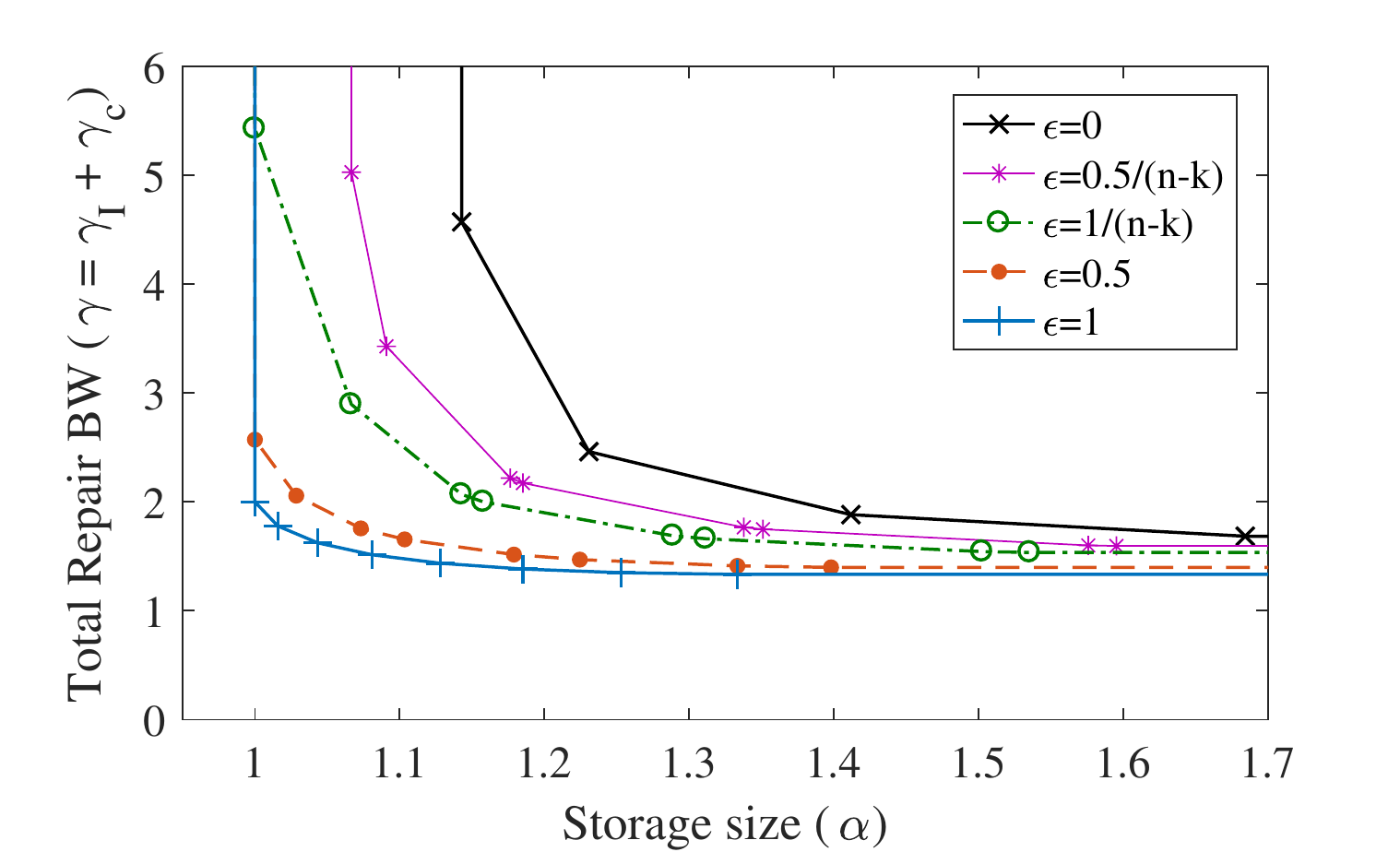}
	\caption{Optimal tradeoff between node storage size $\alpha$ and total repair bandwidth $\gamma$, under the setting of $n=15, k=8, L=3, \mathcal{M} = 8$}
	\label{Fig:kappa_various_latex}
\end{figure}

\begin{figure}[t]
	\centering
	\includegraphics[height=50mm]{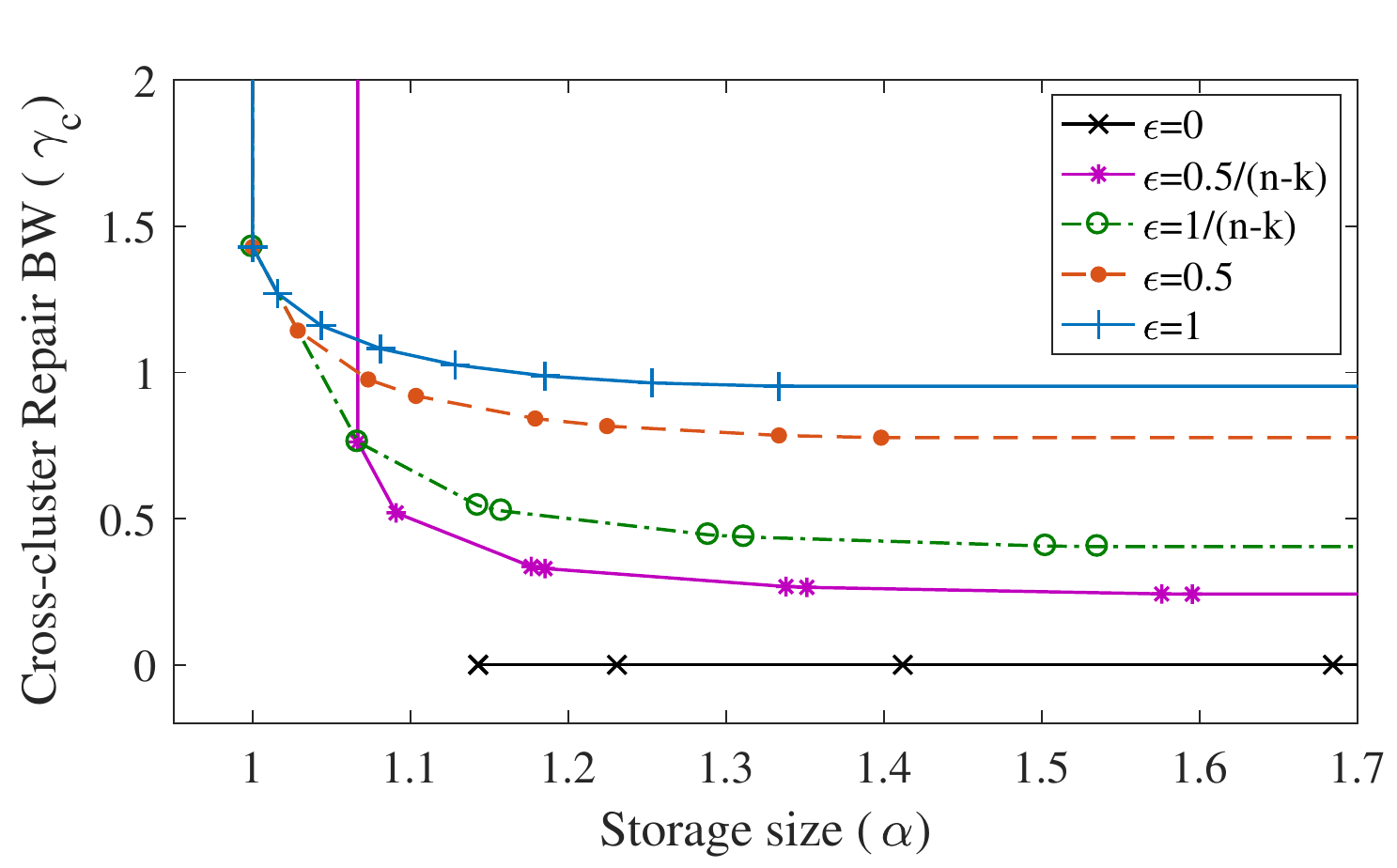}
	\caption{Optimal tradeoff between node storage size $\alpha$ and cross-rack repair bandwidth $\gamma_c$, under the setting of $n=15, k=8, L=3, \mathcal{M} = 8$}
	\label{Fig:gamma_c_various_latex}
\end{figure}

An example for the trade-off results of Corollary \ref{Corollary:Feasible Points_large_epsilon} is illustrated in Fig. \ref{Fig:kappa_various_latex}, for various $0 \leq \epsilon \leq 1$ values.
Here, the $\epsilon = 1$ (i.e., $\beta_I = \beta_c$) case corresponds to the symmetric repair in the non-clustered scenario \cite{dimakis2010network}.
The plot for $\epsilon = 0$ (or $\beta_c = \gamma_c = 0$) shows that the cross-cluster repair bandwidth can be reduced to zero with extra resources ($\alpha$ or $\gamma_I$), where the amount of required resources are specified in Corollary \ref{Corollary:Feasible Points_large_epsilon}.
\cmt{Note that in the case of $\epsilon=0$, the storage system is \textit{completely localized}, i.e., nodes can be repaired exclusively from their cluster.}
From Fig. \ref{Fig:kappa_various_latex}, we can confirm that as $\epsilon$ decreases, extra resources ($\gamma$ or $\alpha$) are required to reliably store the given data $\mathcal{M}$. 
Moreover, Corollary \ref{Corollary:Feasible Points_large_epsilon} suggests a mathematically interesting result, stated in the following Theorem, the proof of which is in Appendix \ref{Section:proof_of_thm_condition_for_min_storage}.

\begin{theorem} \label{Thm:condition_for_min_storage}
	A clustered DSS can reliably store file $\mathcal{M}$ with the minimum storage overhead $\alpha = \mathcal{M}/k$ if and only if 
	\begin{equation}
	\epsilon \geq \frac{1}{n-k}.
	\end{equation}
\end{theorem}

Note that $\alpha = \mathcal{M}/k$ is the minimum storage overhead which can satisfy the MDS property, as stated in \cite{dimakis2010network}. 
The implication of Theorem \ref{Thm:condition_for_min_storage} is shown in Fig. \ref{Fig:kappa_various_latex}. Under the $\mathcal{M} = k = 8$ setting, data $\mathcal{M}$ can be reliably stored with minimum storage overhead $\alpha=\mathcal{M}/k=1$ for $\frac{1}{n-k} \leq \epsilon \leq 1$, while it is impossible to achieve minimum storage overhead $\alpha=\mathcal{M}/k=1$ for $0 \leq \epsilon < \frac{1}{n-k}$.
Finally, since reducing the cross-cluster repair burden $\gamma_c$ is regarded as a much harder problem compared to reducing the intra-cluster repair burden $\gamma_I$, we also plotted feasible $(\alpha, \gamma_c)$ pairs for various $\epsilon$ values in Fig. \ref{Fig:gamma_c_various_latex}.
The plot for $\epsilon = 0$ obviously has zero $\gamma_c$, while $\epsilon > 0$ cases show a trade-off relationship. As $\epsilon$ increases, the minimum $\gamma_c$ value increases gradually.

\subsection{Minimum-Bandwidth-Regenerating (MBR) point and Minimum-Storage-Regenerating (MSR) point}\label{Section:MSR_MBR}

\begin{figure}[t]
	\centering
	\includegraphics[height=50mm]{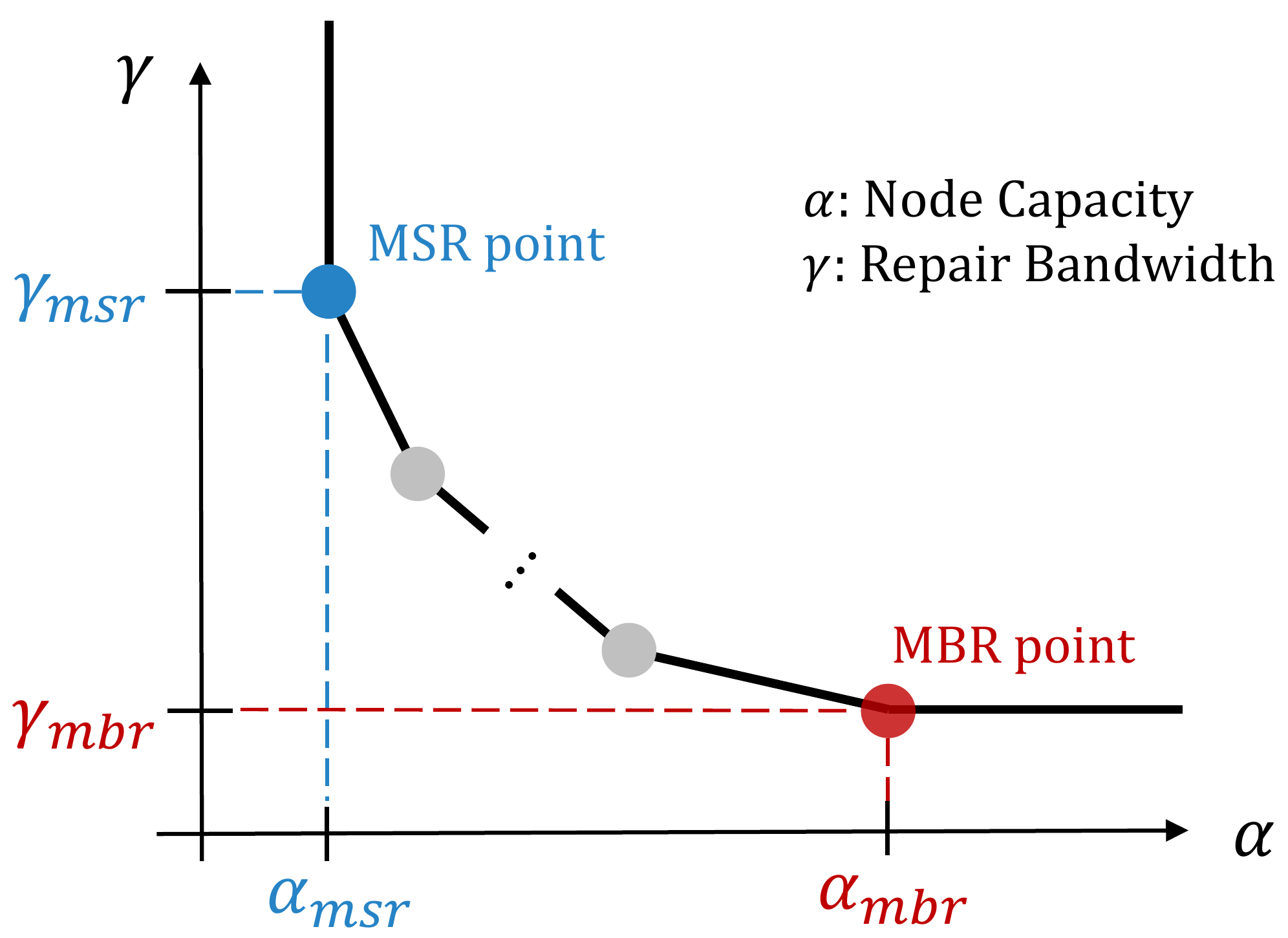}
	\caption{Set of Feasible $(\alpha, \gamma)$ Points}
	\label{Fig:MBR_MSR_points}
\end{figure}

According to Corollary \ref{Corollary:Feasible Points_large_epsilon}, the set of feasible $(\alpha, \gamma)$ points shows a trade-off curve as in Fig. \ref{Fig:MBR_MSR_points}, for arbitrary $n,k,L,\epsilon$ settings.
Here we focus on two extremal points: the minimum-bandwidth-regenerating (MBR) point and the minimum-storage-regenerating (MSR) point. As originally defined in \cite{dimakis2010network}, we call the point on the trade-off with minimum bandwidth $\gamma$ as MBR. Similarly, we call the point with minimum storage $\alpha$ as MSR\footnote{\cmt{Among multiple points with minimum $\gamma$, the point having the smallest $\alpha$ is called the MBR point. Similarly, among points with minimum $\alpha$, the point with the minimum $\gamma$ is called the MSR point}.}.
Let $(\alpha_{msr}^{(\epsilon)}, \gamma_{msr}^{(\epsilon)})$ be the $(\alpha, \gamma)$ pair of the MSR point for given $\epsilon$. Similarly, define $(\alpha_{mbr}^{(\epsilon)}, \gamma_{mbr}^{(\epsilon)})$ as the parameter pair for MBR points. 
 According to Corollary \ref{Corollary:Feasible Points_large_epsilon}, the explicit $(\alpha, \gamma)$ expression for the MSR and MBR points are as in the following Corollary, the proof of which is given in Appendix \ref{Section:Proof of Corollary_msr_mbr_points}.

\begin{corollary}\label{Coro:msr_mbr_points}
For a given $0 \leq \epsilon \leq 1$, we have
\begin{align}
&(\alpha_{msr}^{(\epsilon)}, \gamma_{msr}^{(\epsilon)}) \nonumber\\
&=
\begin{cases}
(\frac{\mathcal{M}}{\tau + \sum_{i=\tau+1}^{k}z_i},  \frac{\mathcal{M}}{\tau + \sum_{i=\tau+1}^{k}z_i} \frac{\sum_{i=\tau+1}^{k}z_i}{s_{\tau}}
), & 0 \leq \epsilon < \frac{1}{n-k} \\
(\frac{\mathcal{M}}{k}, \frac{\mathcal{M}}{k}\frac{1}{s_{k-1}}), & \frac{1}{n-k} \leq \epsilon \leq 1
\end{cases} \label{Eqn:MSR_point}
\end{align}
\begin{align}
(\alpha_{mbr}^{(\epsilon)}, \gamma_{mbr}^{(\epsilon)}) &=
(\frac{\mathcal{M}}{s_0}, \frac{\mathcal{M}}{s_0}).
\label{Eqn:MBR_point}
\end{align}
\end{corollary}


Now we compare the MSR and MBR points for two extreme cases of $\epsilon = 0$ and $\epsilon = 1$. Using $R\coloneqq k/n$ and the dispersion ratio $\sigma$ defined in (\ref{Eqn:sigma}), the asymptotic behaviors of MBR and MSR points are illustrated in the following theorem, the proof of which is in Appendix \ref{Proof_of_Thm_MSR_MBR}.

\begin{theorem}\label{Thm:MSR_MBR}
	Consider the MSR point $(\alpha_{msr}^{(\epsilon)}, \gamma_{msr}^{(\epsilon)})$ and the MBR point $(\alpha_{mbr}^{(\epsilon)}, \gamma_{mbr}^{(\epsilon)})$ for $\epsilon = 0, 1$. The minimum node storage for $\epsilon=0$ is asymptotically equivalent to the minimum node storage for $\epsilon=1$, i.e., 
	\begin{equation}\label{Eqn:msr_property}
	\alpha_{msr}^{(0)} \sim  \alpha_{msr}^{(1)} = \mathcal{M}/k
	\end{equation}
	as $n \rightarrow \infty$ for arbitrary fixed $\sigma$ and $R=k/n$.
	Moreover, the MBR point for $\epsilon=0$ approaches the MBR point for $\epsilon=1$, i.e.,
	\begin{equation}\label{Eqn:mbr_property}
	(\alpha_{mbr}^{(0)}, \gamma_{mbr}^{(0)}) \rightarrow  (\alpha_{mbr}^{(1)}, \gamma_{mbr}^{(1)}),
	\end{equation}
	as $R=k/n \rightarrow 1$. The ratio between $\gamma_{mbr}^{(0)}$ and $\gamma_{mbr}^{(1)}$ is expressed as
	\begin{equation}\label{Eqn:mbr_ratio}
\frac{\gamma_{mbr}^{(0)}}{\gamma_{mbr}^{(1)}} \leq 2- \frac{k-1}{n-1} 
	\end{equation}
\end{theorem}

\begin{figure}[t]
	\centering
	\includegraphics[height=50mm]{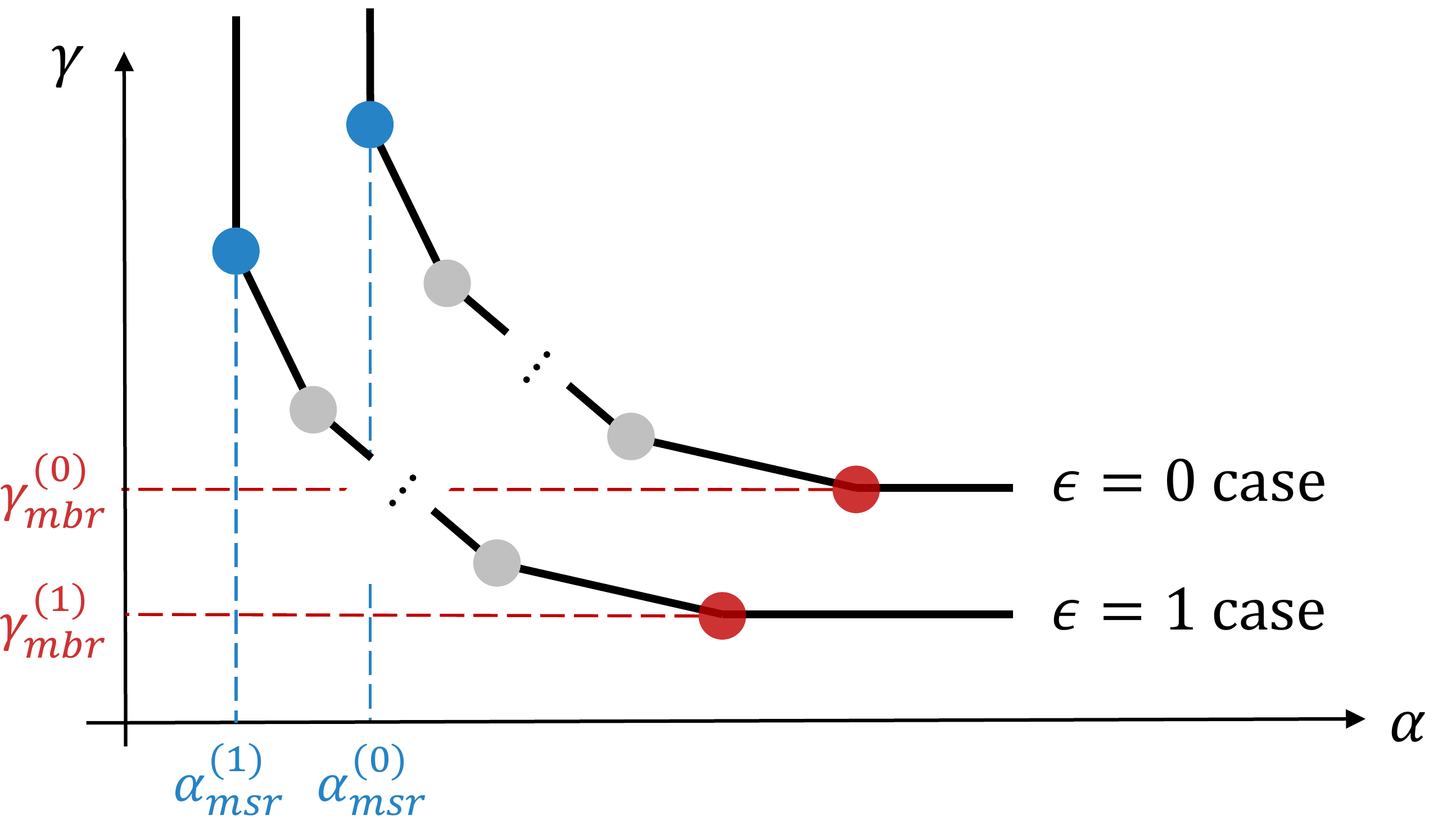}
	\caption{Optimal trade-off curves for $\epsilon = 0, 1$}
	\label{Fig:tradeoff_min}
\end{figure}

\cmt{Note that under the setting of Theorem \ref{Thm:MSR_MBR}, we have}
\begin{align*}
\alpha_{msr}^{(1)}&= \text{constant}, \quad k= nR = \Theta(n), \quad \mathcal{M}= \Theta(n)
\end{align*}
\cmt{in the asymptotic regime of large $n$.}
According to Theorem \ref{Thm:MSR_MBR}, the minimum storage for $\epsilon=0$ can achieve $\mathcal{M}/k$ as $n \rightarrow \infty$ with fixed $R=k/n$. This result coincides with the result of Theorem \ref{Thm:condition_for_min_storage}. According to Theorem \ref{Thm:condition_for_min_storage}, the sufficient and necessary condition for achieving the minimum storage of $\alpha=\mathcal{M}/k$ is
\begin{equation}
\epsilon \geq \frac{1}{n-k} = \frac{1}{n(1-R)}.
\end{equation}
As $n$ increases with a fixed $R$, the lower bound on $\epsilon$ reduces, so that in the asymptotic regime, $\epsilon=0$ can achieve $\alpha = \mathcal{M}/k$.

\begin{figure}
	\centering
	\subfloat[][$\epsilon = 1$]{\includegraphics[width=80mm ]{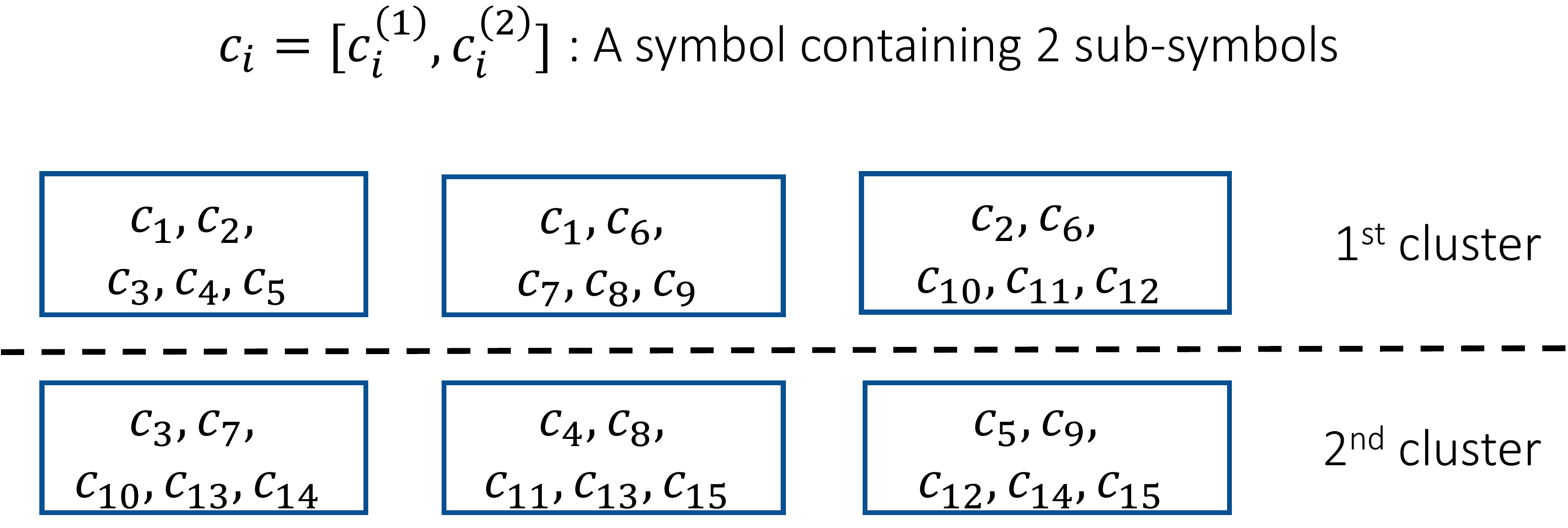}\label{Fig:mbr_epsilon_1}}
	\quad \quad
	\vspace{5mm}
	\subfloat[][$\epsilon = 0$]{\includegraphics[width=80mm]{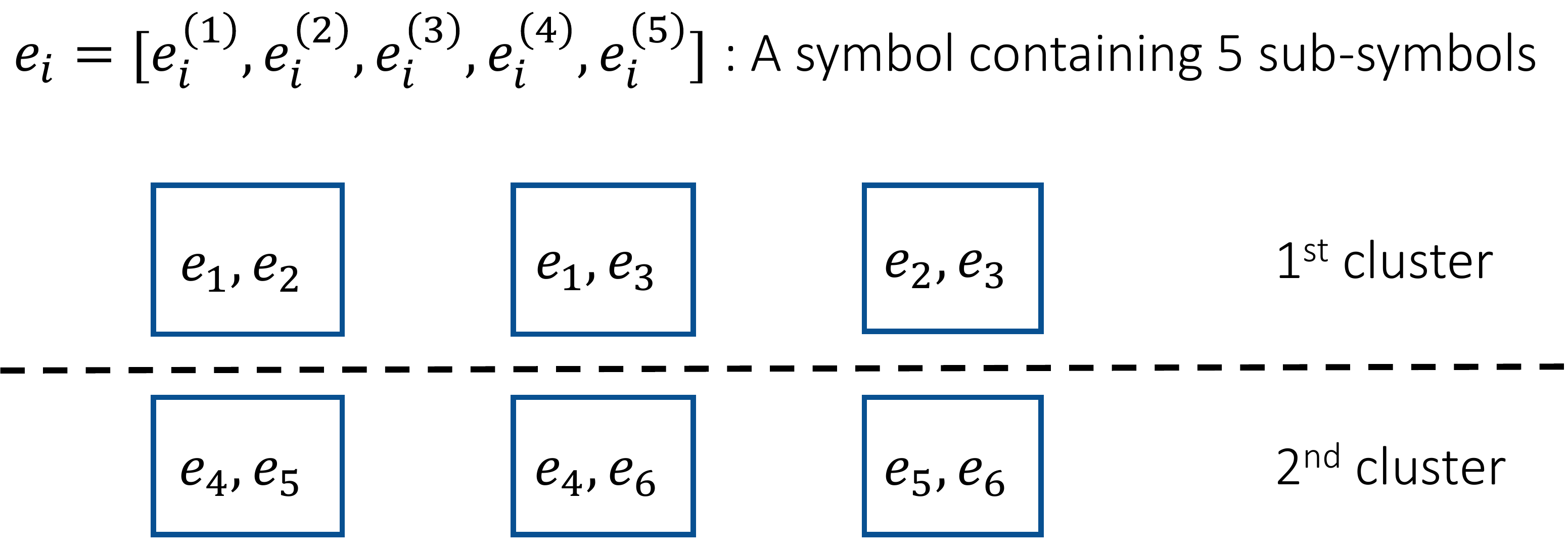}\label{Fig:mbr_epsilon_0}}
	\caption{MBR coding schemes for $n=6, k=5, L=2, \alpha = 10$ [sub-symbol], $\gamma = 10$  [sub-symbol]}
	\label{Fig:MBR_coding_schemes}
\end{figure}

Moreover, Theorem \ref{Thm:MSR_MBR} states that the MBR point for $\epsilon=0$ approaches the MBR point for $\epsilon=1$ as $R=k/n$ goes to 1. Fig. \ref{Fig:MBR_coding_schemes} provides two MBR coding schemes with $(n,k,L) = (6,5,2)$, which has different $\epsilon$ values; one coding scheme in Fig. \ref{Fig:mbr_epsilon_1} satisfies $\epsilon = 1$, while the other in Fig. \ref{Fig:mbr_epsilon_0} satisfies $\epsilon=0$. 
The RSKR coding scheme \cite{rashmi2009explicit} is applied to the six nodes in Fig. \ref{Fig:mbr_epsilon_1}.
Each node (illustrated as a rectangular box) contains five $c_i$ symbols, where each symbol $c_i$ consists of two sub-symbols, $c_i^{(1)}$ and $c_i^{(2)}$. Note that any symbol $c_i$ is shared by exactly two nodes in Fig. \ref{Fig:mbr_epsilon_1}, which is due to the property of RSKR coding. 
This system can reliably store fifteen symbols $\{c_i\}_{i=1}^{15}$, or $\mathcal{M} = 30$ sub-symbols $\{c_i^{(1)}, c_i^{(2)}\}_{i=1}^{15}$, since it satisfies two properties \textendash \ the exact repair property and the data recovery property \textendash \ as illustrated below. 
First, when a node fails, five other nodes transmit five symbols (one distinct symbol by each node), which exactly regenerates the failed node. 
Second, we can retrieve data, $\mathcal{M} = 30$ sub-symbols $\{c_i^{(1)}, c_i^{(2)}\}_{i=1}^{15}$, by contacting any $k=5$ nodes. 
In Fig. \ref{Fig:mbr_epsilon_0}, each node contains two $e_i$ symbols, where each symbol $e_i$ consists of five sub-symbols, $\{e_i^{(j)}\}_{j=1}^{5}$. Note that in Fig. \ref{Fig:mbr_epsilon_0}, any symbol $e_i$ is shared by exactly two nodes which reside in the same cluster. This is because we applied RSKR coding at each cluster in the system of Fig. \ref{Fig:mbr_epsilon_0}. This system can reliably store six symbols $\{e_i\}_{i=1}^{6}$, or $\mathcal{M} = 30$ sub-symbols $\bigcup_{i \in [6], j\in [5]} \{e_i^{(j)}\}$, 
since it satisfies the exact repair property and the data recovery property.

Note that both DSSs in Fig. \ref{Fig:MBR_coding_schemes} reliably store $\mathcal{M} = 30$ sub-symbols, by using the node capacity of $\alpha = 10$ sub-symbols and the repair bandwidth of $\gamma = 10$ sub-symbols. 
However, the former system requires $\gamma_c = 6$ cross-cluster repair bandwidth for each node failure event, while the latter system requires $\gamma_c = 0$ cross-cluster repair bandwidth. For example, if the leftmost node of the $1^{st}$ cluster fails in Fig. \ref{Fig:mbr_epsilon_1}, then four sub-symbols $\{c_1^{(i)}, c_2^{(i)}\}_{i=1}^2$ are transmitted within that cluster, while six sub-symbols $\{c_3^{(i)}, c_4^{(i)},  c_5^{(i)}\}_{i=1}^2$ are transmitted from the $2^{nd}$ cluster. In the case of $\epsilon=0$ in Fig. \ref{Fig:mbr_epsilon_0}, ten sub-symbols $\{e_1^{(i)}, e_2^{(i)}\}_{i=1}^5$ are transmitted within the $1^{st}$ cluster and no sub-symbols are transmitted across the clusters, when the leftmost node of the $1^{st}$ cluster fails. Thus, transition from the former system (Fig. \ref{Fig:mbr_epsilon_1}) to the latter system (Fig. \ref{Fig:mbr_epsilon_0}) reduces the cross-cluster repair bandwidth to zero, while maintaining the storage capacity $\mathcal{M}$ and the required resource pair $(\alpha, \gamma)$.
Likewise, we can reduce the cross-cluster repair bandwidth to \textit{zero} while maintaining the storage capacity, in the case of $R=k/n \rightarrow 1$.

Note that $\gamma_{mbr}^{(\epsilon)} = \alpha_{mbr}^{(\epsilon)}$ for $0 \leq \epsilon \leq 1$, from (\ref{Eqn:MBR_point}). Thus, the result of (\ref{Eqn:mbr_ratio}) in Theorem \ref{Thm:MSR_MBR} can be expressed as Fig. \ref{Fig:MBR_relationship} at the asymptotic regime of large $n,k$. 
According to Fig. \ref{Fig:MBR_relationship}, intra-cluster only repair ($\epsilon = 0$ or $\beta_c = 0$) is possible by using additional resources ($\alpha$ and $\gamma$) in the $1-R$ portion, compared to the symmetric repair ($\epsilon =1$) case.

\begin{figure}[t]
	\centering
	\includegraphics[height=50mm]{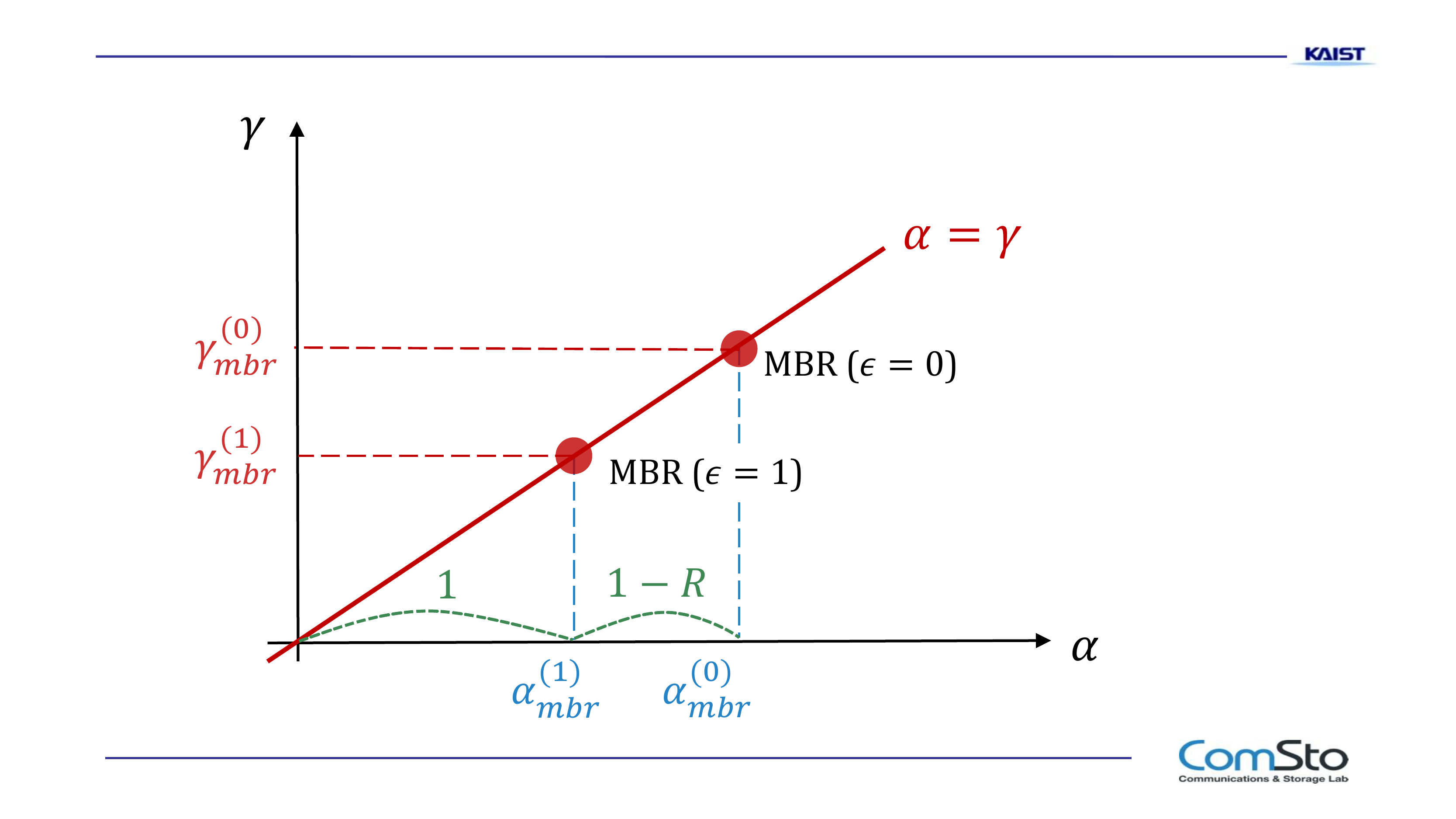}
	\caption{Relationship between MBR point of $\epsilon = 0$ and that of $\epsilon = 1$}
	\label{Fig:MBR_relationship}
\end{figure}

\subsection{Intra-cluster Repairable Codes versus Locally Repairable Codes \cite{papailiopoulos2014locally}}\label{Section:IRC_versus_LRC}

Here, we define a family of network coding schemes for clustered DSS, which we call the \textit{intra-cluster repairable codes}. 
In Corollary \ref{Corollary:Feasible Points_large_epsilon}, we considered DSSs with $\epsilon=0$, which can repair any failed node by using intra-cluster communication only. The optimal $(\alpha, \gamma)$ trade-off curve which satisfies $\epsilon=0$ is illustrated as the solid line with cross markers in Fig. \ref{Fig:kappa_various_latex}. Each point on the curve is achievable (i.e., there exists a  network coding scheme), according to the result of \cite{ahlswede2000network}. We call the network coding schemes for the points on the curve of $\epsilon=0$ the \textit{intra-cluster repairable codes}, since these coding schemes can repair any failed node by using intra-cluster communication only. 
The relationship between the intra-cluster repairable codes and the locally repairable codes (LRC) of \cite{papailiopoulos2014locally} are investigated in Theorem \ref{Thm:IRC_versus_LRC}, the proof of which is given in Appendix \ref{Section:Proof_of_IRC_versus_LRC}. 
\cmt{Note that according to the definition in \cite{papailiopoulos2014locally}, an $(n,l_0,m_0,\mathcal{M},\alpha)$-LRC encodes a file of size $\mathcal{M}$ into $n$ coded symbols, where each symbol contains $\alpha$ bits. In addition, any coded symbol of the LRC is regenerated by accessing at most $l_0$ other symbols (i.e., the code has repair locality of $l_0$), while the minimum distance of the code is $m_0$.}

\begin{theorem}\label{Thm:IRC_versus_LRC}
	The intra-cluster repairable codes with storage overhead $\alpha$ are the $(n,l_0,m_0,\mathcal{M},\alpha)$-LRC codes of \cite{papailiopoulos2014locally} where
	\begin{align*}
	l_0 &= n_I-1, \\
	m_0 &= n-k+1.
	\end{align*}
	It is confirmed that Theorem 1 of \cite{papailiopoulos2014locally},
	\begin{equation}\label{Eqn:locality_ineq}
	m_0 \leq n - \ceil[\bigg]{\frac{\mathcal{M}}{\alpha}} - \ceil[\bigg]{\frac{\mathcal{M}}{l_0\alpha}} + 2,
	\end{equation}
	holds for every intra-cluster repairable code with storage overhead $\alpha$. Moreover, the equality of (\ref{Eqn:locality_ineq}) holds if
	\begin{align}
	\alpha = \alpha_{msr}^{(0)},
	(k  \Mod{n_I}) \neq 0. \nonumber
	\end{align}
\end{theorem}

\subsection{Required $\beta_c$ for a given $\alpha$}

Here we focus on the following question: when the available intra-cluster repair bandwidth is abundant, how much cross-cluster repair bandwidth is required to reliably store file $\mathcal{M}$? We consider scenarios when the intra-cluster repair bandwidth (per node) has its maximum value, i.e., $\beta_I = \alpha$. Under this setting, Theorem \ref{Thm:beta_c_alpha_trade} specifies the minimum required $\beta_c$ which satisfies $\mathcal{C}(\alpha, \beta_I, \beta_c) \geq \mathcal{M}$. The proof of Theorem \ref{Thm:beta_c_alpha_trade} is given in Appendix \ref{Section:Proof_of_Thm_beta_c_alpha_trade}.


\begin{theorem}\label{Thm:beta_c_alpha_trade}
	Suppose the intra-cluster repair bandwidth is set at the maximum value, i.e., $\beta_I = \alpha$. For a given node capacity $\alpha$, the clustered DSS can reliably store data $\mathcal{M}$ if and only if $\beta_c \geq \beta_c^*$ where
	\begin{align}
	\beta_c^* &= 
	\begin{cases}
	\frac{\mathcal{M} - (k-1)\alpha}{n-k}, & \text{ if } \alpha \in [\frac{\mathcal{M}}{k}, \frac{\mathcal{M}}{f_{k-1}}  ) \\
	\frac{\mathcal{M} - m\alpha}{\sum_{i=m+1}^{k} (n-i)}, & \text{ if } \alpha \in [\frac{\mathcal{M}}{f_{m+1}}, \frac{\mathcal{M}}{f_m}  ) \\
	& \quad \quad (m=k-2, k-3, \cdots, s+1) \\
	\frac{\mathcal{M} - k_0 \alpha}{\sum_{i=k_0+1}^{k} (n-i)}, & \text{ if } \alpha \in [\frac{\mathcal{M}}{f_{k_0+1}}, \frac{\mathcal{M}}{k_0} ) \\
	0, & \text{ if } \alpha \in [ \frac{\mathcal{M}}{k_0}, \infty ),
	\end{cases} \label{Eqn:beta_c_star}\\
	f_m &= m + \frac{\sum_{i=m+1}^{k} (n-i)}{n-m}, \label{Eqn:f_m}\\
	k_0 &= k-\floor{\frac{k}{n_I}} \label{Eqn:k_0}.
	\end{align}
\end{theorem}

\begin{figure}[!t]
	\centering
	\includegraphics[height=50mm]{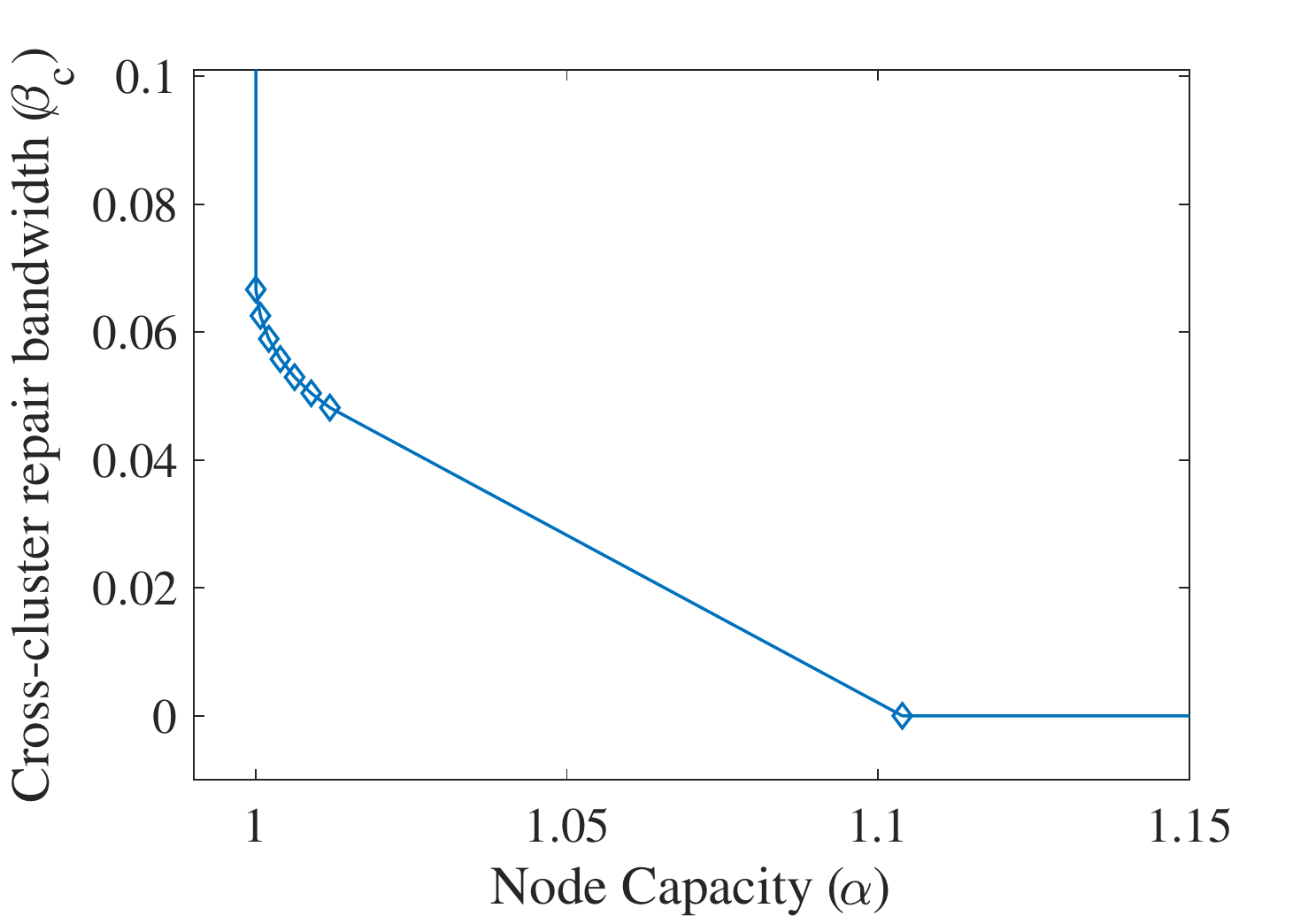}
	\caption{Optimal tradeoff between cross-cluster repair bandwidth and node capacity, when $n=100, k=85, L=10, \mathcal{M} = 85$ and $\beta_I = \alpha$}
	\label{Fig:beta_c_alpha}
\end{figure}

Fig. \ref{Fig:beta_c_alpha} provides an example of the optimal trade-off relationship between $\beta_c$ and $\alpha$, explained in Theorem \ref{Thm:beta_c_alpha_trade}. For $\alpha \geq \mathcal{M}/k_0 = 1.104$, the cross-cluster burden $\beta_c$ can be reduced to zero. 
However, as $\alpha$ decreases from $\mathcal{M}/k_0$, the system requires a larger $\beta_c$ value. For example, if $\alpha = \beta_I= 1.05$ in Fig. \ref{Fig:beta_c_alpha}, $\beta_c \geq 0.03$ is required to satisfy $\mathcal{C}(\alpha, \beta_I, \beta_c) \geq \mathcal{M} = 85$. Thus, for each node failure event, a cross-cluster repair bandwidth of $\gamma_c = (n-n_I)\beta_c \geq 2.7$ is required. Theorem \ref{Thm:beta_c_alpha_trade} provides an explicit equation for the cross-cluster repair bandwidth we need to pay, in order to reduce node capacity $\alpha$.

\section{Further Comments \& Future Works}\label{Section:Future_Works}

%

\subsection{Explicit coding schemes for clustered DSS}

According to the part I proof of Theorem \ref{Thm:Capacity of clustered DSS}, there exists an information flow graph $G^*$ which has the min-cut value of $\mathcal{C}$, the capacity of clustered DSS. Thus, according to \cite{ahlswede2000network}, there exists a linear network coding scheme which achieves capacity $\mathcal{C}$. Although the existence of a coding scheme is verified, explicit network coding schemes which achieve capacity need to be specified for implementing practical systems. 
\cmt{Recently, under the setting of clustered DSS modeled in the present paper, MBR codes for all $n,k,L,\epsilon$ are constructed in \cite{sohn2018explicit} and MSR codes for limited parameters are designed in \cite{sohn2018class}. Explicit code construction for general parameters and/or construction of codes that requires small field size are interesting remaining issues.}

\subsection{Optimal number of clusters}

According to Theorem \ref{Thm:cap_dec_ftn_L}, capacity $\mathcal{C}$ is asymptotically a monotonically decreasing function of $L$, the number of clusters. Thus, reducing the number of clusters (i.e., gathering storage nodes into a smaller number of clusters) increases storage capacity. However, as mentioned in Section \ref{Section:C_versus_L}, we typically want to have a sufficiently large $L$, to tolerate the failure of a cluster. Then, the remaining question is in finding optimal $L^*$ which not only allows sufficiently large storage capacity, but also a tolerance to cluster failures. We regard this problem as a future research topic, the solution to which will provide a guidance on the strategy for distributing storage nodes into multiple clusters.

\subsection{Extension to general $d_I, d_c$ settings}

The present paper assumed a maximum helper node setting, $d_I = n_I-1$ and $d_c = n-n_I$, since it maximizes the capacity as stated in Proposition \ref{Prop:max_helper_nodes}. However, 
waiting for all helper nodes gives rise to a latency issue. If we reduce the number of helper nodes $d_I$ and $d_c$, low latency repair would be possible, while the achievable storage capacity decreases. 
Thus, we consider obtaining the capacity expression for general $d_I, d_c$ settings, and discover the trade-off between capacity and latency for various $d_I, d_c$ values.

\subsection{Scenarios of aggregating the helper data within each cluster}

\cmt{In recent work [28], [38] on multi-rack (multi-cluster) distributed storage, the authors discuss aggregation of repair data leaving a given cluster. The ideas is to allow aggregation and compression of all helper data leaving each cluster to aid reconstruction taking place in some other cluster containing a failed node. This type of repair link aggregation has been shown to reduce the cross-cluster repair burden in [28], [38]. We expect that the same method would also change the tradeoff picture for our distributed cluster model. This is certainly an interesting and important topic to investigate, but careful analysis including the effect of security breach on links will provide a more complete assessment of the merits and potential perils of repair link aggregation. We will leave this as a future endeavor.}


\section{Conclusion}\label{Section:Conclusion}
This paper considered a practical distributed storage system where storage nodes are dispersed  into several clusters. Noticing that the traffic burdens of intra- and cross-cluster communications are typically different, a new system model for clustered distributed storage systems is suggested. Based on the cut-set bound analysis of information flow graph, the storage capacity $\mathcal{C}(\alpha, \beta_I, \beta_c)$ of the suggested model is obtained in a closed-form, as a function of three main resources: node storage capacity $\alpha$, intra-cluster repair bandwidth $\beta_I$ and cross-cluster repair bandwidth $\beta_c$. It is shown that the asymmetric repair ($\beta_I > \beta_c$) degrades capacity, which is the cost for lifting the cross-cluster repair burden. 
Moreover, in the asymptotic regime of a large number of storage nodes,
capacity is shown to be asymptotically equivalent to a monotonic decreasing function of $L$, the number of clusters. Thus, reducing $L$ (i.e., gathering nodes into less clusters) is beneficial for increasing capacity, although we would typically need to guarantee sufficiently large $L$ to tolerate rack failure events. 

Using the capacity expression, we obtained the feasible set of ($\alpha, \beta_I, \beta_c$) triplet which satisfies $\mathcal{C}(\alpha, \beta_I, \beta_c) \geq \mathcal{M}$, i.e., it is possible to reliably store file $\mathcal{M}$ by using the resource value set ($\alpha, \beta_I, \beta_c$). The closed-form solution on the feasible set shows a different behavior depending on $\epsilon = \beta_c/\beta_I$, the ratio of cross- to intra-cluster repair bandwidth. 
It is shown that the minimum storage of $\alpha = \mathcal{M}/k$ is achievable if and only if $\epsilon \geq \frac{1}{n-k}$. Moreover, in the special case of $\epsilon = 0$, we can construct a reliable storage system without using cross-cluster repair bandwidth. A family of network codes which enable $\epsilon = 0$, called the \textit{intra-cluster repairable codes}, has been shown to be a class of the \textit{locally repairable codes} defined in \cite{papailiopoulos2014locally}.

%


%

\appendices
\numberwithin{equation}{section}

\section{Proof of Theorem \ref{Thm:Capacity of clustered DSS}}\label{Section:Proof of Thm 1}
Here, we prove Theorem \ref{Thm:Capacity of clustered DSS}. First, denote the right-hand-side (RHS) of (\ref{Eqn:Capacity of clustered DSS_rev}) as
\begin{equation}\label{Eqn:Capacity_value}
T \coloneqq  \sum_{i=1}^{n_I} \sum_{j=1}^{g_i} \min \{\alpha, \rho_i\beta_I + (n-\rho_i - j - \sum_{m=1}^{i-1}g_m)
 \beta_c \}.
\end{equation}
For other notations used in this proof, refer to subsection \ref{Subsection: notation}.
 The proof proceeds in two parts.

\vspace{5mm}
\textit{Part I}. Show an information flow graph $G^* \in \mathcal{G}$ and a cut-set $c^* \in C(G^*)$ such that $w(G^*, c^*) = T$:

\begin{figure*}
	\centering
	\includegraphics[width=\textwidth-45mm]{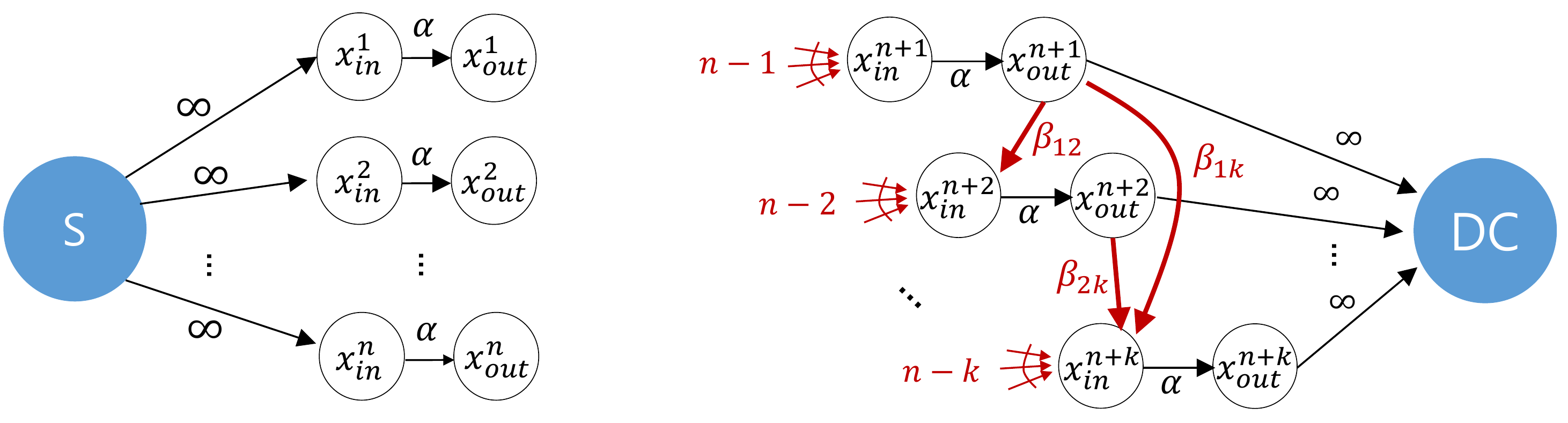}
	\caption{The information flow graph $G^*$ for the part I proof of Theorem 1}
	\label{Fig:partI_proof_Thm1}
\end{figure*}

Consider the information flow graph $G^*$ illustrated in Fig. \ref{Fig:partI_proof_Thm1}, which is obtained by the following procedure. 
First, data from source node $S$ is distributed into $n$ nodes labeled from $x^1$ to $x^n$. As mentioned in Section \ref{Section:Info_flow_graph}, the storage node $x^i = (x^i_{in}, x^i_{out})$ consists of an input-node $x_{in}^i$ and an output-node $x_{out}^i$. Second, storage node $x^t$ fails and is regenerated at the newcomer node $x^{n+t}$ for $t \in [k]$. The newcomer node $x^{n+t}$ connects to $n-1$ survived nodes $\{x^{m}\}_{m=t+1}^{n+t-1}$ to regenerate $x^t$. 
Third, data collector node $DC$ contacts $\{x^{n+t}\}_{t=1}^{k}$ to retrieve data. This whole process is illustrated in the information flow graph $G^*$. 

\begin{figure}[!t]
	\centering
	\includegraphics[height=45mm]{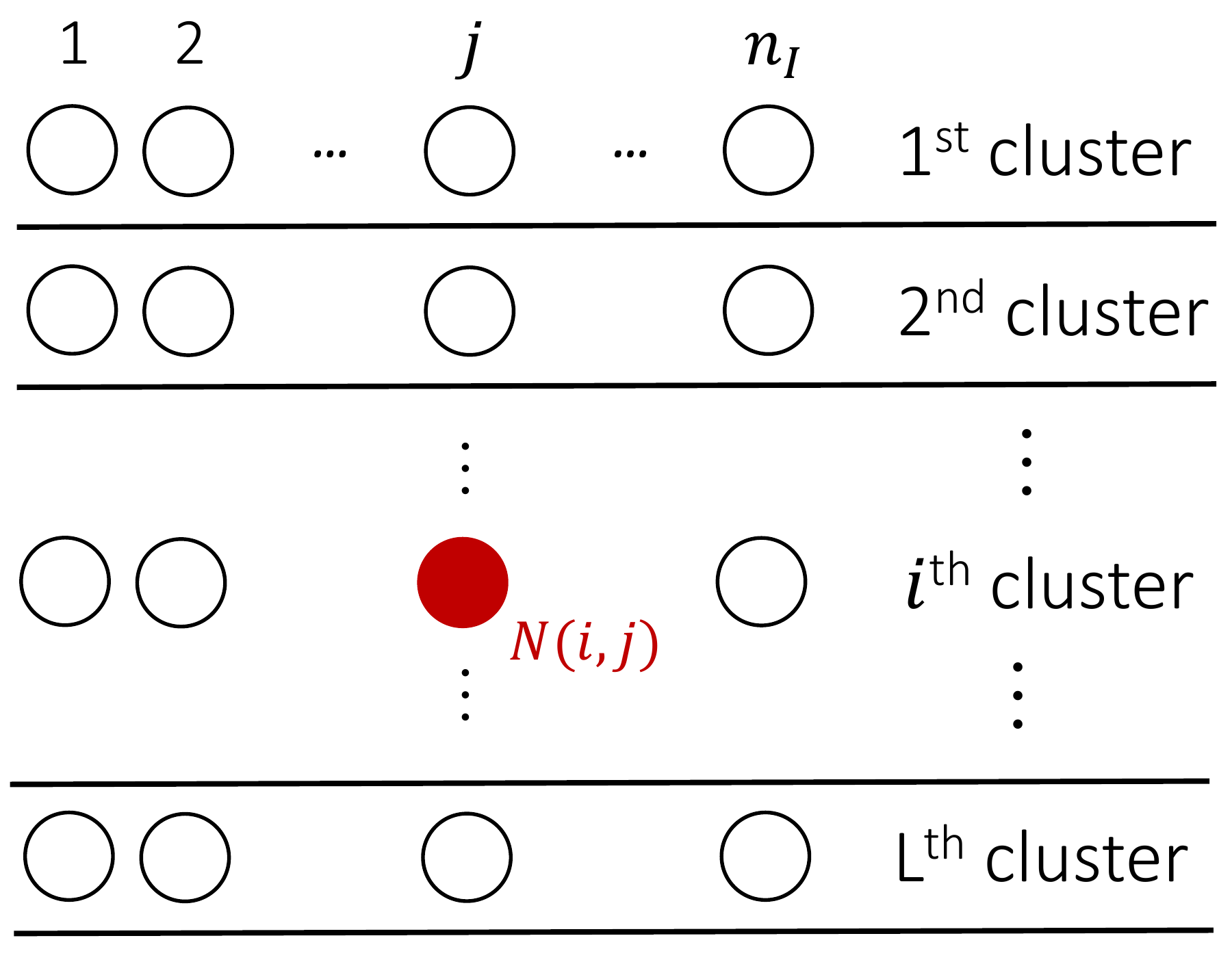}
	\caption{2-dim. structure representation}
	\label{Fig:partI_proof_Thm1_2}
\end{figure} 

To specify $G^*$, here we determine the 2-dimensional location of the $k$ newcomer nodes $\{x^{n+t}\}_{t=1}^{k}$.  
First, consider the 2-dimensional structure representation of clustered distributed storage, illustrated in Fig. \ref{Fig:partI_proof_Thm1_2}. In this figure, each row represents each cluster, and each node is represented as a 2-dimensional $(i,j)$ point for $i \in [L]$ and $j \in [n_I]$. 
The symbol $N(i,j)$ denotes the node at $(i,j)$ point. 
Here we define the set of $n$ nodes,
\begin{equation}\label{Eqn:set_of_nodes}
\mathcal{N} \coloneqq \{N(i,j) : i \in [L], j \in [n_I] \}.
\end{equation}
For $t \in [k]  $, consider selecting the newcomer node $x^{n+t}$ as
\begin{equation}\label{Eqn:mapping_btw_t_and_ij}
x^{n+t} = N(i_t,j_t)
\end{equation}
where
\begin{align}
i_t &= \min\{\nu \in [n_I] : \sum_{m=1}^{\nu} g_m \geq t \},\label{Eqn:i_from_t}\\ 
j_t &= t - \sum_{m=1}^{i_t-1} g_m \label{Eqn:j_from_t},
\end{align}
and $g_m$ used in the method is defined in (\ref{Eqn:g_m}). 
%
%
%
%
The location of $k$ newcomer nodes selected by this method are illustrated in Fig. \ref{Fig:partI_proof_Thm1_3}. Moreover, for the $n=12, L=3, k=9$ case, the newcomer nodes $\{x^{n+t}\}_{t=1}^k$ are depicted in Fig. \ref{Fig:partI_proof_Thm1_4}. In these figures, the node with number $t$ inside represents the newcomer node labeled as $x^{n+t}$. 

\begin{figure}[!t]
	\centering
	\includegraphics[width=85mm]{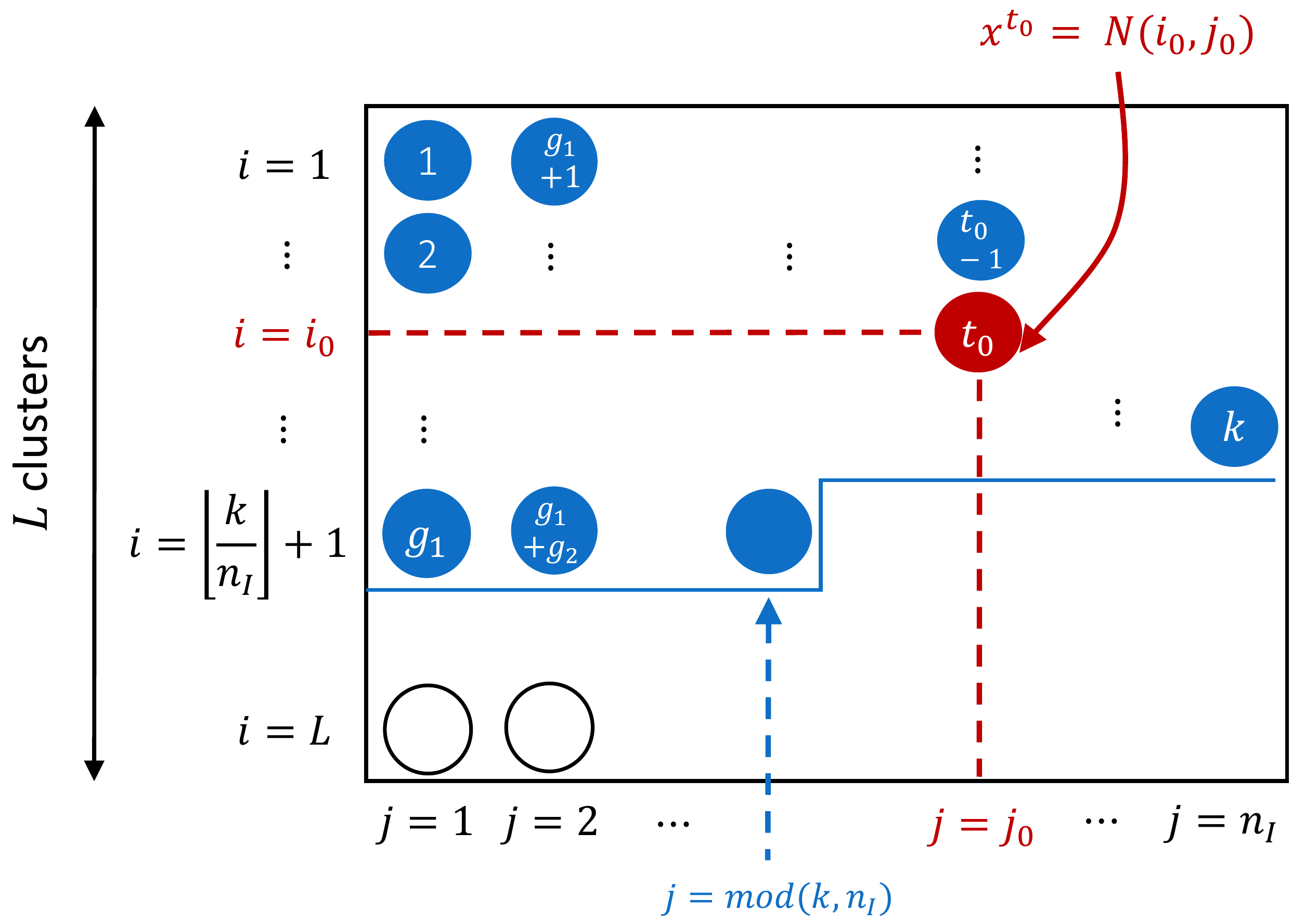}
	\caption{The location of $k$ newcomer nodes: $x^{n+1}, \cdots, x^{n+k}$}
	\label{Fig:partI_proof_Thm1_3}
\end{figure} 

\begin{figure}[!t]
	\centering
	\includegraphics[height=20mm]{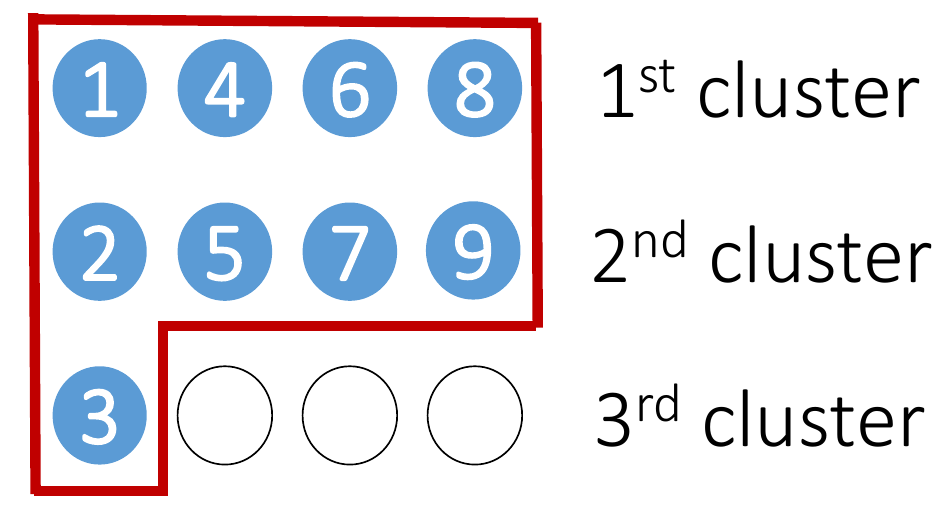}
	\caption{The location of $k$ newcomer nodes for $n=12, L=3, k=9$ case}
	\label{Fig:partI_proof_Thm1_4}
\end{figure} 


For the given graph $G^*$, now we consider a cut-set $c^* \in C(G^*)$ defined as below. 
The cut-set $c^*=(U,\overbar{U})$ can be defined by specifying $U$ and $\overbar{U}$ (complement of $U$), which partition the set of vertices in $G^*$.
First, let $x^{i}_{in}, x^{i}_{out} \in U$ for $i \in [n]$ and $x^{n+i}_{out} \in \overbar{U}$ for $ i \in [k]$.
\cmt{For $i \in [k]$, the input node $x^{n+i}_{in}$ is included in either $U$ or $\overbar{U}$, depending on the condition specified in the next paragraph.}
Moreover, let $S \in U$ and $DC \in \overbar{U}$. See Fig. \ref{Fig:cut_analysis}.

Let $U_0 = \bigcup_{i=1}^n \{x^i_{out}\} $. For $t \in [k]$,  let $\omega_t^*$ be the sum of capacities of edges from $U_0$ to $x^{n+t}_{in}$. If $\alpha \leq \omega_t^*$, then we include $x_{in}^{n+t}$ in $U$. Otherwise, we include $x^{n+t}_{in}$ in $\overbar{U}$. Then, the cut-set $c^*$ has the cut-value of 
\begin{equation}\label{Eqn:optimal_omega_t}
w(G^*,c^*) = \sum_{t=1}^{k} \min \{\alpha, \omega_t^*\}.
\end{equation} 
All that remains is to show that (\ref{Eqn:optimal_omega_t}) is equal to the expression in (\ref{Eqn:Capacity_value}). In other words, we will obtain the expression for $\omega_t^*$.

Recall that in the generation process of $G^*$, any newcomer node $x^{n+t}$ connects to $n-1$ helper nodes $\{x^{m}\}_{m=t+1}^{n+t-1}$ to regenerate $x^t$. 
Among the $n-1$ helper nodes, the $n_I - 1$ nodes reside in the same cluster with $x^t$, while the $n-n_I$ nodes are in other clusters. From our system setting in Section \ref{Section:Clustered_DSS}, the helper nodes in the same cluster as the failed node help by $\beta_I$, while the helper nodes in other clusters help by $\beta_c$. Therefore, the total repair bandwidth to regenerate any failed node is 
\begin{equation}\label{Eqn:gamma_recall}
\gamma = (n_I-1)\beta_I + (n-n_I)\beta_c
\end{equation}
as in (\ref{Eqn:gamma}). 

The newcomer node $x^{n+1}_{in}$ connects to $\{x_{out}^{m}\}_{m=2}^n$, all of which are included in $U_0$. Therefore,  
$\omega_1^* = \gamma = (n_I - 1)\beta_I + (n-n_I)\beta_c$
holds. 
Next, $x^{n+2}_{in}$ connects to $n-2$ nodes $\{x_{out}^{m}\}_{m=3}^{n}$ from $U_0$ and one node $x_{out}^{n+1}$ from $\overbar{U}$. 
 Define variable $\beta_{lm}$ as the repair bandwidth from $x_{out}^{n+l}$ to $x_{in}^{n+m}$. Then, $\omega_2^* = \gamma - \beta_{12}$.
From (\ref{Eqn:mapping_btw_t_and_ij}), we have $x^{n+1} = N(1,1)$ and $x^{n+2} = N(1,2)$. Therefore, $x^{n+1}$ and $x^{n+2}$ are in different clusters, which result in $\beta_{12} = \beta_c$. 
Therefore, 
$\omega_2^* = \gamma - \beta_{12} = \gamma - \beta_c.$

\begin{figure}
	\includegraphics[width=90mm]{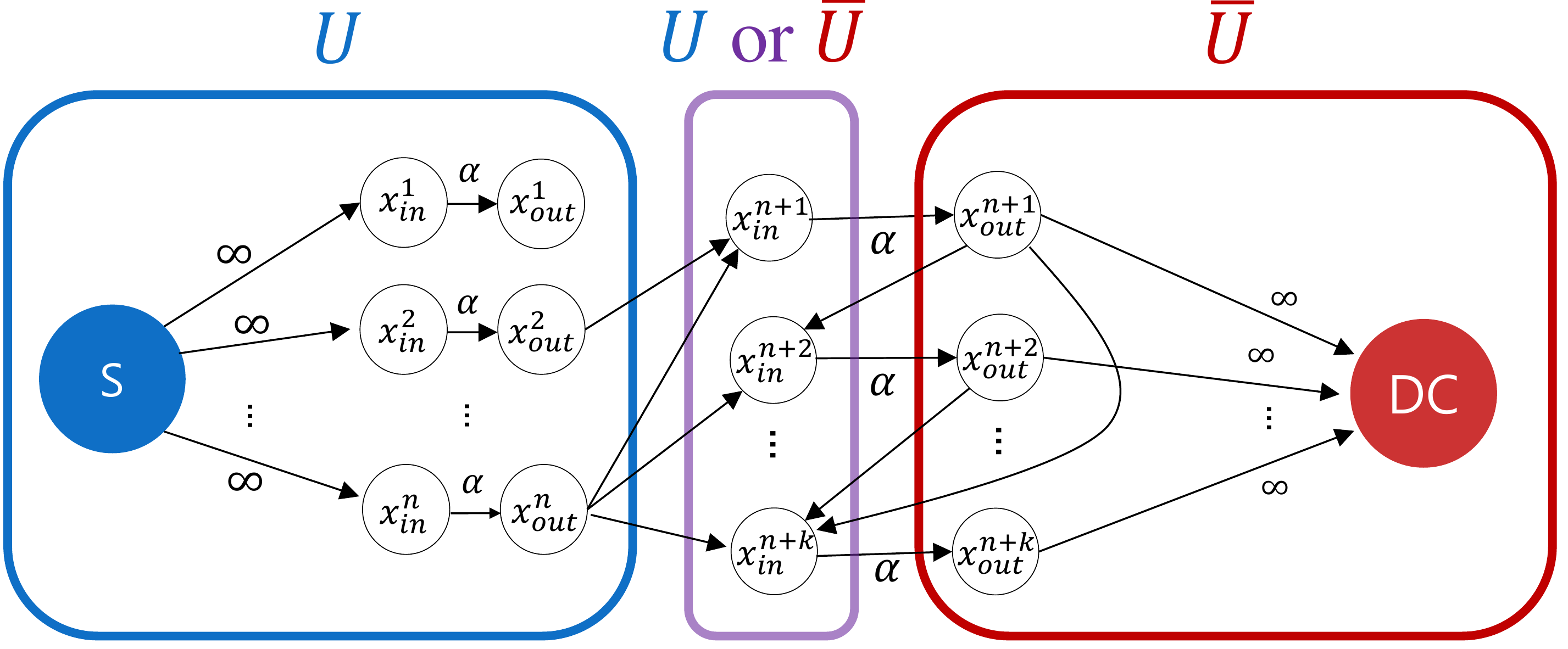}
	\caption{The cut-set $c^* = (U, \overbar{U})$ for the graph $G^*$}
	\label{Fig:cut_analysis}
	\vspace{-4mm}
\end{figure}


In general, $x^{n+t}_{in}$ connects to $n-t$ nodes $\{x_{out}^{m}\}_{m = t+1}^{n}$ from $U_0$, and $t-1$ nodes $\{x_{out}^{n+m}\}_{m=1}^{t-1}$ from $\overbar{U}$. Thus, $\omega_t^*$ for $t \in [k]$ can be expressed as 
$\omega_t^* = \gamma - \sum_{l=1}^{t-1} \beta_{lt}$
where
\begin{equation*}
	\beta_{lt} = 
	\begin{cases}
	\beta_I, & \text{ if } x^{n+l} \text{ and } x^{n+t} \text{ are in the same cluster} \\
	\beta_c, & \text{ otherwise.}
	\end{cases}
\end{equation*}

Recall Fig. \ref{Fig:partI_proof_Thm1_3}. For arbitrary newcomer node $x^{n+t} = N(i_t, j_t)$, the set $\{x^{n+m}\}_{m=1}^{t - 1}$ contains $i_t - 1$ nodes which reside in the same cluster with $x^{n+t}$, and $t - i_t$ nodes in other clusters. Therefore, $\omega_t^*$
can be expressed as
\begin{equation*}
\omega_t^* = \gamma - (i_t-1)\beta_I - (t-i_t)\beta_c
\end{equation*}
where $i_t$ is defined in (\ref{Eqn:i_from_t}). Combined with (\ref{Eqn:gamma_recall}) and (\ref{Eqn:j_from_t}), we get
\begin{align*}
\omega_t^* &= (n_I - i_t)\beta_I + (n-n_I - t + i_t)\beta_c \nonumber\\
&= (n_I - i_t)\beta_I + (n-n_I - j_t - \sum_{m=1}^{i_t-1} g_m + i_t)\beta_c.
\end{align*}
Then, (\ref{Eqn:optimal_omega_t}) can be expressed as 
\begin{align} \label{Eqn:optimal_mincut}
& w(G^*, c^*) = \nonumber\\
& \sum_{i=1}^{n_I} \sum_{t \in T_i} \min \{\alpha, (n_I-i)\beta_I + (n-n_I-j_t - \sum_{m=1}^{i-1}g_m + i) \beta_c\}
\end{align}
where $T_i = \{t \in [k]: i_t = i\}$. From the definition of $i_t$ in (\ref{Eqn:i_from_t}), we have
\begin{equation*}
T_i = \{\sum_{m=1}^{i-1} g_m + 1, \sum_{m=1}^{i-1} g_m + 2, \cdots, \sum_{m=1}^{i} g_m \}.
\end{equation*} 
Thus, $j_t = 1,2,\cdots, g_i$ for $t \in T_i$. Therefore, (\ref{Eqn:optimal_mincut}) can be expressed as
\begin{align*}
& w(G^*, c^*) = \\
& \sum_{i=1}^{n_I} \sum_{j=1}^{g_i} \min \{\alpha, (n_I-i)\beta_I + (n-n_I-j - \sum_{m=1}^{i-1}g_m + i) \beta_c\}, 
\end{align*}
which is identical to $T$ in (\ref{Eqn:Capacity_value}), where  $\rho_i$
used in this equation is defined in (\ref{Eqn:rho_i}).
Therefore, the specified information flow graph $G^*$ and the specified cut-set $c^*$ satisfy $\omega(G^*, c^*) = \sum_{t=1}^k \min\{\alpha, \omega_t^*\} = T$.



\vspace{5mm}
\textit{Part II}. Show that for every information flow graph $G \in \mathcal{G}$ and for every cut-set $c \in C(G)$, the cut-value $w(G,c)$ is greater than or equal to $T$ in (\ref{Eqn:Capacity_value}).
In other words, $\forall G \in \mathcal{G}, \forall c \in C(G)$, we have $w(G,c) \geq T$.

The proof is divided into 2 sub-parts: Part II-1 and Part II-2.

\vspace{5mm}
\textit{Part II-1}. Show that $\forall G \in \mathcal{G}, \forall c \in C(G),$ we have $w(G,c) \geq B(G,c)$ where $B(G,c)$ is in (\ref{Eqn:Lower_bound}):

Consider an arbitrary information flow graph $G \in \mathcal{G}$ and an arbitrary cut-set $c \in C(G)$ of the graph $G$. Denote the cut-set as $c=(U, \overbar{U})$.
Consider an output node $x_{out}^{i}$ connected to $DC$. If $x_{out}^{i} \in U$, then the cut-value $w(G,c)$ is infinity, which is a trivial case for proving $w(G,c) \geq B(G,c)$. Therefore, the $k$ output nodes connected to $DC$ are assumed to be in $\overbar{U}$. In other words, at least $k$ output nodes exist in $\overbar{U}$.
Note that every directed acyclic graph can be topologically sorted \cite{bang2008digraphs}, where vertex $u$ is followed by vertex $v$ if there exists a directed edge from $u$ to $v$. Consider labeling the topologically first $k$ output nodes in $\overbar{U}$ as $v_{out}^{1}, \cdots,v_{out}^{k}$.
Similar to the notation for a storage node $x^{i} = (x^{i}_{in}, x^{i}_{out})$ in Section \ref{Section:Info_flow_graph}, we denote the storage node which contains $v^{i}_{out}$ as $v^{i} = (v^{i}_{in}, v^{i}_{out})$.
\cmt{Then, the set of ordered $k-$tuples $(v^{i})_{i=1}^k$ can be represented as}
\begin{align} \label{Eqn:V_k}
	\mathcal{V}_k = \{(v^1, v^2, \cdots, v^k) : & \quad v^t \in \mathcal{N} \text{ for } t \in [k] \nonumber\\
	& \quad  v^{t_1} \neq v^{t_2} \text{ for } t_1 \neq t_2 \}.
\end{align}
We also define $u_i$, the sum of capacities of edges from $U$ to $v_{in}^i$. See Fig. \ref{Fig:proof_thm1_part2_1}.

\begin{figure}[!t]
	\centering
	\includegraphics[height=30mm]{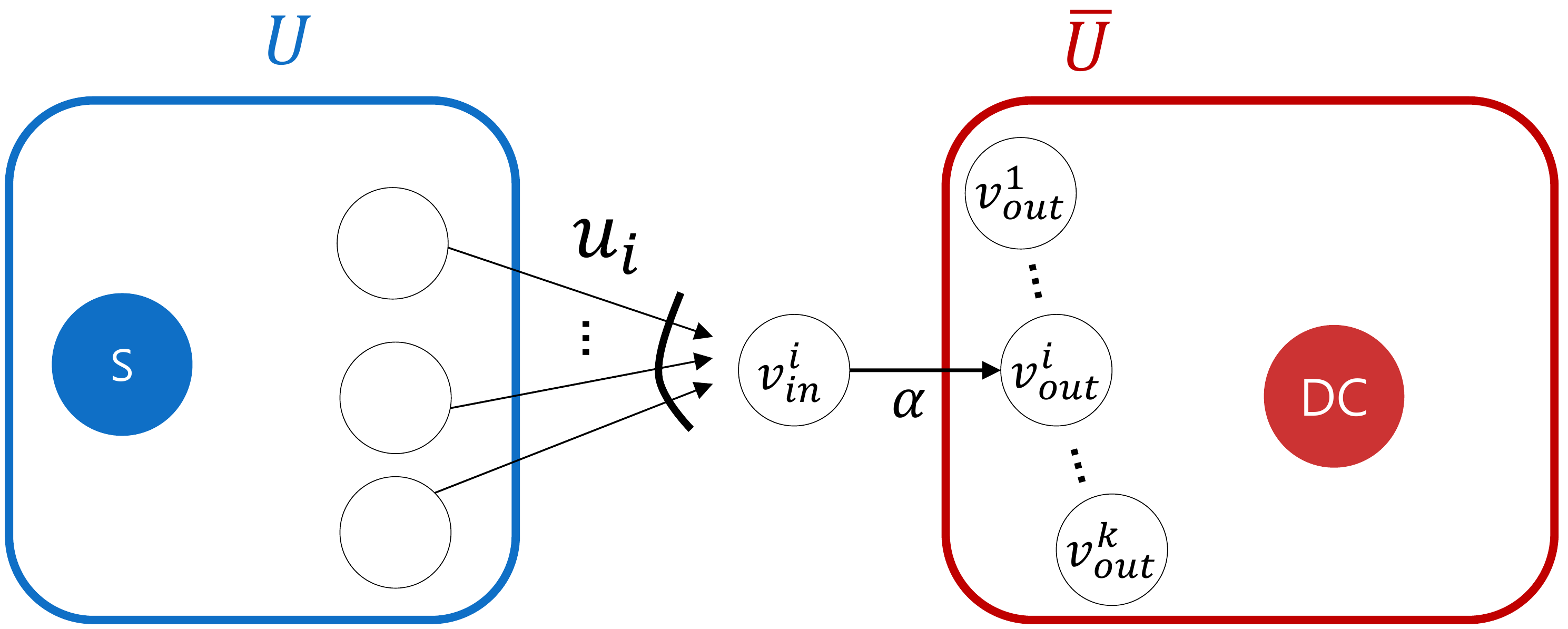}
	\caption{Arbitrary information flow graph $G \in \mathcal{G}$ and arbitrary cut-set $c \in C(G)$}
	\label{Fig:proof_thm1_part2_1}
\end{figure}

If $v_{in}^1 \in U$, then the cut-set $c =(U, \overbar{U})$ should include the edge from $v_{in}^1$ to $v_{out}^1$, which has the edge capacity $\alpha$. Otherwise (i.e., $v_{in}^1 \in \overbar{U}$), the cut-set $c =(U, \overbar{U})$ should include the edges from $U$ to $v_{in}^1$. If $v_{in}^1$ node is directly connected to the source node $S$, the cut-value $w(G,c)$ is infinity (trivial case for proving $w(G,c) \geq B(G,c)$). Therefore, $v_{in}^1$ node is assumed to be a newcomer node helped by $n-1$ helper nodes. Note that all helper nodes of $v_{in}^1$ are in $U$, since $v_{out}^1$ is the topologically first output node in $\overbar{U}$. Thus, the cut-set $c$ should include the edges from $U$ to $v_{in}^1$, where the sum of capacities of these edges are
\begin{equation*}
u_1 = \gamma = (n_I - 1)\beta_I + (n-n_I)\beta_c. 
\end{equation*}

If $v_{in}^2 \in U$, then the cut-set $c =(U, \overbar{U})$ should include the edge from $v_{in}^2$ to $v_{out}^2$, which has the edge capacity $\alpha$. Otherwise (i.e., $v_{in}^2 \in \overbar{U}$), the cut-set $c =(U, \overbar{U})$ should include the edges from $U$ to $v_{in}^2$. As we discussed in the case of $v_{in}^1$, we can assume $v_{in}^2$ is a newcomer node helped by $n-1$ helper nodes. Since $v_{in}^2$ is the topologically second node in $\overbar{U}$, it may have one helper node $v_{out}^1 \in \overbar{U}$; at least $n-2$ helper nodes in $U$ help to generate $v_{in}^2$. 
 Note that the total amount of data coming into $v_{in}^2$ is $\gamma = (n_I - 1)\beta_I + (n-n_I)\beta_c$, while the amount of information coming from $v_{out}^1$ to $v_{in}^2$, denoted $\beta_{12}$, is as follows: if $v^1$  and $v^2$ are in the same cluster, $\beta_{12} = \beta_I$, otherwise $\beta_{12} = \beta_c$.
Recall that the cut-set should include the edges from $U$ to $v_{in}^2$. The sum of capacities of these edges are 
\begin{equation*}
u_2 \geq \gamma - \beta_{12},
\end{equation*}
while the equality holds if and only if $v_{out}^1$ helps $v_{in}^2$. In a similar way, for $i \in [k]$, $u_i$ can be bounded as
\begin{equation}\label{Eqn:u_t_lower_bound}
u_i \geq \omega_i 
\end{equation}
where 
\begin{align}
\omega_i &= \gamma -  \displaystyle\sum_{j=1}^{i-1} \beta_{ji},\label{Eqn:omega_i}\\
\beta_{ji} &= 
\begin{cases}
\beta_I, & v^j \text{ and } v^i \text{ are in the same cluster} \\
\beta_c, & \text{otherwise.} \\
\end{cases}\label{Eqn:beta_ji_value}
\end{align}
The equality in (\ref{Eqn:u_t_lower_bound}) holds if and only if $v_{out}^{j}$ helps $v_{in}^i$ for $j \in [i-1]$.

Thus, $v_{out}^i$ contributes at least $\min \{\alpha, \omega_i\}$ to the cut value, for $i \in [k]$.
In summary, for arbitrary graph $G \in \mathcal{G}$, an arbitrary cut-set $c$ has cut-value $w(G,c)$ of at least $\sum_{i=1}^k \min \{\alpha, \omega_i\}$: 
\begin{equation} \label{Eqn:omega_bound}
w(G,c) \geq \sum_{i=1}^k \min \{\alpha, \omega_i\}, \quad \quad  \forall G \in \mathcal{G}, \forall c \in C(G).
\end{equation}
Note that $\{\omega_i\}$ depends on the relative position of $\{v_{out}^i\}_{i=1}^k$, which is determined when an arbitrary information flow graph $G \in \mathcal{G}$ and arbitrary cut-set $c \in C(G)$ are specified.
This relationship is illustrated in Fig. \ref{Fig:Dependency_graph}. Therefore, we define
\begin{equation}\label{Eqn:Lower_bound}
B(G,c) = \sum_{i=1}^k \min \{\alpha, \omega_i\} 
\end{equation}
for arbitrary $G \in \mathcal{G}$ and arbitrary $c \in C(G)$.
Combining (\ref{Eqn:omega_bound}) and (\ref{Eqn:Lower_bound}) completes the proof \textit{part II-1}.

\begin{figure}[!t]
	\centering
	\includegraphics[width=85mm]{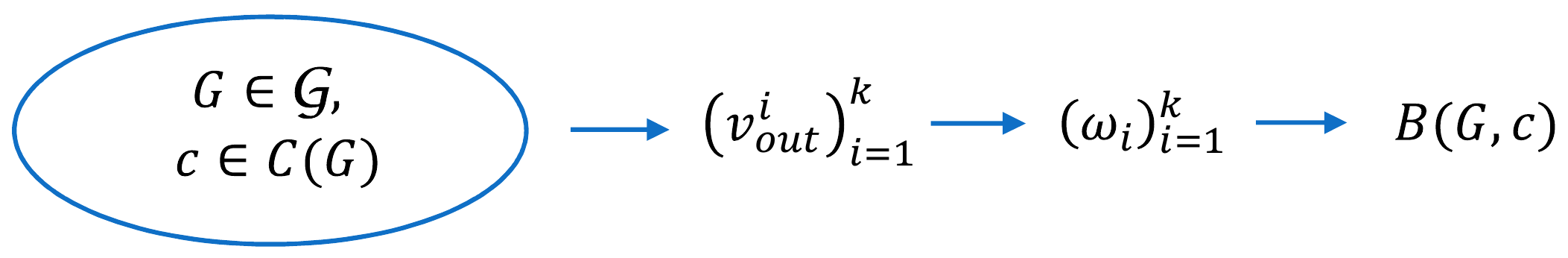}
	\caption{Dependency graph for variables in the proof of Part II.}
	\label{Fig:Dependency_graph}
\end{figure}

\vspace{5mm}

\textit{Part II-2}. $\displaystyle\min_{G \in \mathcal{G}} \ \min_{c \in C(G)} B(G,c) = R$:


Assume that $\alpha$ and $k$ are fixed. See Fig. \ref{Fig:Dependency_graph}.
Note that for a given graph $G \in \mathcal{G}$ and a cut-set $c = (U, \overbar{U})$ with $ c \in C(G)$, the sequence of topologically first $k$ output nodes $(v_{out}^i)_{i=1}^k$ in $\overbar{U}$ is determined. Moreover, for a given sequence $(v_{out}^i)_{i=1}^k$, we have a fixed $(\omega_i)_{i=1}^k$, which determines $B(G,c)$ in (\ref{Eqn:Lower_bound}).
Thus, $\displaystyle\min_{G \in \mathcal{G}} \ \min_{c \in C(G)} B(G,c)$ can be obtained by finding the optimal $(v_{out}^i)_{i=1}^k$ sequence which minimizes $B(G,c)$. It is identical to finding the optimal $(v^i)_{i=1}^k$, the sequence of $k$  different nodes out of $n$ existing nodes in the system.
Therefore, based on (\ref{Eqn:Lower_bound}) and (\ref{Eqn:omega_i}), we have 
\begin{equation}\label{minmin_equiv_1}
\displaystyle\min_{G \in \mathcal{G}} \ \min_{c \in C(G)} B(G,c) = \displaystyle\min_{(v^i)_{i=1}^k \in \mathcal{V}_k} \left(\sum_{i=1}^k \min \{\alpha, \gamma - \sum_{j=1}^{i-1}\beta_{ji}\}\right)
\end{equation}
where
\begin{equation}\label{Eqn:beta_ji_orig}
\beta_{ji} = 
\begin{cases}
\beta_I, & v^j \text{ and } v^i \text{ are in the same cluster} \\
\beta_c, & \text{otherwise.} \\
\end{cases}
\end{equation}
holds as defined in (\ref{Eqn:beta_ji_value}), and $\mathcal{V}_k$ is defined in (\ref{Eqn:V_k}).
In order to obtain the solution for RHS of (\ref{minmin_equiv_1}), all we need to do is to find the optimal sequence $(v^i)_{i=1}^k$ of $k$  different nodes, which can be divided into two sub-problems: 
$i)$ finding the optimal  way of selecting $k$ nodes $\{v^i\}_{i=1}^k$ out of $n$ nodes, and $ii)$ finding the optimal order of selected $k$ nodes. 
Note that there are $n \choose k$ selection methods and $k!$ ordering methods.
Each selection method can be assigned to a \textit{selection vector} $\bm{s}$ defined in Definition \ref{Def:Selection vector}, and each ordering method can be assigned to an \textit{ordering vector} $\bm{\pi}$ defined in Definition \ref{Def:Ordering vector}.

First, we define a selection vector $\bm{s}$ for a given $\{v^i\}_{i=1}^k$.
 \begin{definition}\label{Def:Selection vector}
 	Assume arbitrary $k$ nodes are selected as $\{v^{i}\}_{i=1}^k$. Label each cluster by the number of selected nodes in a descending order. In other words, the $1^{st}$ cluster contains a maximum number of selected nodes, and the $L^{th}$ cluster contains a minimum number of selected nodes. Under this setting, define the selection vector $\bm{s} = [s_1, s_2, \cdots, s_L]$ where $s_i$ is the number of selected nodes in the $i^{th}$ cluster. 
 \end{definition}

 \begin{figure}[!t]
 	\centering
 	\includegraphics[height=25mm]{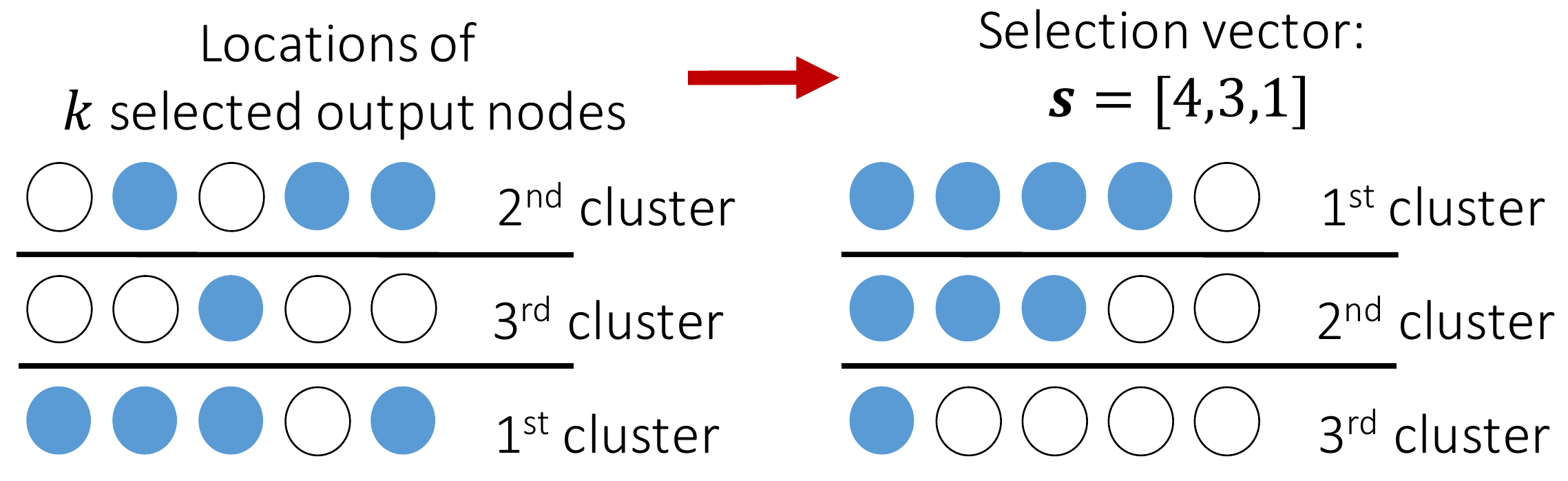}
 	\caption{Obtaining the selection vector for given $k$ output nodes $\{v_{out}^{i}\}_{i=1}^k$ ($n=15, k=8, L=3$)}
 	\label{Fig:Selection Vector}
 	\vspace{-1mm}
 \end{figure}

 Fig. \ref{Fig:Selection Vector} shows an example of selection vector $\bm{s}$ corresponding to the selected $k$ nodes $\{v^{i}\}_{i=1}^k$.
 From the definition of the selection vector, the set of possible selection vectors can be specified as follows.
 \begin{align*}
 \mathcal{S} =  \big\{\bm{s} &=[s_1, \cdots, s_L] : 0 \leq s_i \leq n_I, s_{i+1} \leq s_{i}, \sum_{i=1}^{L} s_i = k \big\}.
 \end{align*}
 Note that even though ${n\choose k}$ different selections exist, the $\{\omega_i\}$ values in (\ref{Eqn:omega_i}) are only determined by the corresponding selection vector $\bm{s}$.
 This is because $\{\omega_i\}$ depends only on the relative positions of $\{v^{i}\}_{i=1}^k$, whether they are in the same cluster or in different clusters.
 Therefore, comparing the $\{\omega_i\}$ values of all $\vert \mathcal{S} \vert$ possible selection vectors $\bm{s}$ is enough; it is not necessary to compare the $\{\omega_i\}$ values of ${n\choose k}$ selection methods.
 Now, we define the ordering vector $\bm{\pi}$ for a given selection vector $\bm{s}$. 
 \begin{definition}\label{Def:Ordering vector}
 	Let the locations of $k$ nodes $\{v^{i}\}_{i=1}^k$ be fixed with a corresponding selection vector $\bm{s} = [s_1, \cdots, s_L]$. Then, for arbitrary ordering of the selected $k$ nodes, define the ordering vector $\bm{\pi} = [\pi_1, \cdots, \pi_k]$ where $\pi_i$ is the index of the cluster which contains $v^{i}$.
 \end{definition}
 For a given $\bm{s}$, the ordering vector $\bm{\pi}$ corresponding to an arbitrary ordering of $k$ nodes is illustrated in Fig. \ref{Fig:Ordering Vector}. 
 In this figure (and the following figures in this paper), the number $i$ written inside each node means that the node is $v^{i}$.
  From the definition, an ordering vector $\bm{\pi}$ has $s_l$ components with value $l$, for all $l\in [L]$. 
 The set of possible ordering vectors can be specified as
 \begin{equation}\label{Eqn:set of ordering vectors}
 \Pi(\bm{s}) = \big\{ \bm{\pi} = [\pi_1, \cdots, \pi_k] : \sum_{i=1}^{k} \mathds{1}_{\pi_i = l} = s_l, \ \forall l \in [L] \big\} 
 \end{equation}
 Note that for given $k$ selected nodes, there exists $k! $ different ordering methods. However, the $\{\omega_i\}$ values in (\ref{Eqn:omega_i}) are only determined by 
 the corresponding ordering vector $\bm{\pi} \in \Pi(\bm{s})$, by similar reasoning for compressing ${n\choose k}$ selection methods to $\vert \mathcal{S} \vert$ selection vectors. 
 Therefore, comparing the $\{\omega_i\}$ values of all possible ordering vectors $\bm{\pi}$ is enough; it is not necessary to compare the $\{\omega_i\}$ values of all $k!$ ordering methods.
 
 \begin{figure}[!t]
 	\centering
 	\includegraphics[height=25mm]{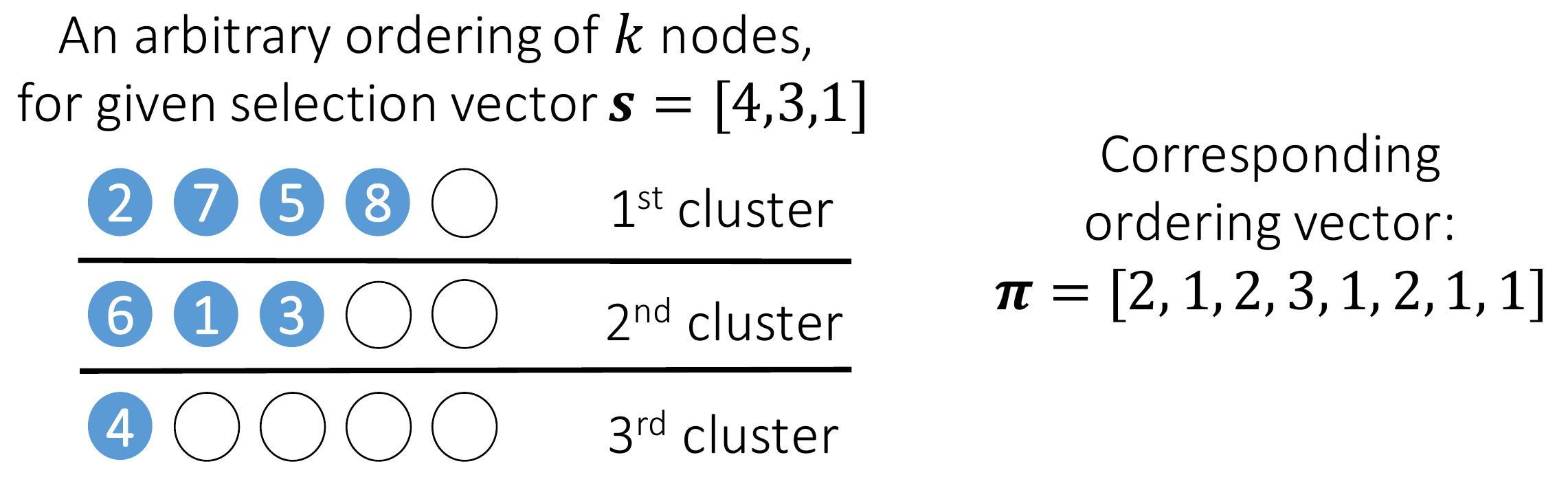}
 	\caption{Obtaining the ordering vector for given an arbitrary order of $k$ output nodes ($n=15, k=8, L=3$)}
 	\label{Fig:Ordering Vector}
 	\vspace{-3mm}
 \end{figure}

 Thus, finding the optimal sequence $(v^i)_{i=1}^k$ is identical to specifying the optimal ($\bm{s, \pi}$) pair, which is obtained as follows.
Recall that from the definition of $\bm{\pi} = [\pi_1, \pi_2, \cdots, \pi_k]$ in Definition \ref{Def:Ordering vector}, $\pi_i = \pi_j$ holds if and only if $v^i$ and $v^j$ are in the same cluster. Therefore,  $\beta_{ji}$ in  (\ref{Eqn:beta_ji_orig}) can be expressed by using $\bm{\pi}$ notation:
\begin{equation}\label{Eqn:beta_ji}
\beta_{ji} (\bm{\pi})= 
\begin{cases}
\beta_I	&	\text{ if } \pi_i = \pi_j\\
\beta_c &   \text{ otherwise}
\end{cases}
\end{equation}
Thus, the $\sum_{j=1}^{i-1}\beta_{ji}$ term in  (\ref{minmin_equiv_1}) can be expressed as
\begin{align}\label{Eqn:sum_beta_ji}
\sum_{j=1}^{i-1}\beta_{ji}(\bm{\pi}) &= (\sum_{j=1}^{i-1} \mathds{1}_{\pi_j = \pi_i}) \beta_I + (i - 1 - \sum_{j=1}^{i-1} \mathds{1}_{\pi_j = \pi_i}) \beta_c. \end{align}
Combining (\ref{Eqn:gamma}), (\ref{minmin_equiv_1}) and (\ref{Eqn:sum_beta_ji}), we have
\begin{equation*}
\displaystyle\min_{G \in \mathcal{G}} \ \min_{c \in C(G)} B(G,c) = \displaystyle\min_{\bm{s} \in S} \ \min_{\bm{\pi} \in \Pi(\bm{s})} L  (\bm{s},\bm{\pi})  
\end{equation*}
where 
\begin{align}
L  (\bm{s},\bm{\pi}) &= \sum_{i=1}^k \min\{\alpha, \omega_i(\bm{\pi})\}, \label{Eqn:lower_bound}\\
\omega_i(\bm{\pi}) &= \gamma - \sum_{j=1}^{i-1} \beta_{ji} (\bm{\pi})= a_i(\bm{\pi}) \beta_I + (n-i-a_i(\bm{\pi})) \beta_c, \label{Eqn:weight vector}\\
a_i(\bm{\pi}) &= n_I - 1 - \sum_{j=1}^{i-1} \mathds{1}_{\pi_j = \pi_i}. \label{Eqn:a_i}
\end{align}

Therefore, the rest of the proof in \textit{part II-2} shows that
\begin{equation}\label{Eqn:Capacity_expression_minmin}
\displaystyle\min_{\bm{s} \in S} \ \min_{\bm{\pi} \in \Pi(\bm{s})} L  (\bm{s},\bm{\pi}) = R
\end{equation} holds.
We begin by stating a property of $\omega_i(\bm{\pi})$ seen in (\ref{Eqn:weight vector}).
\begin{prop}\label{Prop:weight vector}
	Consider a fixed selection vector $\bm{s}$. We claim that $\sum_{i=1}^k \omega_i(\bm{\pi})$ is constant irrespective of the ordering vector $\bm{\pi} \in \Pi(\bm{s})$. 
\end{prop}
\begin{proof}
	Let $\bm{s}= [s_1, \cdots, s_L]$.
	For an arbitrary ordering vector $\bm{\pi} \in \Pi(\bm{s})$, let $b_i(\bm{\pi})= n-i-a_i(\bm{\pi})$ where $a_i(\bm{\pi})$ is as given in (\ref{Eqn:a_i}). 
	For simplicity, we denote $a_i(\bm{\pi})$, $b_i(\bm{\pi})$ and $\omega_i(\bm{\pi})$  as $a_i$, $b_i$ and $\omega_i$, respectively.
	Then, 
	\begin{equation}\label{Eqn:proof of proposition 2}
	\sum_{i=1}^{k}(a_i + b_i) = \sum_{i=1}^{k}(n-i) = \textit{constant (const.)}
	\end{equation}
	for fixed $n,k$.
	Note that
	\begin{equation}\label{Eqn:sum_a_i}
	\sum_{i=1}^{k}a_i = k(n_I - 1) - \sum_{i=1}^{k} \sum_{j=1}^{i-1} \mathds{1}_{\pi_j = \pi_i}
	\end{equation} 
	from (\ref{Eqn:a_i}). Also, from the definition of $\Pi(\bm{s})$ in (\ref{Eqn:set of ordering vectors}), an arbitrary ordering vector $\bm{\pi} \in \Pi(\bm{s})$ has $s_l$ components with value $l$, for all $l\in [L]$. 
	 If we define 
	 \begin{equation}\label{Eqn:I_l}
	 I_l (\bm{\pi}) = \{ i \in [k] : \pi_i = l \},
	 \end{equation}
	 then $|I_l(\bm{\pi})| = s_l$ holds for $l \in [L]$. Then, 
	\begin{equation*}
	\sum_{i\in I_l(\bm{\pi})} \sum_{j=1}^{i-1} \mathds{1}_{\pi_j = \pi_i} = 0 + 1 + \cdots + (s_l-1) = \sum_{t=0}^{s_l - 1} t.	
	\end{equation*}
	Therefore,
	\begin{equation*}
	\sum_{i=1}^{k} \sum_{j=1}^{i-1} \mathds{1}_{\pi_j = \pi_i} = \sum_{l=1}^{L} \sum_{i\in I_l(\bm{\pi})} \sum_{j=1}^{i-1} \mathds{1}_{\pi_j = \pi_i} = \sum_{l=1}^{L} \sum_{t=0}^{s_l-1} t = \textit{const.}
	\end{equation*}
	for fixed $L, \bm{s}$.
	Combining with (\ref{Eqn:sum_a_i}), 
	\begin{equation}\label{Eqn:sum_a_i_const}
	\sum_{i=1}^{k}a_i  = k(n_I - 1) - \sum_{l=1}^{L} \sum_{t=0}^{s_l-1} t = \textit{const.}
	\end{equation}
	for fixed $n,k,L,\bm{s}$.
	From (\ref{Eqn:proof of proposition 2}) and (\ref{Eqn:sum_a_i_const}) we get 
	\begin{equation*}
	\sum_{i=1}^k b_i = \textit{const.}
	\end{equation*}
	Therefore, from (\ref{Eqn:weight vector}),
	\begin{equation*}
	\sum_{i=1}^k \omega_i = (\sum_{i=1}^k a_i) \beta_I + (\sum_{i=1}^k b_i) \beta_c = \textit{const.}
	\end{equation*}	
	for every ordering vector $\bm{\pi} \in \Pi(\bm{s})$, if $n,k,L,\beta_I, \beta_c$ and $\bm{s}$ are fixed.
\end{proof}

Now, we define a special ordering vector $\bm{\pi}_v$ called \textit{vertical ordering vector}, which is shown to be the optimal ordering vector $\bm{\pi}$ which minimizes $L  (\bm{s},\bm{\pi}) $ for an arbitrary selection vector $\bm{s}$.

\begin{definition}\label{Def:vertical ordering vector}
	For a given selection vector $\bm{s} \in \mathcal{S}$, the corresponding vertical ordering vector $\bm{\pi}_v (\bm{s})$, or simply denoted as $\bm{\pi}_v$, is defined as the output of Algorithm \ref{Alg:vertical ordering}.
\end{definition}

The vertical ordering vector $\bm{\pi}_v$ is illustrated in Fig. \ref{Fig:vertical ordering}, for a given selection vector $\bm{s}$ as an example.
For $\bm{s}=[4,3,1]$, Algorithm 1 produces the corresponding vertical ordering vector $\bm{\pi}_v = [1,2,3,1,2,1,2,1]$. Note that the order of $k=8$ output nodes is  illustrated in Fig. \ref{Fig:vertical ordering}, as the numbers inside each node. 
Although the vertical ordering vector $\bm{\pi}_v$ depends on the selection vector $\bm{s}$, we use simplified notation $\bm{\pi}_v$ instead of $\bm{\pi}_v(\bm{s})$.
From Fig. \ref{Fig:vertical ordering}, obtaining $\bm{\pi}_v$ using Algorithm \ref{Alg:vertical ordering} can be analyzed as follows.
Moving from the leftmost column to the rightmost column, the algorithm selects one node per cluster. After selecting all $k$ nodes, $\pi_i$ stores the index of the cluster which contains the $i^{th}$ selected node.
Now, the following Lemma shows that the vertical ordering vector $\bm{\pi}_v$ is optimal in the sense of minimizing $L  (\bm{s},\bm{\pi}) $ for an arbitrary selection vector $\bm{s}$.

\begin{algorithm}[t]
	\small
	\caption{Generate vertical ordering $\bm{\pi}_v  $}
	\label{Alg:vertical ordering}
	\begin{algorithmic}
		\STATE \textbf{Input:} $\bm{s} = [s_1, \cdots, s_L]$
		\STATE \textbf{Output:} $\bm{\pi}_v = [\pi_1, \cdots, \pi_k]$
		\STATE Initialization: $l \leftarrow 1$
		\FOR{$i=1$ to $k$}
		\IF{$s_l = 0$}
		\STATE $l \leftarrow 1$
		\ENDIF
		\STATE $\pi_i \leftarrow l$ \quad \quad \quad \quad \quad \quad  (Store the index of cluster)
		\STATE $s_{\pi_i} \leftarrow s_{\pi_i} - 1$ \quad \quad \quad (Update the remaining node info.)
		\STATE $l \leftarrow (l \Mod{L}) + 1$ \quad \  (Go to the next cluster)
		\ENDFOR
	\end{algorithmic}
\end{algorithm}

\begin{figure}[t]
	\centering
	\includegraphics[height=25mm]{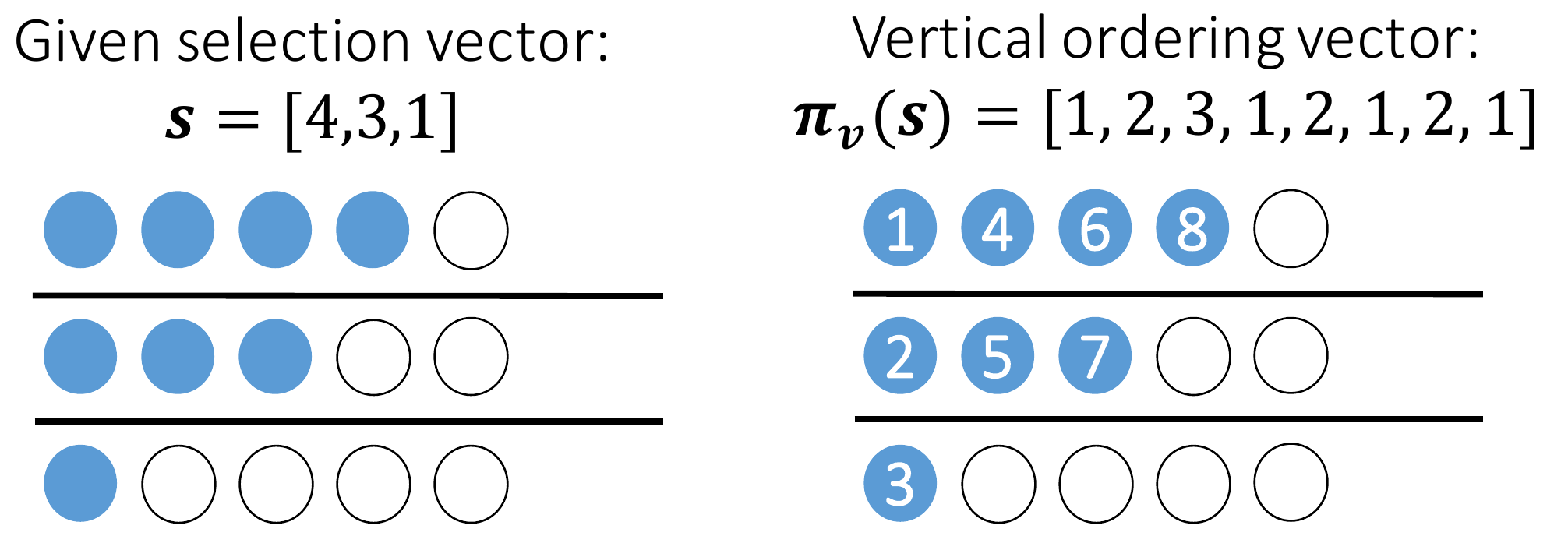}
	\caption{The vertical ordering vector $\bm{\pi}_v $ for the given selection vector $\bm{s}=[4,3,1]$ (for $n=15, k=8, L=3$ case) }
	\label{Fig:vertical ordering}
\end{figure}

\begin{lemma} \label{Lemma:optimal ordering vector}
	Let $\bm{s} \in \mathcal{S}$ be an arbitrary selection vector. Then, a vertical ordering vector $\bm{\pi}_v$ minimizes  $ L  (\bm{s},\bm{\pi}) $.
	In other words, $ L  (\bm{s},\bm{\pi}_v)   \leq  L  (\bm{s},\bm{\pi}) $ holds for arbitrary $\bm{\pi} \in \Pi(\bm{s})$.
\end{lemma}
\begin{proof}
	In the case of $\beta_c = 0$, we show that $L  (\bm{s},\bm{\pi}) $ is constant for every $\bm{\pi} \in \Pi(\bm{s})$. 
	From (\ref{Eqn:lower_bound}), 
	\begin{equation}\label{Eqn:lower_bound_zero_betac}
	L(\bm{s}, \bm{\pi}) = \sum_{i=1}^k \min \{\alpha, a_i(\bm{\pi}) \beta_I\}
	\end{equation}
	holds for $\beta_c = 0$. 
	Using $I_l(\bm{\pi})$ in (\ref{Eqn:I_l}), note that $[k]$ can be partitioned into $L$ disjoint subsets as
	$[k] = \bigcup_{l=1}^L I_l(\bm{\pi}).$
	Therefore, (\ref{Eqn:lower_bound_zero_betac}) can be written as
	\begin{equation*}
		L(\bm{s}, \bm{\pi}) = \sum_{l=1}^L \sum_{i \in I_l(\bm{\pi})} \min \{\alpha, a_i(\bm{\pi}) \beta_I\}.
	\end{equation*}
	Recall that $\bm{\pi} \in \Pi(\bm{s})$ contains $s_l$ components with value $l$ for $l \in [L]$. Thus, from (\ref{Eqn:a_i}),
	\begin{equation*}
	\bigcup_{i\in I_l (\bm{\pi})} \{a_i(\bm{\pi})\} = \{n_I - 1, n_I-2, \cdots, n_I - s_l\}
	\end{equation*}
	for $l \in [L]$.
	Therefore, 
	$\sum_{i \in I_l(\bm{\pi})} \min \{\alpha, a_i(\bm{\pi}) \beta_I\}$
	is constant $\forall \bm{\pi} \in \Pi(\bm{s})$ for arbitrary $l \in [L]$. In conclusion, $L  (\bm{s},\bm{\pi})$  in (\ref{Eqn:lower_bound_zero_betac}) is constant irrespective of $\bm{\pi} \in \Pi(\bm{s})$ for $\beta_c = 0$.

	The rest of the proof deals with the $\beta_c \neq 0$ case.
	For a given arbitrary $\bm{s} \in \mathcal{S}$, define two subsets of $\Pi(\bm{s})$ as
	\begin{align}
	&\Pi_r = \{\bm{\pi^*} \in \Pi(\bm{s}) :\ S_t(\bm{\pi}) \leq S_t (\bm{\pi^*}) , \forall t \in [k], \ \forall \bm{\pi} \in \Pi(\bm{s}) \}\label{Eqn:running_sum_maximizer}\\
	&\Pi_m = \{\bm{\pi^*} \in \Pi(\bm{s}) :  L  (\bm{s},\bm{\pi})  \geq  L  (\bm{s},\bm{\pi}^*), \forall \bm{\pi} \in \Pi(\bm{s}) \} \label{Eqn:min_cut_minimizer}
	\end{align} 
	where 
	$S_t(\bm{\pi}) = \sum_{i=1}^t w_i (\bm{\pi})$
	is the running sum of $w_i(\bm{\pi})$.
	Here, we call $\Pi_r$ the \textit{running sum maximizer} and $\Pi_m$ the \textit{min-cut minimizer}. Now the proof proceeds in two steps. 
    The first step proves that the \textit{running sum maximizer} minimizes min-cut, i.e., $\Pi_r \subseteq \Pi_m$. 
	The second step proves that the vertical ordering vector is a running sum maximizer, i.e., $\bm{\pi}_v \in \Pi_r$.
	
	\textbf{Step 1.} Prove $\Pi_r \subseteq \Pi_m$:
	
	Define two index sets for a given ordering vector $\bm{\pi}$:
	\begin{align}\label{Eqn:Omega_L,s}
	\Omega_L (\bm{\pi}) &= \{ i \in [k] : w_i (\bm{\pi}) \geq \alpha \} \nonumber\\
	\Omega_s (\bm{\pi}) &= \{ i \in [k] : w_i (\bm{\pi}) < \alpha \} 
	\end{align}
	Now define a set of ordering vectors as
	\begin{equation}\label{Eqn:partitionable}
	\Pi_p = \{ \bm{\pi} \in \Pi(\bm{s}) : i \leq j \ \forall i \in \Omega_L (\bm{\pi}), \ \forall j \in \Omega_s (\bm{\pi}) \}.
	\end{equation}
	The rest of the proof is divided into 2  sub-steps.
	
	\textbf{Step 1-1.} Prove  $\Pi_m \subseteq \Pi_p$ and $ \Pi_r \subseteq \Pi_p$ by transposition:
	
	Consider arbitrary $\bm{\pi}=[\pi_1, \cdots, \pi_k] \in \Pi_p^c$. Use a short-hand notation $\omega_i$ to represent $\omega_i(\bm{\pi})$ for $i \in [k]$.
	From (\ref{Eqn:partitionable}), there exists $i > j$ such that $i \in \Omega_L(\bm{\pi})$ and $j \in \Omega_s(\bm{\pi})$. Therefore, there exists $t \in [k-1]$ such that $t+1 \in \Omega_L(\bm{\pi})$ and $t \in \Omega_s(\bm{\pi})$ hold. By (\ref{Eqn:Omega_L,s}), 
	\begin{equation}\label{Eqn:omega_relative}
	\omega_{t+1} \geq \alpha > \omega_{t}
	\end{equation} for some $t \in [k-1]$.
	Note that from (\ref{Eqn:weight vector}), $\pi_t = \pi_{t+1}$ implies $\omega_{t+1} = \omega_t - \beta_I < \omega_t$. Therefore, 
	\begin{equation}\label{Eqn:pi_relative}
	\pi_t \neq \pi_{t+1}
	\end{equation} should hold to satisfy (\ref{Eqn:omega_relative}).
	Define an ordering vector $\bm{\pi'} = [\pi'_1, \cdots, \pi'_k] $ as
	\begin{align}\label{Eqn:pi_and_pi_prime}
	\begin{cases}
	\pi'_i = \pi_i & i \neq t, t+1 \\
	\pi'_t = \pi_{t+1}	 \\
	\pi'_{t+1} = \pi_{t}.
	\end{cases}
	\end{align}
	Use a short-hand notation $\omega'_i$ to represent $\omega_i(\bm{\pi}')$ for $i \in [k]$.
	Note that $\{\omega'_i\}$ satisfies 
	\begin{align}\label{Eqn:omega_and_omega_prime}
	\begin{cases}
	\omega'_i = \omega_i & i \neq t, t+1 \\
	\omega'_{t} = \omega_{t+1} + \beta_c \\
	\omega'_{t+1} = \omega_{t} - \beta_c
	\end{cases}
	\end{align}
	for the following reason.
	First, use simplified notations $a_i$ and $a'_i$ to mean $a_i(\bm{\pi})$ and $a_i(\bm{\pi}')$, respectively. Then, using (\ref{Eqn:a_i}), (\ref{Eqn:pi_relative}) and (\ref{Eqn:pi_and_pi_prime}), we have
	\begin{align}\label{Eqn:a_t_a_t+1}
	a_{t}' &= n_I - 1 - \sum_{j=1}^{t-1}  \mathds{1}_{\pi_j' = \pi_{t}'} = n_I - 1 - \sum_{j=1}^{t-1}  \mathds{1}_{\pi_j = \pi_{t+1}}\nonumber\\
	&= n_I - 1 - \sum_{j=1}^{t}  \mathds{1}_{\pi_j = \pi_{t+1}} = a_{t+1}.
	\end{align}
	Similarly,
	\begin{align} \label{Eqn:a_t+1_a_t}
	a_{t+1}' &= n_I - 1 - \sum_{j=1}^{t}  \mathds{1}_{\pi_j' = \pi_{t+1}'} = n_I - 1 - \sum_{j=1}^{t-1}  \mathds{1}_{\pi_j' = \pi_{t+1}'}\nonumber\\
	&= n_I - 1 - \sum_{j=1}^{t-1}  \mathds{1}_{\pi_j = \pi_{t}} = a_{t}.
	\end{align}
	Therefore, from (\ref{Eqn:omega_i}), (\ref{Eqn:a_t_a_t+1}) and (\ref{Eqn:a_t+1_a_t}), we have
	\begin{align*}
	\omega'_t &= a'_t \beta_I + (n-t-a'_t)\beta_c \\
	& = a_{t+1}\beta_i + ( n-t-a_{t+1} ) \beta_c = \omega_{t+1} + \beta_c.
	\end{align*}
    Similarly, $\omega'_{t+1} = \omega_{t} - \beta_c$ holds.
	This proves (\ref{Eqn:omega_and_omega_prime}).
	Thus, from (\ref{Eqn:omega_relative}) and (\ref{Eqn:omega_and_omega_prime}),
	\begin{equation*}
	L(\bm{s}, \bm{\pi}) = \sum_{i=1}^k \min\{\alpha, \omega_i\} = \sum_{i \notin \{t,t+1\} } \min\{\alpha, \omega_i\} + \omega_t + \alpha,
	\end{equation*}
	\begin{align*}
	L(\bm{s}, \bm{\pi}') &= \sum_{i=1}^k \min\{\alpha, \omega_i'\} \nonumber\\
	&= \sum_{i \notin \{t,t+1\} } \min\{\alpha, \omega_i'\} + \alpha + (\omega_t - \beta_c) \nonumber\\
	&= \sum_{i \notin \{t,t+1\} } \min\{\alpha, \omega_i\} + \alpha + (\omega_t - \beta_c).	
	\end{align*}
	Therefore, 
	$L(\bm{s}, \bm{\pi}) > L(\bm{s}, \bm{\pi}')$
	holds for $\beta_c > 0$. In other words, if $\pi \in \Pi_p^c$, then $\pi \in \Pi_m^c$. This proves that $\Pi_p^c \subseteq \Pi_m^c$ holds for $\beta_c \neq 0$.
	
	Similarly, $\Pi_p^c \subseteq \Pi_r^c$ can be proved as follows. For the pre-defined ordering vectors $\bm{\pi}$ and $\bm{\pi'}$, we have $S_t(\bm{\pi}) = \sum_{i=1}^{t-1}\omega_i + \omega_t$ and  $S_t(\bm{\pi'}) = \sum_{i=1}^{t-1}\omega_i + \omega_{t+1} + \beta_c$. Using (\ref{Eqn:omega_relative}), we have $S_t(\bm{\pi}) < S_t(\bm{\pi'})$, so that $\bm{\pi}$ cannot be a running-sum maximizer. Therefore, $\Pi_p^c \subseteq \Pi_r^c$ holds. 
	
	\textbf{Step 1-2.} Prove that $L(\bm{s, \pi^*}) \leq L(\bm{s, \pi})$, $\forall \bm{\pi}^* \in \Pi_r, \forall \bm{\pi} \in \Pi_p \cap \Pi_r^c$ :
	
	Consider arbitrary $\bm{\pi}^* \in \Pi_r$ and $\bm{\pi} \in \Pi_p \cap \Pi_r^c$. For $i \in [k]$, let $\omega_i^*$ and $\omega_i$ be short-hand notations for $\omega_i(\bm{\pi}^*)$ and $\omega_i(\bm{\pi})$, respectively.
	Note that from Proposition \ref{Prop:weight vector}, 
	\begin{equation}\label{Eqn:omega_invariant}
	\sum_{i=1}^{k} \omega_i = \sum_{i=1}^{k} \omega_i^*.
	\end{equation}
	Let 
	\begin{align}\label{Eqn:t_and_t^*}
	t & =  \max \{i \in [k] : w_i \geq \alpha \}, \nonumber\\
	t^* & = \max \{i \in [k] : w_i^* \geq \alpha \}
	\end{align}
	Then, from (\ref{Eqn:running_sum_maximizer}), 
	\begin{equation}\label{Eqn:omega_running_sum_compare}
	\sum_{i=1}^{t} \omega_i \leq \sum_{i=1}^{t} \omega_i^*.
	\end{equation}
	Combining with (\ref{Eqn:omega_invariant}), we obtain
	\begin{equation}\label{Eqn:omega_compare_result}
	\sum_{i=t+1}^k (\omega_i - \omega_i^*) \geq 0.
	\end{equation}

	Note that from the result of Step 1-1, both $\bm{\pi}$ and $\bm{\pi}^*$ are in $\Pi_p$. Therefore, $\omega_i \geq \alpha$ for $i \in [t]$.  Similarly, $\omega_i^* \geq \alpha$ for $i \in [t^*]$. Therefore, (\ref{Eqn:lower_bound}) can be expressed as
	\begin{align}
	L(\bm{s}, \bm{\pi}) &= \sum_{i=1}^k \min\{\alpha, \omega_i\} = \sum_{i=1}^t \alpha + \sum_{i=t+1}^k \omega_i, \label{Eqn:lower_bound_pi}\\
	L(\bm{s}, \bm{\pi}^*) &= \sum_{i=1}^k \min\{\alpha, \omega_i^*\} = \sum_{i=1}^{t^*} \alpha + \sum_{i=t^*+1}^k \omega_i^*. \label{Eqn:lower_bound_pi^*}
	\end{align}

	If $t = t^*$, then we have 
	\begin{equation*}
	L(\bm{s, \pi}) - L(\bm{s, \pi^*}) = \sum_{i=t+1}^{k} (\omega_i - \omega_i^*)  \geq 0
	\end{equation*}
	from (\ref{Eqn:omega_compare_result}).
	If  $t > t^*$, we get 
	\begin{equation*}
	L(\bm{s, \pi}) - L(\bm{s, \pi^*}) = \sum_{i=t^* + 1}^t (\alpha - \omega_i^*) + \sum_{i=t + 1}^k (\omega_i - \omega_i^*).
	\end{equation*}
	From (\ref{Eqn:t_and_t^*}), $\omega_i^* < \alpha$ holds for $i > t^*$. Therefore, 
	\begin{equation*}
	L(\bm{s, \pi}) - L(\bm{s, \pi^*}) > \sum_{i=t + 1}^k (\omega_i - \omega_i^*) \geq 0
	\end{equation*}
	from (\ref{Eqn:omega_compare_result}).	 	
	In the case of $t < t^*$, define $\Delta = \sum_{i=1}^t (\omega_i - \alpha)$ and $\Delta^* = \sum_{i=1}^t (\omega_i^* - \alpha)$. 
	From (\ref{Eqn:lower_bound_pi}),
	\begin{equation*}
	 L(\bm{s, \pi}) = \sum_{i=1}^t \alpha + \sum_{i=t+1}^k \omega_i = ( \sum_{i=1}^k \omega_i ) - \Delta.
	\end{equation*}
	Similarly, from (\ref{Eqn:lower_bound_pi^*}),
	\begin{align*}
	L(\bm{s, \pi^*}) &=   \sum_{i=1}^k \omega_i^*  - \sum_{i=1}^{t^*} (\omega_i^* - \alpha)\\
	&=  \sum_{i=1}^k \omega_i^*  - \Delta^* - \sum_{i=t+1}^{t^*} (\omega_i^* - \alpha) \leq  \sum_{i=1}^k \omega_i^*  - \Delta^*
	\end{align*}  where the last inequality is from (\ref{Eqn:t_and_t^*}). Combined with (\ref{Eqn:omega_invariant}) and (\ref{Eqn:omega_running_sum_compare}), we obtain 
	\begin{equation*}
	 L(\bm{s, \pi}) - L(\bm{s, \pi^*}) \geq \Delta^* - \Delta \geq 0.
	\end{equation*}
	In summary, $L(\bm{s, \pi^*}) \leq L(\bm{s, \pi})$ irrespective of the $t,t^*$ values, which completes the proof for Step 1-2. 
	From the results of Step 1-1 and Step 1-2, the relationship between the sets can be depicted as in Fig. \ref{Fig:set_relationship}.
	
		\begin{figure}[!t]
		\centering
		\includegraphics[height=30mm]{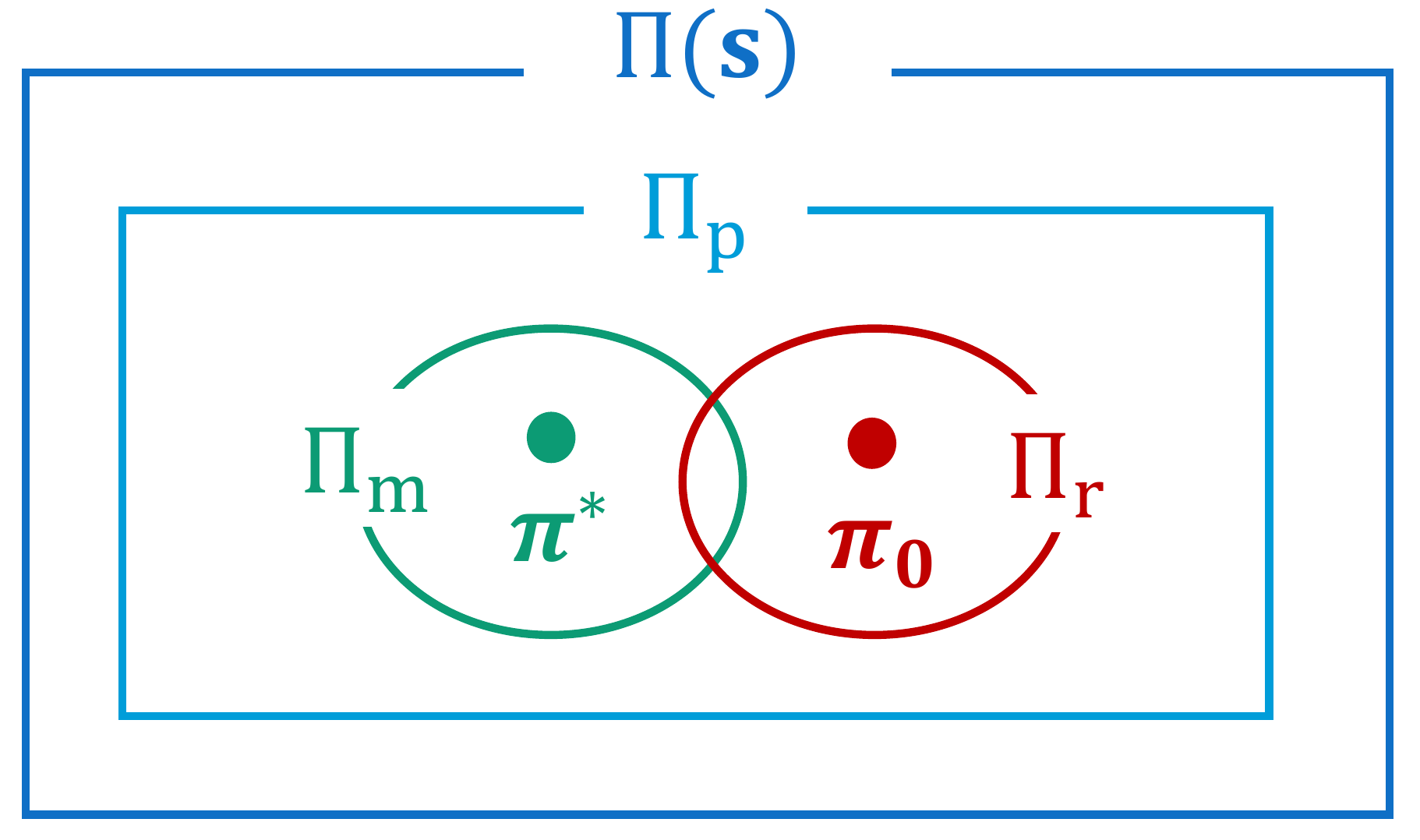}
		\caption{Relationship between sets}
		\label{Fig:set_relationship}
	\end{figure}

	Consider $\bm{\pi}_0 \in \Pi_r$ and $\bm{\pi}^* \in \Pi_m \cap \Pi_r^c$. 
	Then, $L(\bm{s}, \bm{\pi}_0) \leq L(\bm{s}, \bm{\pi}^*)$ from the result of Step 1-2. Based on the definition of $\Pi_m$ in (\ref{Eqn:min_cut_minimizer}), we can write
$	L(\bm{s}, \bm{\pi_0}) \leq L(\bm{s}, \bm{\pi})$ for every $\bm{\pi} \in \Pi(\bm{s}).$
	In other words,
$	\pi_0 \in \Pi_m$
	holds for arbitrary $\pi_0 \in \Pi_r$.
	Therefore, 
	$\Pi_r \subseteq \Pi_m$ holds.

	\textbf{Step 2.} Prove $\bm{\pi}_v \in \Pi_r$:
	
	\begin{figure}[!t]
		\centering
		\includegraphics[height=40mm]{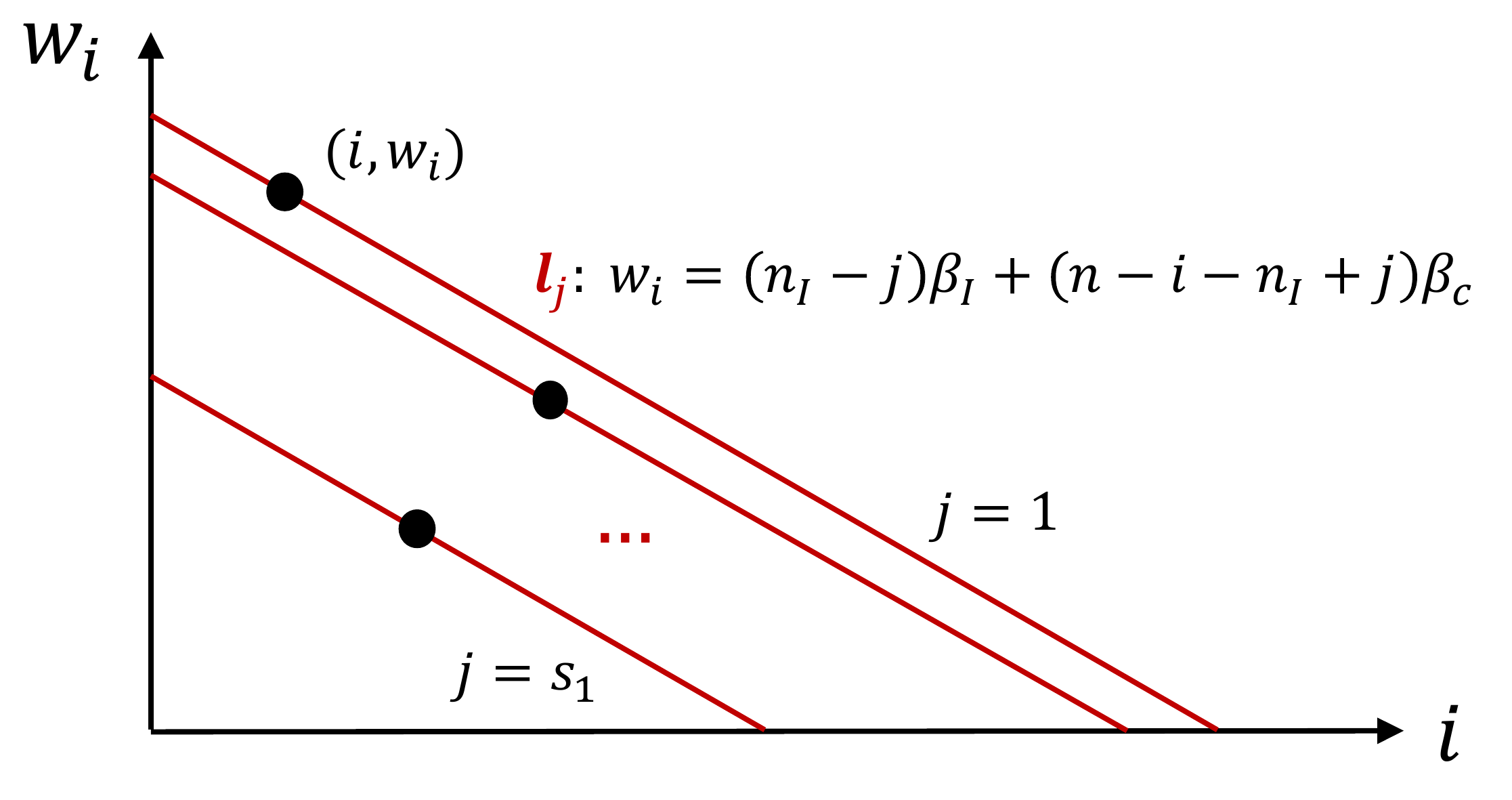}
		\caption{Set of $s_1$ lines where $(i,\omega_i)$ points can position}
		\label{Fig:omega_I_line}
	\end{figure} 
	
	For a given selection vector $\bm{s}=[s_1, \cdots, s_L]$, consider an arbitrary ordering vector $\bm{\pi} = [\pi_1, \cdots, \pi_k] \in \Pi(\bm{s})$.
	The corresponding $\omega_i(\bm{\pi})$ defined in (\ref{Eqn:weight vector}) is written as
	\begin{equation}\label{Eqn:omega_i_notation_j}
	w_i = (n_I - j) \beta_I + (n-i-n_I + j)\beta_c
	\end{equation}
	where $j = \sum_{t=1}^{i} \mathds{1}_{\pi_t = \pi_i} $. 	
	
	Consider a set of lines $\{l_j\}_{j=1}^{n_I}$, where line $l_j$ represents an equation: $w_i = (n_I - j) \beta_I + (n-i-n_I + j)\beta_c$. Since we assume $\beta_I \geq \beta_c$, these lines can be illustrated as in Fig. \ref{Fig:omega_I_line}. 
	For a given $\bm{\pi}$, consider marking a $(i,\omega_i)$ point for $i \in [k]$.
	Note that the $(i, \omega_i)$ point is on line $l_j$ if and only if 
	\begin{equation}\label{Eqn:l_j_condition}
	j = \sum_{t=1}^{i} \mathds{1}_{\pi_t = \pi_i}
	\end{equation}
	where the summation term in (\ref{Eqn:l_j_condition}) represents the number of occurrence of $\pi_i$ value in $\{\pi_t\}_{t=1}^{i}$. For the example in Fig. \ref{Fig:vertical ordering}, when $\bm{s} = [4, 3, 1]$ and $\bm{\pi} = [1, 2, 3, 1, 2, 1, 2, 1] \in \Pi(\bm{s})$, line $l_3$ contains the point $(i,\omega_i)=(6,\omega_6)$ since $3 = \sum_{t=1}^{6} \mathds{1}_{\pi_t = \pi_6}$.
	
	Recall that 
$	I_l (\bm{\pi}) = \{ i \in [k] : \pi_i = l \}$,
	as defined in (\ref{Eqn:I_l}), where $|I_l (\bm{\pi})| = s_l$ holds $\forall l \in [L]$.
	For $j \in [n_I]$, consider $l \in [L]$ with $s_l \geq j$. Let $i_0$ be the $j^{th}$ smallest element in $I_l (\bm{\pi})$. Then, 
$	j = \sum_{t=1}^{i_0} \mathds{1}_{\pi_t = l}$ 
and $\pi_{i_0} = l$
	hold. Thus, the $(i_0, \omega_{i_0})$ point is on line $l_j$. Similarly, we can find 
	\begin{equation} \label{Eqn:point vector components}
	p_j = \sum_{l=1}^{L} \mathds{1}_{s_l \geq j}
	\end{equation}
	points on line $l_j$, irrespective of the ordering vector $\bm{\pi} \in \Pi(\bm{s})$.
	Note that 
	\begin{equation}\label{Eqn:point vector property1}
	\sum_{j=1}^{n_I} p_j = \sum_{l=1}^{L} \sum_{j=1}^{n_I} \mathds{1}_{s_l \geq j} = \sum_{l=1}^{L} s_l = k,
	\end{equation} 
	which confirms that Fig. \ref{Fig:omega_I_line} contains $k$ points.
	Moreover, 
	\begin{equation}\label{Eqn:point vector property2}
	\forall j \in [n_I-1], \ p_j \geq p_{j+1}
	\end{equation}
	holds from the definition in (\ref{Eqn:point vector components}).
	
	\begin{figure}[!t]
		\centering
		\includegraphics[height=42mm]{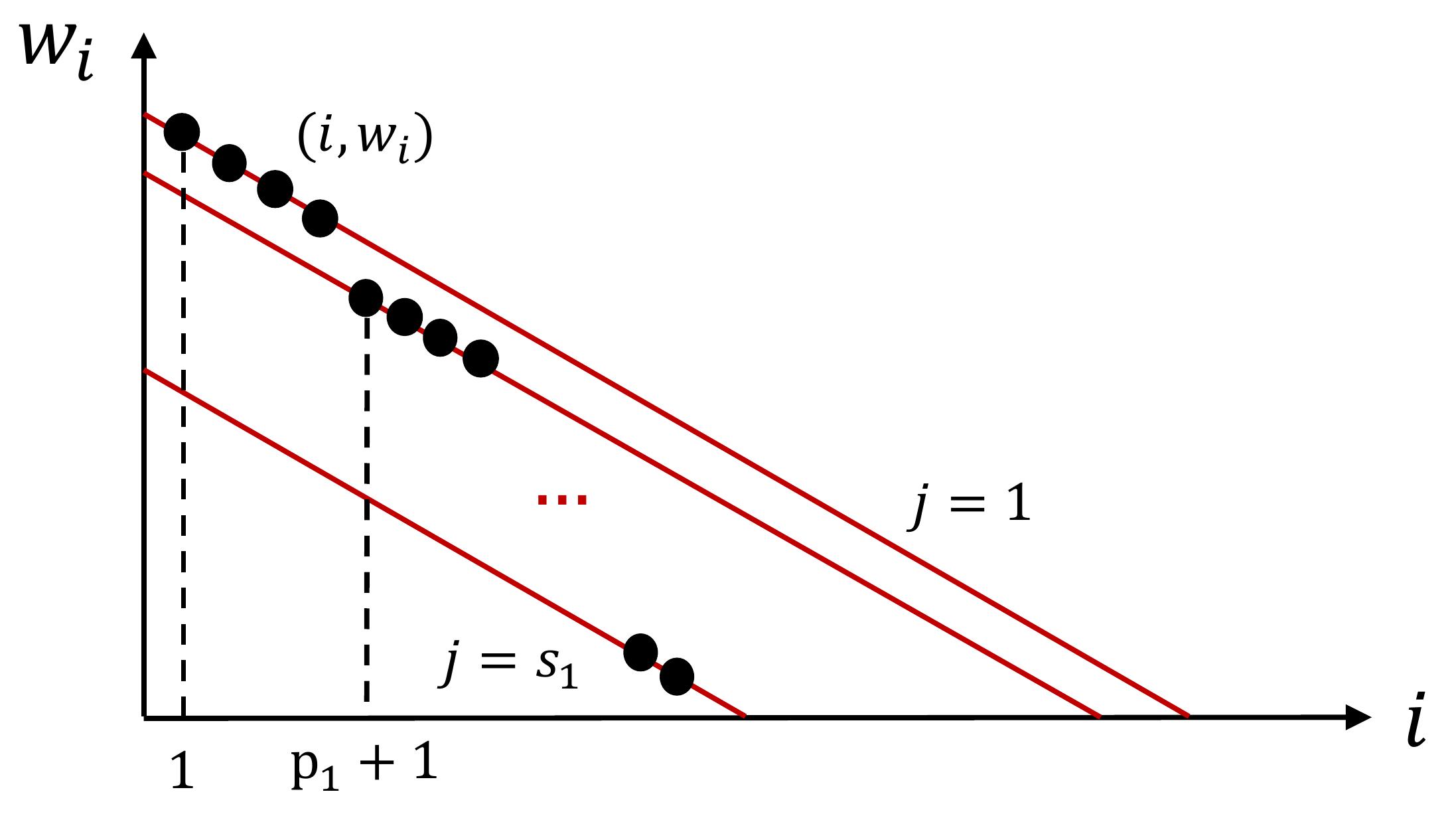}
		\caption{Optimal packing of $k$ points}
		\label{Fig:omega_I_line_packing}
	\end{figure} 
	
	In order to maximize the running sum $S_t(\bm{\pi}) = \sum_{i=1}^{t} w_i(\bm{\pi})$ for every $t$, the optimal ordering vector packs $p_1$ points in the leftmost area ($i=1, \cdots, p_1$), pack $p_2$ points in the leftmost remaining area ($i=p_1 + 1, \cdots, p_1 + p_2$), and so on. This packing method corresponds to Fig. \ref{Fig:omega_I_line_packing}.

	Note that from the definition of $p_j$ in (\ref{Eqn:point vector components}) and Fig. \ref{Fig:vertical ordering}, vertical ordering $\bm{\pi}_v$ in Definition \ref{Def:vertical ordering vector} first chooses $p_1$ points on line $l_1$, then chooses $p_2$ points on line 
	$l_2$, and so on. Thus, $\bm{\pi}_v$ achieves optimal packing in Fig. \ref{Fig:omega_I_line_packing}, which maximizes the running sum $S_t(\bm{\pi})$. 
	Therefore, vertical ordering is a running sum maximizer, i.e., $\bm{\pi}_v \in \Pi_r$.
	Combining Steps 1 and 2, we conclude that $\bm{\pi}_v$ minimizes $L  (\bm{s},\bm{\pi})$ among $\bm{\pi} \in \Pi(\bm{s})$ for arbitrary $\bm{s} \in S$.
\end{proof}


Now, we define a special selection vector called the \textit{horizontal selection vector}, which is shown to be the optimal selection vector which minimizes $L  (\bm{s},\bm{\pi}_v)$.
\begin{definition}\label{Def:horizontal selection vector}
	The horizontal selection vector $\bm{s}_h = [s_1, \cdots, s_L] \in \mathcal{S}$ is defined as:
	\begin{equation}
	s_i =
	\begin{cases}
	n_I,    &    i \leq \floor{\frac{k}{n_I}} \nonumber\\
	(k \Mod{n_I}),    & i = \floor{\frac{k}{n_I}} + 1 \nonumber\\
	0		&     i > \floor{\frac{k}{n_I}} + 1.
	\end{cases}
	\end{equation} 
\end{definition}

\begin{figure}[t]
	\centering
	\includegraphics[height=25mm]{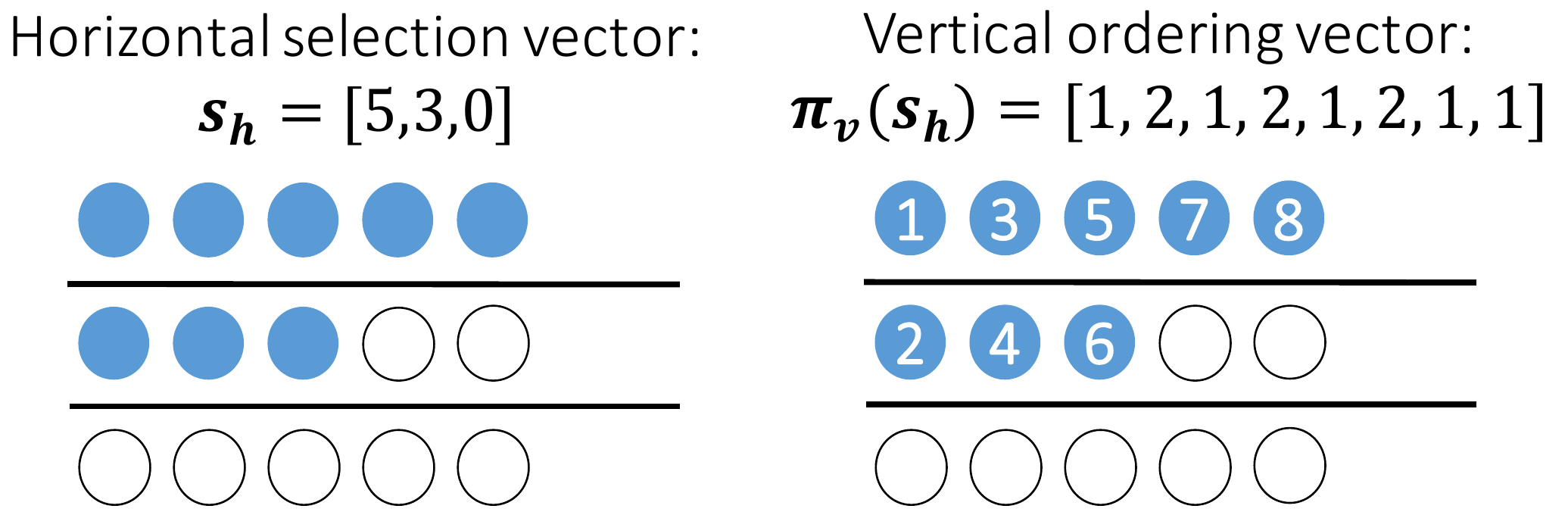}
	\caption{The optimal selection vector $\bm{s}_h$  and the optimal ordering vector $\bm{\pi}_v $ (for $n=15, k=8, L=3$ case)}
	\label{Fig:horizontal selection}
\end{figure}

The graphical illustration of the horizontal selection vector is on the left side of Fig. \ref{Fig:horizontal selection}, in the case of $n=15, k=8, L=3$.
The following Lemma states that the horizontal selection vector minimizes $L  (\bm{s},\bm{\pi}_v)$. 

\begin{lemma} \label{Lemma:optimal selection vector}
	Consider applying the vertical ordering vector $\bm{\pi}_v$.
	Then, the horizontal selection vector $\bm{s}_h$ minimizes the lower bound $L  (\bm{s},\bm{\pi})$ on the min-cut.
	In other words, $L(\bm{s}_h, \bm{\pi}_v) \leq L(\bm{s}, \bm{\pi}_v ) \ \forall \bm{s}\in \mathcal{S}$.
\end{lemma}


\begin{proof}
	From the proof of Lemma \ref{Lemma:optimal ordering vector}, the optimal ordering vector turns out to be the vertical ordering vector, where the corresponding $\omega_i(\bm{\pi}_v)$ sequence is illustrated in Fig. \ref{Fig:omega_I_line_packing}. Depending on the selection vector $\bm{s} = [s_1, \cdots, s_L]$, the number $p_j$ of points on each line $l_j$ changes.
	
	Consider an arbitrary selection vector $\bm{s}$. Define a point vector $\bm{p(s)}=[p_1,\cdots, p_{n_I}]$ where $p_j$ is the number of points on $l_j$, as defined in (\ref{Eqn:point vector components}). Similarly, define $\bm{p(s_h)}=[p_1^*,\cdots, p_{n_I}^*]$. Using Definition \ref{Def:horizontal selection vector} and (\ref{Eqn:point vector components}), we have 
	\begin{equation} \label{Eqn:point vector of horizontal selection vector}
	p_j^* = 
	\begin{cases}
	\floor{\frac{k}{n_I}} + 1,	& \	j \leq (k \Mod{n_I})\\
	\floor{\frac{k}{n_I}}	&	\ otherwise
	\end{cases}
	\end{equation}
	Now, we prove 
	\begin{equation}\label{Eqn:Lemma_2_proof_1}
	\forall t \in [n_I], \ \sum_{j=1}^t p_j^* \leq \sum_{j=1}^t p_j.
	\end{equation}
	The proof is divided into two steps: \textit{base case} and \textit{inductive step}.	
	
	\textbf{Base Case: } We wish to prove that $p_1^* \leq p_1$. Suppose $p_1^* > p_1$ (i.e., $p_1 \leq p_1^* - 1$). Then, 
	\begin{equation}\label{Eqn:base_case_eqn1}
	\sum_{l=1}^{n_I}p_l \leq \sum_{l=1}^{n_I}p_1 \leq \sum_{l=1}^{n_I} (p_1^* - 1)
	\end{equation}
	where the first inequality is from (\ref{Eqn:point vector property2}).
	Note that if $(k \Mod{n_I}) = 0$, then
	\begin{equation}\label{Eqn:base_case_eqn2}
		\sum_{l=1}^{n_I} p_l^* = \sum_{l=1}^{n_I} p_1^* > \sum_{l=1}^{n_I} (p_1^* - 1).
	\end{equation}
	Otherwise, 
	\begin{align}
		\sum_{l=1}^{n_I} p_l^* &= \sum_{l=1}^{(k \Mod{n_I})} p_l^*  + \sum_{l=(k \Mod{n_I}) + 1}^{n_I} p_l^*  \nonumber\\
		&= \sum_{l=1}^{(k \Mod{n_I})} (\floor{\frac{k}{n_I}} + 1)  + \sum_{l=(k \Mod{n_I}) + 1}^{n_I} \floor{\frac{k}{n_I}} \nonumber\\
		&> \sum_{l=1}^{n_I} \floor{\frac{k}{n_I}} = \sum_{l=1}^{n_I} (p_1^* - 1).\label{Eqn:base_case_eqn3}
	\end{align}
	Therefore, combining (\ref{Eqn:base_case_eqn1}), (\ref{Eqn:base_case_eqn2}) and (\ref{Eqn:base_case_eqn3}) results in
	$\sum_{l=1}^{n_I}p_l < \sum_{l=1}^{n_I} p_l^* = k$,
	which contradicts (\ref{Eqn:point vector property1}). Therefore, $p_1^* \leq p_1$ holds.

	\textbf{Inductive Step: } Assume that $\sum_{l=1}^{l_0} p_l^* \leq \sum_{l=1}^{l_0} p_l$ for arbitrary $l_o \in [n_I-1]$. Now we prove that $\sum_{l=1}^{l_0+1} p_l^* \leq \sum_{l=1}^{l_0+1} p_l$ holds.
	Suppose not. Then, 
	\begin{equation}\label{Eqn:point_vector_inequality}
	p_{l_0 + 1}^* - \Theta - 1 \geq p_{l_0 + 1}
	\end{equation}
	holds where 
	\begin{equation}\label{Eqn:inductive_step_eqn2}
	\Theta = \sum_{l=1}^{l_0} (p_l - p_l^*).
	\end{equation} 
	Using (\ref{Eqn:point vector property2}) and (\ref{Eqn:point_vector_inequality}), we have
	\begin{align}
	\sum_{l=l_0 + 1}^{n_I}p_l &\leq \sum_{l=l_0 + 1}^{n_I}p_{l_0 + 1} \leq \sum_{l=l_0 + 1}^{n_I} (p_{l_0 + 1}^* - 1 - \Theta )\nonumber\\
	&\leq \sum_{l=l_0 + 1}^{n_I} (p_{l_0 + 1}^* - 1) - \Theta \label{Eqn:inductive_step_eqn1}
	\end{align}
	where equality holds for the last inequality iff $l_0 = n_I - 1$.
	Using analysis similar to (\ref{Eqn:base_case_eqn2}) and (\ref{Eqn:base_case_eqn3}) for the \textit{base case}, we can find that 
$	\sum_{l=l_0 + 1}^{n_I} (p_{l_0 + 1}^* - 1) < \sum_{l=l_0 + 1}^{n_I} p_l^*.$
Combining with (\ref{Eqn:inductive_step_eqn1}), we get
	\begin{equation}\label{Eqn:inductive_step_eqn3}
		\sum_{l=l_0+1}^{n_I} p_l < \sum_{l=l_0 + 1}^{n_I} p_l^* - \Theta.
	\end{equation}
	Equations (\ref{Eqn:inductive_step_eqn2}) and (\ref{Eqn:inductive_step_eqn3}) imply  
$		\sum_{l=1}^{n_I}p_l < \sum_{l=1}^{n_I} p_l^* = k,$
	which contradicts (\ref{Eqn:point vector property1}). Therefore, (\ref{Eqn:Lemma_2_proof_1}) holds.

	Now define
	\begin{align}
	f_i &= \min \{s \in [n_I]: \sum_{l=1}^s p_l \geq i \} \label{Eqn:t_i}\\
	h_i &= \min \{s \in [n_I]: \sum_{l=1}^s p_l^* \geq i \} \label{Eqn:t_i^*}
	\end{align}
	for $i\in[k]$. 	
	Then,
	\begin{equation}\label{Eqn:Lemma_2_proof_2}
	\forall i \in [k], \ h_i \geq f_i
	\end{equation}
	holds directly from (\ref{Eqn:Lemma_2_proof_1}).
	Note that since $p_i^*$ in (\ref{Eqn:point vector of horizontal selection vector}) is identical to $g_i$ in (\ref{Eqn:g_m}), $h_i$ can be written as
	\begin{equation}
	h_i =  \min \{s \in [n_I]: \sum_{l=1}^s g_l \geq i \} \label{Eqn:t_i^*_reform}
	\end{equation}
	Consider $\bm{\pi}_v(\bm{s})$, the vertical ordering vector for a given selection vector $\bm{s}$. 
	Recall that as in Fig. \ref{Fig:omega_I_line_packing}, vertical ordering packs the leftmost $p_1$ points on line $l_1$, packs the next $p_2$ points on line $l_2$, and so on. Using (\ref{Eqn:omega_i_notation_j}), we can write 
	\begin{align*}
	\omega_i = 
	\begin{cases}
	(n_I-1)\beta_I + (n-i-n_I+1)\beta_c, \quad  & \text{if } i \in [p_1] \\
	(n_I-2)\beta_I + (n-i-n_I+2)\beta_c, \quad  & \text{if } i - p_1 \in [p_2] \\
	\quad \quad \vdots  \\
	0 \cdot \beta_I + (n-i)\beta_c, & \text{if } i - \sum_{t=1}^{n_I-1}p_t \\
	& \quad \quad \in [p_{n_I}]  \end{cases}
	\end{align*}
	Therefore, using (\ref{Eqn:t_i}), we further write
	\begin{equation}\label{Eqn:Lemma_2_proof_3}
	\omega_i(\bm{\pi}_v(\bm{s})) = (n_I - f_i) \beta_I + (n-i-n_I + f_i) \beta_c.
	\end{equation}
	Similarly, 
	\begin{equation}\label{Eqn:Lemma_2_proof_4}
	\omega_i(\bm{\pi}_v(\bm{s}_h)) = (n_I - h_i) \beta_I + (n-i-n_I + h_i) \beta_c.
	\end{equation}
	Combining (\ref{Eqn:Lemma_2_proof_2}), (\ref{Eqn:Lemma_2_proof_3}) and (\ref{Eqn:Lemma_2_proof_4}), we have
	\begin{equation*}
	\omega_i(\bm{\pi}_v(\bm{s}_h)) \leq \omega_i(\bm{\pi}_v(\bm{s})) \ \forall i=1,\cdots,k, \forall \bm{s} \in \mathcal{S},
	\end{equation*}
	since we assume $\beta_I \geq \beta_c$.
	Therefore, combining with (\ref{Eqn:lower_bound}), we conclude that $L(\bm{s}_h, \bm{\pi}_v) \leq L(\bm{s}, \bm{\pi}_v)$ for arbitrary $\bm{s} \in \mathcal{S}$, which completes the proof of Lemma \ref{Lemma:optimal selection vector}.
\end{proof}

From Lemmas \ref{Lemma:optimal ordering vector} and \ref{Lemma:optimal selection vector}, we have 
\begin{equation*}
\forall \bm{s} \in \mathcal{S}, \forall \bm{\pi} \in \Pi(\bm{s}), \ L(\bm{s}_h, \bm{\pi}_v) \leq L(\bm{s}, \bm{\pi}).
\end{equation*}
All that remains is to compute $L(\bm{s}_h, \bm{\pi}_v)$ and check that it is identical to (\ref{Eqn:Capacity_value}).

From (\ref{Eqn:lower_bound}), $L(\bm{s}_h, \bm{\pi}_v)$ can be written as 
\begin{equation}\label{Eqn:lower_bound_optimal}
L(\bm{s}_h, \bm{\pi}_v) = \sum_{i=1}^k \min \{\alpha, \omega_i(\bm{\pi}_v(\bm{s}_h) )\}
\end{equation}
where $\omega_i(\bm{\pi}_v(\bm{s}_h) )$
is defined in (\ref{Eqn:Lemma_2_proof_4}).
From $h_i$ in (\ref{Eqn:t_i^*_reform}), we have
\begin{equation}\label{Eqn:t_i^*_cases}
h_i = 
\begin{cases}
1,  &  i \in [g_1] \\
2, & i - g_1 \in [g_2] \\
& \vdots \\
n_I, & i - \sum_{t=1}^{n_I-1}g_t \in [g_{n_I}]  
\end{cases}
\end{equation}
If we define
\begin{equation*}
I_m^* = \{i \in [k] : h_i = m \},
\end{equation*}
then $L(\bm{s}_h, \bm{\pi}_v)$ in (\ref{Eqn:lower_bound_optimal}) can be expressed as
\begin{align}\label{Eqn:lower_bound_intermediate}
& L(\bm{s}_h, \bm{\pi}_v) = \sum_{i=1}^k \min\{\alpha, (n_I-h_i)\beta_I + (n-n_I-i+h_i)\beta_c\} \nonumber\\
&= \sum_{m=1}^{n_I} \sum_{i \in I_m^*} \min \{\alpha, (n_I-m)\beta_I + (n-n_I-i+m)\beta_c\}\nonumber\\
&= \sum_{m=1}^{n_I} \sum_{i \in I_m^*} \min \{\alpha, \rho_m\beta_I + (n-\rho_m-i)\beta_c\}
\end{align}
where $\rho_m$ is defined in (\ref{Eqn:rho_i}).

Using (\ref{Eqn:t_i^*_cases}), we have 
\begin{equation*}
I_m^* = \{\sum_{t=1}^{m-1}g_t + 1, \cdots, \sum_{t=1}^{m-1}g_t + g_m\}.
\end{equation*}
for $m \in [n_I]$. Therefore,  $i \in I_m^*$ can be represented as 
\begin{equation*}
i = \sum_{t=1}^{m-1} g_t + l
\end{equation*} 
for $l \in [g_m]$. 
Using this notation, (\ref{Eqn:lower_bound_intermediate}) can be written as
\begin{align}\label{Eqn:capacity_final}
L(\bm{s}_h, \bm{\pi}_v) = \sum_{m=1}^{n_I} \sum_{l=1}^{g_m} \min \{\alpha, \rho_m\beta_I + (n-\rho_m - s_m^{(l)})\beta_c\}
\end{align}
where $s_m^{(l)} = \sum_{t=1}^{m-1}g_t + l$. This expression reduces to 
(\ref{Eqn:Capacity_value}).
This completes the proof of Part II-2.
Therefore, the storage capacity of clustered DSS is as stated in Theorem \ref{Thm:Capacity of clustered DSS}.

\section{Proof of Theorem \ref{Thm:condition_for_min_storage}}
\label{Section:proof_of_thm_condition_for_min_storage}

We begin with introducing properties of the parameters $z_t$ and $h_t$ defined in (\ref{Eqn:z_t}) and (\ref{Eqn:h_t}).
\begin{prop}
	\begin{align} 
	h_k &= n_I, \label{Eqn:h_k} \\
	z_k &= (n-k)\epsilon. \label{Eqn:z_k}
	\end{align}
\end{prop}
\begin{proof}
	Since we consider $k > n_I$ case as stated in (\ref{Eqn:k_constraint}), we have 
	\begin{equation*}
	g_i \geq 1 \quad \forall i \in [n_I]
	\end{equation*}
	for $\{g_i\}_{i=1}^{n_I}$ defined in (\ref{Eqn:g_m}).
	Combining with (\ref{Eqn:sum of g is k}) and (\ref{Eqn:h_t}), we can conclude that $h_k = n_I$. Finally, $z_k = (n-k)\epsilon$ is from (\ref{Eqn:z_t}) and (\ref{Eqn:h_k}).
\end{proof}

First, consider the $\epsilon \geq \frac{1}{n-k}$ case. From (\ref{Eqn:Feasible Points Result}), data $\mathcal{M}$ can be reliably stored with node storage $\alpha = \mathcal{M}/k$ if the repair bandwidth satisfies $\gamma \geq \gamma^*$, where
\begin{align*}
\gamma^* &= \frac{\mathcal{M} - (k-1)\mathcal{M}/k}{s_{k-1}} = \frac{\mathcal{M}}{k} \frac{1}{s_{k-1}}\nonumber\\
&= \frac{\mathcal{M}}{k} \frac{(n-k)\epsilon}{(n_I-1) + (n-n_I)\epsilon}
\end{align*}
where the last equality is from (\ref{Eqn:s_t}) and (\ref{Eqn:z_k}).
Thus, $\alpha = \mathcal{M}/k$ is achievable with finite $\gamma$, when $\epsilon \geq \frac{1}{n-k}$.

Second, we prove that it is impossible to achieve $\alpha = \mathcal{M}/k$ for $0 \leq \epsilon < \frac{1}{n-k}$, in order to reliably store file $\mathcal{M}$. 
Recall that the minimum storage for $0 \leq \epsilon < \frac{1}{n-k}$ is 
\begin{equation}\label{Eqn:alpha_MSR}
\alpha = \frac{M}{\tau + \sum_{i=\tau+1}^{k}z_i}
\end{equation} 
from \eqref{Eqn:MSR_point}.
From (\ref{Eqn:tau}) and (\ref{Eqn:z_k}), we have $z_i < 1$ for $i=\tau+1, \tau+2, \cdots, k$.
Therefore, 
\begin{equation*}
\tau + \sum_{i=\tau+1}^{k}z_i < \tau + (k-(\tau+1)+1) = k
\end{equation*}
holds, which result in 
\begin{equation}
\frac{M}{\tau + \sum_{i=\tau+1}^{k}z_i} > \frac{\mathcal{M}}{k}. 
\end{equation}
Thus, the $0 \leq \epsilon \leq \frac{1}{n-k}$ case has the minimum storage $\alpha$ greater than $\mathcal{M}/k$, which completes the proof of Theorem \ref{Thm:condition_for_min_storage}.

\section{Proof of Theorem \ref{Thm:MSR_MBR}}\label{Proof_of_Thm_MSR_MBR}

First, we prove (\ref{Eqn:msr_property}). To begin with, we obtain the expression of $\alpha_{msr}^{(\epsilon)}$, for $\epsilon = 0, 1$.
From (\ref{Eqn:MSR_point}), we obtain
\begin{align}
\alpha_{msr}^{(0)} &= \frac{\mathcal{M}}{\tau + \sum_{i=\tau+1}^{k}z_i}, \label{Eqn:msr_epsilon_0_alpha} \\
\alpha_{msr}^{(1)} &= \frac{\mathcal{M}}{k} \nonumber.
\end{align}
We further simplify the expression for $\alpha_{msr}^{(0)}$ as follows. Recall
\begin{equation}\label{Eqn:z_t_epsilon0}
z_t = n_I - h_t
\end{equation} 
for $t \in [k]$ from \eqref{Eqn:z_t}, when $\epsilon=0$ holds. Note that we have
\begin{equation}\label{Eqn:z_t_epsilon0_cases}
z_t =
\begin{cases}
0, & t \geq k- \floor{\frac{k}{n_I}} + 1 \\
1, & t = k- \floor{\frac{k}{n_I}}
\end{cases}
\end{equation}
from the following reason. First, from \eqref{Eqn:h_t} and \eqref{Eqn:z_t_epsilon0}, $z_t = 0$ holds for \begin{align*}
t &\geq \sum_{l=1}^{n_I-1} g_l + 1 = \sum_{l=1}^{n_I} g_l - g_{n_I} + 1 = k -  \floor{\frac{k}{n_I}} + 1
\end{align*}
where the last equality is from \eqref{Eqn:sum of g is k} and \eqref{Eqn:g_m}. Similarly, we can prove that $z_t = 1$ holds for $t = k - \floor{\frac{k}{n_I}}$.
From \eqref{Eqn:z_t_epsilon0_cases} and \eqref{Eqn:tau}, we obtain 
\begin{equation}\label{Eqn:tau_epsilon0}
\tau = k - \floor{\frac{k}{n_I}}
\end{equation}
when $\epsilon=0$. Combining \eqref{Eqn:msr_epsilon_0_alpha}, \eqref{Eqn:z_t_epsilon0_cases} and \eqref{Eqn:tau_epsilon0}, we have 
\begin{equation}
\alpha_{msr}^{(0)} = \frac{\mathcal{M}}{k - \floor{\frac{k}{n_I}}}.
\end{equation}
Then, using $R=k/n$ and $\sigma = L^2/n$,
\begin{align*}
\frac{\alpha_{msr}^{(0)}}{\alpha_{msr}^{(1)}} &= \frac{k}{k-\floor{\frac{k}{n_I}}} = \frac{nR}{nR - \floor{RL}} \nonumber\\
&= \frac{nR}{nR - \floor{R\sqrt{n\sigma}}} = \frac{R}{R - \floor{R\sqrt{n\sigma}}/n}.
\end{align*}
Thus, for arbitrary fixed $R$ and $\sigma$, 
\begin{equation*}
\lim\limits_{n \rightarrow \infty} \frac{\alpha_{msr}^{(0)}}{\alpha_{msr}^{(1)}} = 1.
\end{equation*}
Therefore, $\alpha_{msr}^{(0)}$ is asymptotically equivalent to $\alpha_{msr}^{(1)}$.

Second, we prove (\ref{Eqn:mbr_property}). Note that from (\ref{Eqn:MBR_point}), $\alpha_{mbr}^{(\epsilon)} = \gamma_{mbr}^{(\epsilon)}$ holds for arbitrary $0 \leq \epsilon \leq 1$. Therefore, all we need to prove is
\begin{equation*}\label{Eqn:gamma_asymptote}
\gamma_{mbr}^{(0)} \rightarrow \gamma_{mbr}^{(1)}.
\end{equation*}
 To begin with, we obtain the expression for $\gamma_{mbr}^{(\epsilon)}$, when $\epsilon = 0, 1$.
For $\epsilon = 1$, $z_t$ in (\ref{Eqn:z_t}) is 
\begin{equation}\label{Eqn:z_t_kappa_1}
z_t = n-t
\end{equation} for $t \in [k]$. Moreover, from (\ref{Eqn:s_t}), 
\begin{equation}\label{Eqn:s_0}
s_0 = \frac{\sum_{i=1}^{k} z_i}{n-1}
\end{equation}
for $\epsilon = 1$.
Therefore, from (\ref{Eqn:MBR_point}), (\ref{Eqn:z_t_kappa_1}) and (\ref{Eqn:s_0}),
\begin{align}
\gamma_{mbr}^{(1)}=\frac{\mathcal{M}}{s_0} &= \frac{(n-1)\mathcal{M}}{\sum_{i=1}^{k} (n-i)} = \frac{\mathcal{M}}{k} \frac{2(n-1)}{2n-k-1}. \label{Eqn:gamma_mbr_1}
\end{align}

Now we focus on the case of $\epsilon=0$. First, let $q$ and $r$ be
\begin{align}
q &\coloneqq \floor{\frac{k}{n_I}}, \label{Eqn:quotient}\\
r &\coloneqq (k \Mod{n_I})	, \label{Eqn:remainder}
\end{align}
which represent the quotient and remainder of $k/n_I$. 
Note that 
\begin{equation}\label{Eqn:quotient_remainder_relation}
qn_I + r = k.
\end{equation} 
Then, we have
\begin{align}\label{Eqn:sum of weighted g}
\sum_{t=1}^{n_I} t g_t &= \sum_{t=1}^{r} (q+1) t + \sum_{t=r+1}^{n_I} qt = q \sum_{t=1}^{n_I} t + \sum_{t=1}^{r} t \nonumber\\
&= q \frac{n_I (n_I+1)}{2} +  \frac{r (r+1)}{2} = \frac{1}{2}(qn_I^2 + r^2 + k)
\end{align}
where the last equality is from \eqref{Eqn:quotient_remainder_relation}.
From \eqref{Eqn:z_t_epsilon0} and \eqref{Eqn:t_i^*_cases}, we have
\begin{align}\label{Eqn:sum_of_z}
\sum_{i=1}^k z_i &= \sum_{t=1}^{n_I} (n_I-t) g_t = n_I k - \frac{1}{2} (qn_I^2 + r^2 + k) \nonumber\\
&= (n_I-1) k - \frac{1}{2} (qn_I^2 + r^2 - k)
\end{align}
where the second last equality is from \eqref{Eqn:sum of g is k} and \eqref{Eqn:sum of weighted g}.
Moreover, using \eqref{Eqn:quotient_remainder_relation}, we have
\begin{align}\label{Eqn:bound}
qn_I^2 + r^2 - k &= qn_I^2 + r^2 - qn_I - r \nonumber\\
&= (qn_I + r) (n_I-1) - r(n_I-r) \nonumber\\
&\leq (qn_I+r)(n_I-1) = k(n_I-1)
\end{align}
where the equality holds if and only if $r=0$.
Furthermore, 
\begin{equation}\label{Eqn:s_0_0}
s_0 = \frac{\sum_{i=1}^{k} z_i}{n_I-1}
\end{equation}
for $\epsilon = 0$ from (\ref{Eqn:s_t}). Combining \eqref{Eqn:MBR_point}, \eqref{Eqn:sum_of_z}, \eqref{Eqn:bound} and \eqref{Eqn:s_0_0} result in 
\begin{align}
\gamma_{mbr}^{(0)}=\frac{\mathcal{M}}{s_0} &\leq  \frac{2\mathcal{M}}{k}. \label{Eqn:gamma_mbr_0}
\end{align}

From (\ref{Eqn:gamma_mbr_1}) and (\ref{Eqn:gamma_mbr_0}), we have
\begin{align*}\label{Eqn:gamma_ratio}
\gamma_{mbr}^{(0)} - \gamma_{mbr}^{(1)} 
& \leq \frac{\mathcal{M}}{k} \left( 2 - \frac{2(n-1)}{2n-k-1} \right) = \frac{\mathcal{M}}{k} \frac{2(n-k)}{2n-k-1} \nonumber\\
&= \frac{\mathcal{M}}{nR} \frac{2n(1-R)}{n(2-R)-1}
\end{align*}
where $R=k/n$.
Thus, for arbitrary $n$, 
\begin{equation*}
\gamma_{mbr}^{(0)} \rightarrow \gamma_{mbr}^{(1)}
\end{equation*}
as $R \rightarrow 1$. This completes the proof of (\ref{Eqn:mbr_property}).
Finally, (\ref{Eqn:gamma_mbr_1}) and (\ref{Eqn:gamma_mbr_0}) provides
\begin{align*}
\frac{\gamma_{mbr}^{(0)}}{\gamma_{mbr}^{(1)}} &\leq \frac{2n-k-1}{n-1} = 2- \frac{k-1}{n-1}, 
\end{align*}
which completes the proof of (\ref{Eqn:mbr_ratio}).

%
%
%
%
%
%
%

\section{Proof of Theorem \ref{Thm:IRC_versus_LRC}}\label{Section:Proof_of_IRC_versus_LRC}
\cmt{Recall the definition of an $(n, l_0, m_0, \mathcal{M}, \alpha)$-LRC which appear right before Theorem \ref{Thm:IRC_versus_LRC}.}
Moreover, recall that the \textit{repair locality} of a code is defined as the number of nodes 
to be contacted in the node repair process \cite{papailiopoulos2014locally}. 
Since each cluster contains $n_I$ nodes, every node in a DSS with $\epsilon = 0$ has the repair locality of 
\begin{equation}\label{Eqn:locality_val}
l_0 = n_I - 1.
\end{equation}
Moreover, note that for any code with minimum distance $m$, the original file $\mathcal{M}$ can be retrieved by contacting $n-m+1$ coded symbols \cite{moon2005error}. Since the present paper considers DSSs such that contacting any $k$ nodes can retrieve the original file $\mathcal{M}$, we have the minimum distance of
\begin{equation}\label{Eqn:min_dist_val}
m_0 = n-k+1.
\end{equation} 
Thus, the intra-cluster repairable code defined in Section \ref{Section:IRC_versus_LRC} is a $(n,l_0,m_0,\mathcal{M},\alpha)$-LRC.

%


Now we show that (\ref{Eqn:locality_ineq}) holds. Note that from Fig. \ref{Fig:tradeoff_min}, we obtain
$\alpha \geq \alpha_{msr}^{(0)}$
for $\epsilon=0$, where
\begin{equation}\label{Eqn:alpha_val}
\alpha_{msr}^{(0)} = \frac{\mathcal{M}}{k-\floor{\frac{k}{n_I}}}
\end{equation}
holds according to (\ref{Eqn:msr_epsilon_0_alpha}). Thus, (\ref{Eqn:locality_ineq}) is proven by showing
\begin{equation}\label{Eqn:local_repair_code_ineq}
m_0 \leq n - \ceil[\bigg]{\frac{\mathcal{M}}{\alpha_{msr}^{(0)}}} - \ceil[\bigg]{\frac{\mathcal{M}}{l_0\alpha_{msr}^{(0)}}} + 2.
\end{equation}
By plugging (\ref{Eqn:locality_val}), (\ref{Eqn:min_dist_val}) and (\ref{Eqn:alpha_val}) into (\ref{Eqn:local_repair_code_ineq}), 
we have
\begin{align*}\label{Eqn:locality_confirm}
n-k+1 & \leq n-\ceil[\bigg]{k-\floor[\bigg]{\frac{k}{n_I}}} - \ceil[\bigg]{\frac{k-\floor{\frac{k}{n_I}}}{n_I-1}} + 2 \nonumber\\
& = n - k + \floor[\bigg]{\frac{k}{n_I}} - \ceil[\bigg]{\frac{k-\floor{\frac{k}{n_I}}}{n_I-1}} + 2.
\end{align*}

Therefore, all we need to prove is 
\begin{equation}\label{Eqn:LRC_WTP}
0 \leq \floor[\bigg]{\frac{k}{n_I}} - \ceil[\bigg]{\frac{k-\floor{\frac{k}{n_I}}}{n_I-1}} + 1,
\end{equation}
which is proved as follows. 
Using $q$ and $r$ defined in \eqref{Eqn:quotient} and \eqref{Eqn:remainder}, the right-hand-side (RHS) of (\ref{Eqn:LRC_WTP}) is
\begin{align*}
RHS &= q - \ceil[\bigg]{\frac{qn_I+r-q}{n_I-1}} + 1 \nonumber\\
&=
\begin{cases}
q - (q+1) + 1 = 0, & \quad r \neq 0\\
q - q + 1 = 1, & \quad r = 0\\
\end{cases}
\end{align*}
Thus, (\ref{Eqn:LRC_WTP}) holds, where the equality condition is $r \neq 0$, or equivalently $(k \Mod{n_I}) \neq 0$. Therefore, (\ref{Eqn:locality_ineq}) holds, where the equality condition is $\alpha = \alpha_{msr}^{(0)}$ and $(k \Mod{n_I}) \neq 0$. This completes the proof of Theorem \ref{Thm:IRC_versus_LRC}.


\section{Proof of Theorem \ref{Thm:beta_c_alpha_trade}}\label{Section:Proof_of_Thm_beta_c_alpha_trade}

According to (\ref{Eqn:lower_bound}) and (\ref{Eqn:capacity_final})  in Appendix B, the capacity can be expressed as
\begin{equation}\label{Eqn:cap_final}
\mathcal{C} = \sum_{i=1}^k \min \{\alpha, \omega_i(\bm{\pi}_v (\bm{s}_h)) \}.
\end{equation}
Using (\ref{Eqn:Lemma_2_proof_4}) and (\ref
{Eqn:t_i^*_cases}), $\omega_i(\bm{\pi}_v (\bm{s}_h))$ in (\ref{Eqn:cap_final}), or simply $\omega_i$ has the following property:
\begin{equation}\label{Eqn:omega_i_relation}
\omega_{i+1} = 
\begin{cases}
\omega_i - \beta_I, & i \in I_{G} \\
\omega_i - \beta_c, & i \in [k] \setminus I_{G}
\end{cases}
\end{equation}
where
$I_{G} = \{g_m\}_{m=1}^{n_I-1}.$
Note that $g_{n_I} = \floor{\frac{k}{n_I}}$ from (\ref{Eqn:g_m}). Therefore, $k_0$ in (\ref{Eqn:k_0}) can be expressed as
\begin{equation}\label{Eqn:k_0_revisited}
k_0 = k - \floor{\frac{k}{n_I}} = k - g_{n_I} \leq k - 1
\end{equation}
where the last inequality is from (\ref{Eqn:k_constraint}). Combining (\ref{Eqn:h_t}) and (\ref{Eqn:k_0_revisited}) result in 
\begin{equation}\label{Eqn:h_k_0}
h_{k_0} = n_I - 1.
\end{equation}
From (\ref{Eqn:Lemma_2_proof_4}), (\ref{Eqn:k_0_revisited}) and (\ref{Eqn:h_k_0}), we have
\begin{align}\label{Eqn:omega_k_0}
\omega_{k_0} &= (n_I - h_{k_0}) \beta_I + (n-k_0- n_I + h_{k_0}) \beta_c \nonumber\\
&= \beta_I + (n-k_0 -1) \beta_c \geq \beta_I + (n-k)\beta_c \geq \beta_I =\alpha.
\end{align}
where the last equality holds due to the assumption of $\beta_I = \alpha$ in the setting of Theorem \ref{Thm:beta_c_alpha_trade}.
Since $(\omega_i)_{i=1}^k$ is a decreasing sequence from (\ref{Eqn:omega_i_relation}), the result of (\ref{Eqn:omega_k_0}) implies that
\begin{equation}
\omega_i \geq \alpha,  \quad \quad 1 \leq i \leq k_0.
\end{equation}
Thus, the capacity expression in (\ref{Eqn:cap_final}) can be expressed as

\begin{align}\label{Eqn:cap_final_cases}
\mathcal{C} &= 
\begin{cases}
k_0 \alpha + \sum_{i=k_0+1}^{k} \omega_i, & \omega_{k_0+1} \leq \alpha \\
m \alpha + \sum_{i=m+1}^{k} \omega_i, & \omega_{m+1} \leq \alpha < \omega_{m} \\
&\ \ \ (k_0 + 1 \leq m \leq k-1)\\
k\alpha, &  \alpha < \omega_k
\end{cases}
\end{align}

Note that from (\ref{Eqn:Lemma_2_proof_4}) and (\ref{Eqn:t_i^*_cases}), we have
$\omega_i = (n-i)\beta_c$
for $i = k_0 + 1, k_0 + 2, \cdots, k$. Therefore, $\mathcal{C}$ in (\ref{Eqn:cap_final_cases}) is
\begin{align}\label{Eqn:cap_final_cases_revisited}
\mathcal{C} &= 
\begin{cases}
k_0 \alpha + \sum_{i=k_0+1}^{k} (n-i)\beta_c, & 0 \leq \beta_c \leq \frac{\alpha}{n-k_0-1}\\
m \alpha + \sum_{i=m+1}^{k} (n-i)\beta_c, & \frac{\alpha}{n-m} < \beta_c \leq \frac{\alpha}{n-m-1} \\
&\ \ \ (k_0 + 1 \leq m \leq k-1)\\
k\alpha, &  \frac{\alpha}{n-k} < \beta_c,
\end{cases}
\end{align}
which is illustrated as a piecewise linear function of $\beta_c$ in Fig. \ref{Fig:beta_c_cap}. Based on (\ref{Eqn:cap_final_cases_revisited}), the sequence  $(T_m)_{m=k_0}^{k}$ in this figure has the following expression:
\begin{equation}\label{Eqn:T_m}
T_m = 
\begin{cases}
k_0 \alpha , & m = k_0 \\
(m + \frac{\sum_{i=m+1}^{k} (n-i)}{n-m}) \alpha, & k_0+1 \leq m \leq k-1 \\
k \alpha, & m = k 
\end{cases}
\end{equation}

\begin{figure}[!t]
	\centering
	\includegraphics[width=85mm]{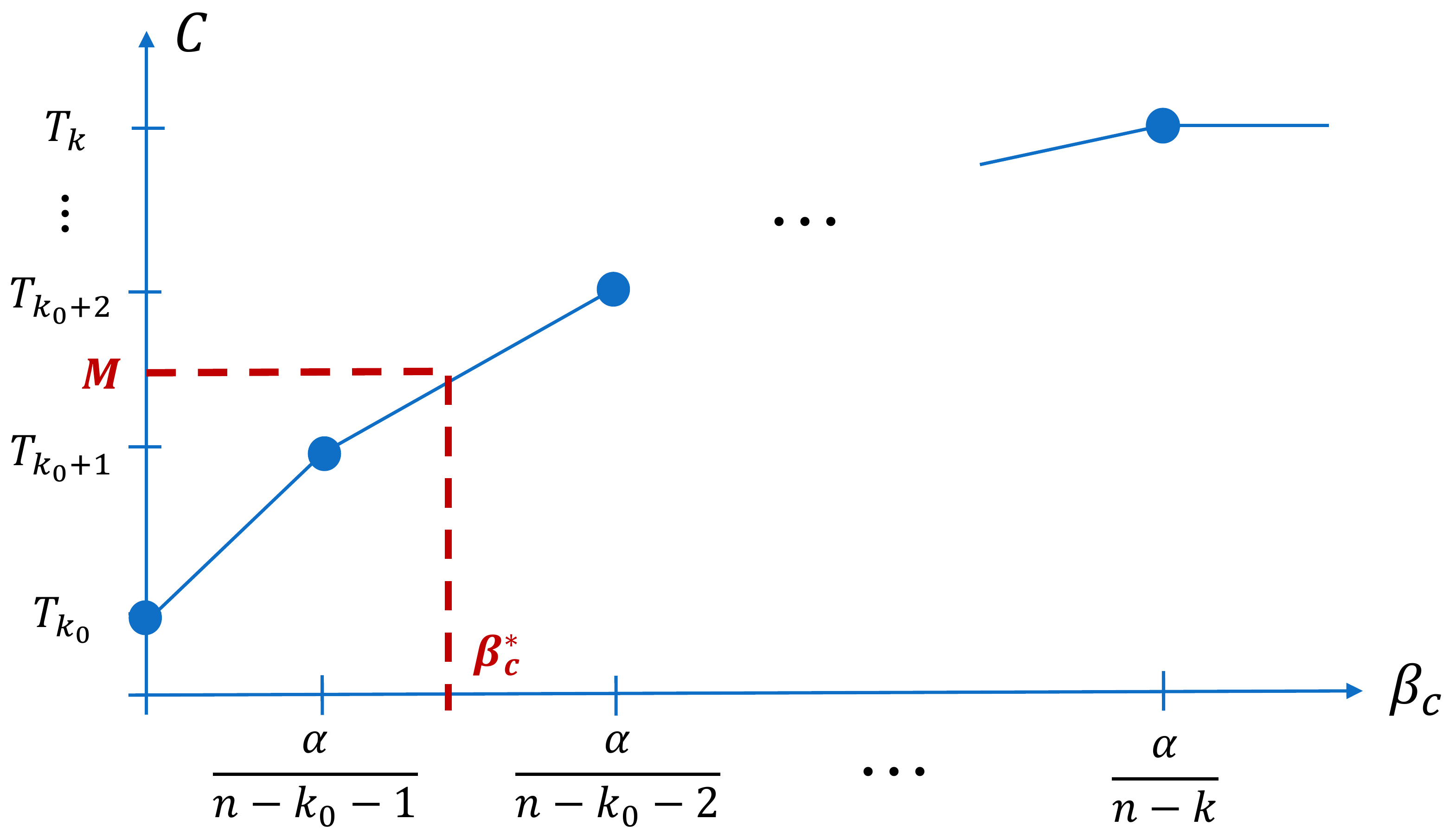}
	\caption{Capacity as a function of $\beta_c$}
	\label{Fig:beta_c_cap}
\end{figure}

From Fig. \ref{Fig:beta_c_cap}, we can conclude that $\mathcal{C} \geq \mathcal{M}$ holds if and only if $\beta_c \geq \beta_c^*$ where
\begin{equation}\label{Eqn:beta_c_star_prev}
\beta_c^* = 
\begin{cases}
0, & \mathcal{M} \in [0, T_{k_0}] \\
\frac{\mathcal{M} - m \alpha}{\sum_{i=m+1}^{k} (n-i)}, & \mathcal{M} \in (T_{m}, T_{m+1}] \\
& \quad (m = k_0, k_0 + 1, \cdots, k-1) \\
\infty, & \mathcal{M} \in (T_k, \infty).
\end{cases}
\end{equation}
Using $T_m$ in (\ref{Eqn:T_m}) and $f_m$ in (\ref{Eqn:f_m}), (\ref{Eqn:beta_c_star_prev}) reduces to (\ref{Eqn:beta_c_star}), which completes the proof.

\section{Proofs of Corollaries}

\subsection{Proof of Corollary \ref{Corollary:Feasible Points_large_epsilon}}
\label{Section: Proof of Corollary:Feasible Points}

From the proof of Theorem \ref{Thm:Capacity of clustered DSS}, the capacity expression is equal to (\ref{Eqn:lower_bound_intermediate}), which is
\begin{equation*}
\mathcal{C} = \sum_{i=1}^k \min \{\alpha, (n_I-h_i)\beta_I + (n-n_I-i+h_i)\beta_c\}.
\end{equation*}
where $h_i$ is defined in (\ref{Eqn:t_i^*}).
Using $\epsilon = \beta_c/\beta_I$ and (\ref{Eqn:gamma}), this can be rewritten as
\begin{equation}\label{Eqn:capacity_kappa_general}
\mathcal{C} = \sum_{i=1}^k \min \{\alpha, \frac{(n-n_I-i+h_i) \epsilon + (n_I - h_i)}{(n-n_I)\epsilon + (n_I-1)} \gamma \}.
\end{equation}

Using $\{z_t\}$ and $\{y_t\}$ defined in (\ref{Eqn:z_t}) and (\ref{Eqn:y_t}), the capacity expression reduces to   
\begin{equation}\label{Eqn:capacity_simple}
\mathcal{C} (\alpha, \gamma)= \sum_{t=1}^{k} \min \{\alpha, \frac{\gamma}{y_t} \},
\end{equation}
which is a continuous function of $\gamma$.

\begin{remark}\label{Rmk:z_dec_y_inc}
$\{z_t\}$ in (\ref{Eqn:z_t}) is a decreasing sequence of $t$. Moreover, $\{y_t\}$ in (\ref{Eqn:y_t}) is an increasing sequence.
\end{remark}
\begin{proof}
Note that from (\ref{Eqn:t_i^*_reform}),
\begin{equation*}
h_{t+1} = 
\begin{cases}
h_t + 1, & t \in T \\
h_t, & t \in [k-1] \setminus T
\end{cases}
\end{equation*}
where
$T = \{g_1, g_1+g_2, \cdots, \sum_{m=1}^{n_I-1} g_m\}. $
Therefore, $\{z_t\}$ in (\ref{Eqn:z_t}) is a decreasing function of $t$, which implies that $\{y_t\}$ is an increasing sequence.
\end{proof}

Moreover, note that 
$\beta_I \leq \alpha$
holds from the definition of $\beta_I$ and $\alpha$ in Table \ref{Table:Params}. 
Thus, combined with $\epsilon = \beta_I/\beta_c$, it is shown that $\gamma$ in (\ref{Eqn:gamma}) is lower-bounded as
\begin{align*}
\gamma &= (n_I-1)\beta_I + (n-n_I)\beta_c  \nonumber\\
&= \{n_I - 1 + (n-n_I)\epsilon\} \beta_I \leq \{n_I - 1 + (n-n_I)\epsilon\} \alpha.
\end{align*}
Here, we define
\begin{equation*}
\overbar{\gamma} = \{n_I - 1 + (n-n_I)\epsilon\} \alpha.
\end{equation*} 
Then, the valid region of $\gamma$ is expressed as 
$\gamma \leq \overbar{\gamma}, $
as illustrated in Figs. \ref{Fig:corollary_proof_1} and \ref{Fig:corollary_proof_2}.
The rest of the proof depends on the range of $\epsilon$ values; we first consider the 
$\frac{1}{n-k} \leq \epsilon \leq 1$ case, and then consider the $0 \leq \epsilon < \frac{1}{n-k}$ case.

\subsubsection{If $ \frac{1}{n-k} \leq \epsilon \leq 1$}

Using (\ref{Eqn:z_k}),
$z_k = (n-k)\epsilon \geq 1$
holds.
Combining with (\ref{Eqn:y_t}), we have
$y_k \leq n_I-1 + \epsilon (n-n_I),$
or equivalently,
$y_k \alpha \leq \overbar{\gamma}.$
If $y_t \alpha  < \gamma \leq y_{t+1}\alpha$ for some $t \in [k-1]$, 
then (\ref{Eqn:capacity_simple}) can be expressed as 
\begin{align*}
\mathcal{C} (\alpha, \gamma)&= t\alpha + \sum_{m=t+1}^k \frac{\gamma}{y_m}  \\
&= t\alpha +  \frac{\gamma \ (\sum_{m=t+1}^k z_m)}{(n_I-1) + \epsilon (n-n_I)} = t\alpha + s_t \gamma
\end{align*}
where $\{s_t\}$ is defined in (\ref{Eqn:s_t}). 
If $0 \leq \gamma \leq y_1\alpha$, then 
\begin{align*}
\mathcal{C}(\alpha,\gamma) &= \sum_{m=1}^k \frac{\gamma}{y_m}  =  \frac{\gamma \ (\sum_{m=1}^k z_m)}{(n_I-1) + \epsilon (n-n_I)} = s_0 \gamma.
\end{align*}
If $y_k \alpha < \gamma \leq \overbar{\gamma}$, then 
$\mathcal{C}(\alpha, \gamma) = \sum_{m=1}^k \alpha = k\alpha. $
In summary, capacity is
\begin{align}\label{Eqn:capacity_cases}
\mathcal{C}(\alpha, \gamma) &= 
\begin{cases}
s_0 \gamma, & 0 \leq \gamma \leq y_1\alpha \\
t\alpha + s_t \gamma, & y_t \alpha < \gamma \leq y_{t+1}\alpha \\
& \ \ \ \ (t = 1, 2, \cdots, k-1)\\
k\alpha, & y_k \alpha < \gamma \leq \overbar{\gamma}
\end{cases}
\end{align}
which is illustrated in Fig. \ref{Fig:corollary_proof_1}.
Since $\{z_t\}$ is a decreasing sequence from Remark \ref{Rmk:z_dec_y_inc}, we have
$	z_t \geq z_k = (n-k)\epsilon > 0$
for $t \in [k]$. Thus, $\{s_t\}_{t=1}^k$ defined in (\ref{Eqn:s_t}) is a monotonically decreasing, non-negative sequence. 
This implies that the curve in Fig. \ref{Fig:corollary_proof_1} is a monotonic increasing function of $\gamma$.

\begin{figure}[!t]
	\centering
	\includegraphics[width=70mm]{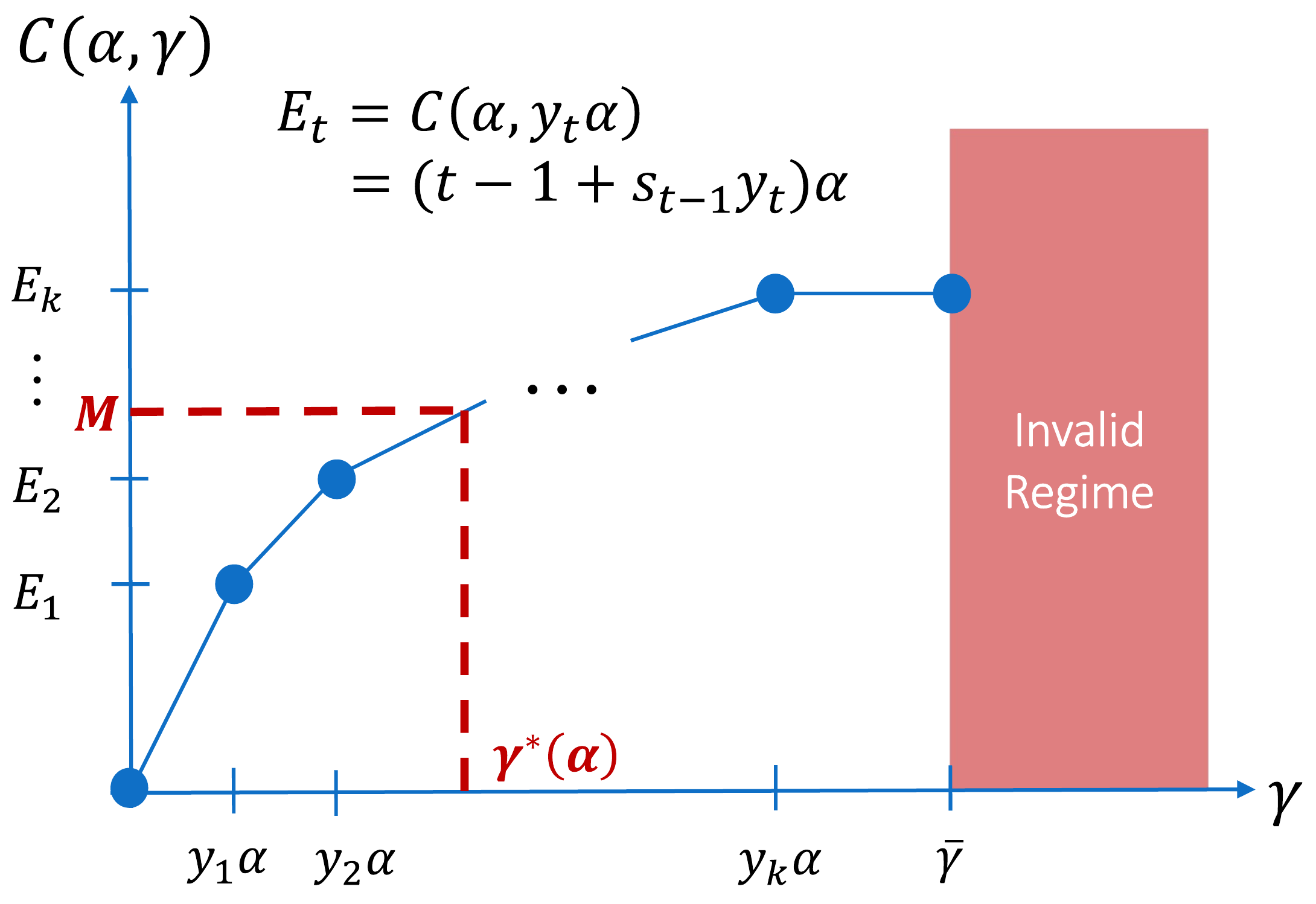}
	\caption{Capacity as a function of $\gamma$, for $\frac{1}{n-k} \leq \epsilon \leq 1$}
	\label{Fig:corollary_proof_1}
\end{figure}

From Fig. \ref{Fig:corollary_proof_1}, it is shown that $\mathcal{C}(\alpha, \gamma) \geq \mathcal{M}$ holds if and only if $\gamma \geq \gamma^*(\alpha)$. From (\ref{Eqn:capacity_cases}), the threshold value $\gamma^*(\alpha)$ can be expressed as 
\begin{equation}\label{Eqn:gamma_th}
\gamma^*(\alpha) =
\begin{cases}
\frac{\mathcal{M}}{s_0}, & \mathcal{M} \in [0, E_1]\\
\frac{\mathcal{M}- t \alpha}{s_t}, & \mathcal{M} \in (E_{t}, E_{t+1}]\\
& \ \ \ \ (t = 1, 2, \cdots, k-1)\\
\infty, & \mathcal{M} \in (E_k, \infty).
\end{cases}
\end{equation} 
where 
\begin{equation}\label{Eqn:E_t}
E_t = \mathcal{C}(\alpha, y_t \alpha) = (t-1+s_{t-1}y_t) \alpha
\end{equation}
for $t \in [k]$.
The threshold value $\gamma^*(\alpha)$ in (\ref{Eqn:gamma_th}) can be expressed as (\ref{Eqn:Feasible Points Result}), which completes the proof.

\vspace{5mm}

\subsubsection{Otherwise (if $0 \leq \epsilon < \frac{1}{n-k}$)}

Using (\ref{Eqn:z_k}),
\begin{equation} \label{Eqn:z_k_small_epsilon}
z_k = (n-k)\epsilon < 1
\end{equation}	 
holds. Since $\{z_t\}$ is a decreasing sequence from Remark \ref{Rmk:z_dec_y_inc}, there exists $\tau \in \{0, 1, \cdots, k-1\}$ such that 
$z_{\tau+1} < 1 \leq z_{\tau}$
holds, or equivalently,
$y_{\tau}\alpha \leq \overbar{\gamma} < y_{\tau+1}\alpha.$

Using the analysis similar to the $\frac{1}{n-k} \leq \epsilon \leq 1$ case, we obtain 
\begin{align}\label{Eqn:capacity_cases_2}
\mathcal{C}(\alpha, \gamma) &= 
\begin{cases}
s_0 \gamma, & 0 \leq \gamma \leq y_1\alpha \\
t\alpha + s_t \gamma, & y_t \alpha < \gamma \leq y_{t+1}\alpha \\
& \ \ \ \ (t = 1, 2, \cdots, \tau-1)\\
\tau\alpha + s_{\tau} \gamma, & y_{\tau} \alpha < \gamma \leq \overbar{\gamma}
\end{cases}
\end{align}
which is illustrated in Fig. \ref{Fig:corollary_proof_2}.

\begin{figure}[!t]
	\centering
	\includegraphics[width=70mm]{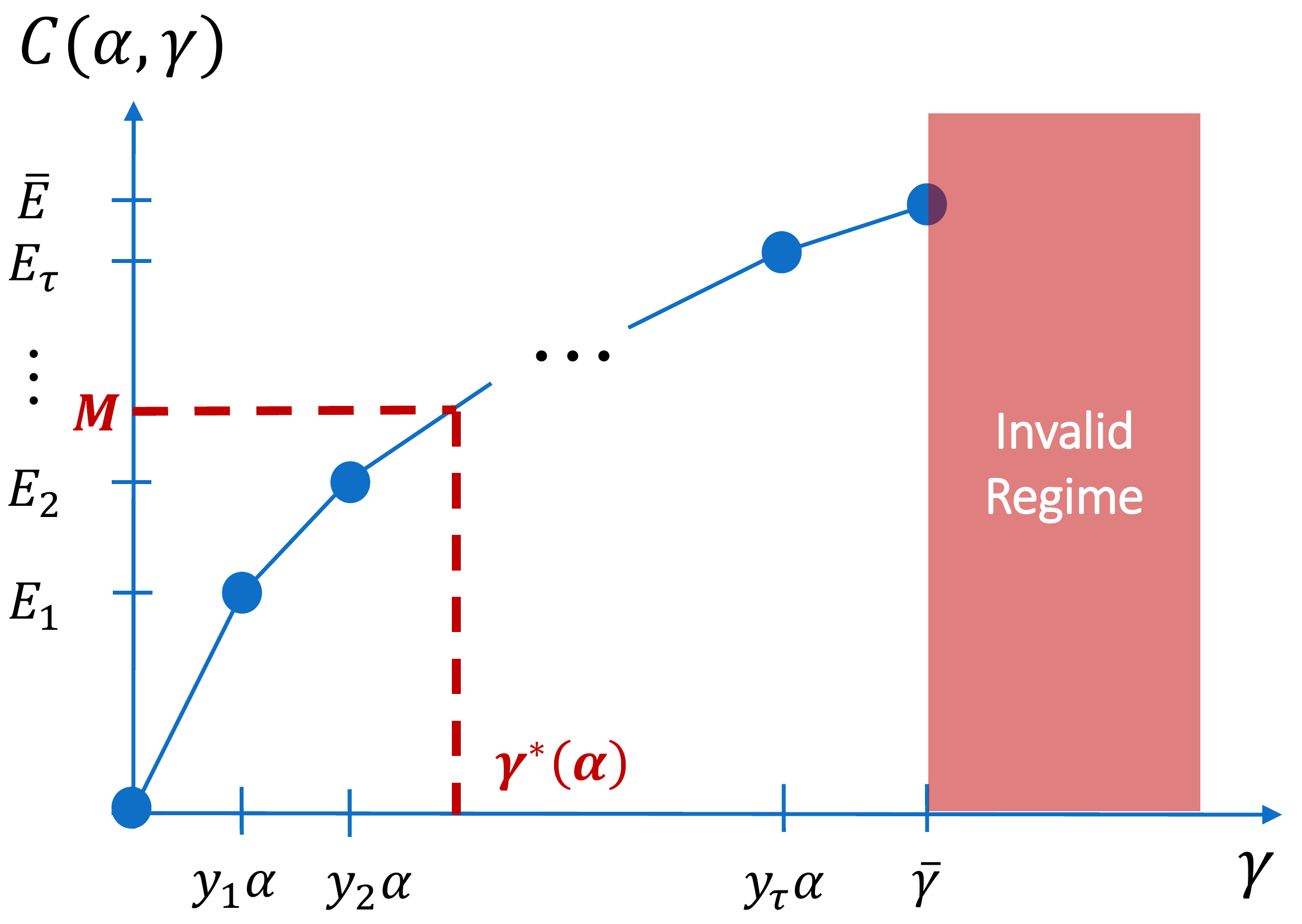}
	\caption{Capacity as a function of $\gamma$, for $0 \leq \epsilon < \frac{1}{n-k} $ case}
	\label{Fig:corollary_proof_2}
\end{figure}

From Fig. \ref{Fig:corollary_proof_2}, it is shown that $\mathcal{C}(\alpha, \gamma) \geq \mathcal{M}$ holds if and only if $\gamma \geq \gamma^*(\alpha)$. From (\ref{Eqn:capacity_cases_2}), the threshold value $\gamma^*(\alpha)$ can be expressed as 
\begin{equation}\label{Eqn:gamma_th_2}
\gamma^*(\alpha) =
\begin{cases}
\frac{\mathcal{M}}{s_0}, & \mathcal{M} \in [0, E_1]\\
\frac{\mathcal{M}- t \alpha}{s_t}, & \mathcal{M} \in (E_{t}, E_{t+1}]\\
& \ \ \ \ (t = 1, 2, \cdots, \tau-1)\\
\frac{\mathcal{M}- \tau \alpha}{s_{\tau}}, & \mathcal{M} \in (E_{\tau}, \overbar{E}]\\
\infty, & \mathcal{M} \in (\overbar{E}, \infty).
\end{cases}
\end{equation} 
where 
$\{E_t\}$ is defined in (\ref{Eqn:E_t}), and
\begin{equation*}
\overbar{E} = \mathcal{C}(\alpha, \overbar{\gamma}) = \tau \alpha + s_{\tau} \alpha \{n_I-1 + (n-n_I)\epsilon\}.
\end{equation*}
The threshold value $\gamma^*(\alpha)$ in (\ref{Eqn:gamma_th_2}) can be expressed as (\ref{Eqn:Feasible Points Result_intermediate_epsilon}), which completes the proof.

\subsection{Proof of Corollary \ref{Coro:msr_mbr_points}}\label{Section:Proof of Corollary_msr_mbr_points}

	First, we focus on the MSR point illustrated in Fig. \ref{Fig:MBR_MSR_points}. From (\ref{Eqn:Feasible Points Result}), the MSR point for $\frac{1}{n-k} \leq \epsilon \leq 1$ is 
	\begin{align}
	(\alpha_{msr}^{(\epsilon)}, \gamma_{msr}^{(\epsilon)}) &= (\frac{\mathcal{M}}{k+s_ky_k}, \frac{\mathcal{M} -  (k-1)\alpha_{msr}^{(\epsilon)}}{s_{k-1}}) \nonumber\\
	&= (\frac{\mathcal{M}}{k}, \frac{\mathcal{M}}{k} \frac{1}{s_{k-1}}) \label{Eqn:proof_of_corollary_MSR}
	\end{align}
	where the last equality is from $s_k = 0$ in (\ref{Eqn:s_t}). 
	Moreover, from (\ref{Eqn:Feasible Points Result_intermediate_epsilon}), the MSR point for $0 \leq \epsilon < \frac{1}{n-k}$ is 
	\begin{align}
	(\alpha_{msr}^{(\epsilon)}, \gamma_{msr}^{(\epsilon)}) &= (\frac{M}{\tau + \sum_{i=\tau+1}^{k}z_i}, \frac{\mathcal{M} -   \tau\alpha_{msr}^{(\epsilon)}}{s_{\tau}}) \nonumber\\
	&= (\frac{M}{\tau + \sum_{i=\tau+1}^{k}z_i}, \frac{M}{\tau + \sum_{i=\tau+1}^{k}z_i} \frac{\sum_{i=\tau+1}^{k}z_i}{s_{\tau}}). \label{Eqn:proof_of_corollary_MSR_2}
	\end{align}
	Equations (\ref{Eqn:proof_of_corollary_MSR}) and (\ref{Eqn:proof_of_corollary_MSR_2}) proves (\ref{Eqn:MSR_point}). 
	The expression (\ref{Eqn:MBR_point}) for the MBR point is directly obtained from Corollary \ref{Corollary:Feasible Points_large_epsilon} and Fig. \ref{Fig:MBR_MSR_points}.

\section{Proofs of Propositions}
\label{Section:proofs_of_remarks}

\subsection{Proof of Proposition \ref{Prop:max_helper_nodes}}\label{Section:max_helper_node_assumption}

As in (\ref{Eqn:Capacity_expression_minmin}), capacity $\mathcal{C}$ for maximum $d_I, d_c$ setting ($d_I = n_I-1, d_c = n-n_I$) is expressed as  
\begin{equation}
\mathcal{C} = \displaystyle\min_{\bm{s} \in S, \bm{\pi} \in \Pi(\bm{s})} L  (\bm{s},\bm{\pi})
\end{equation}
where 
\begin{align}
L  (\bm{s},\bm{\pi}) &= \sum_{i=1}^k \min\{\alpha, \omega_i(\bm{\pi})\}, \nonumber\\
\omega_i(\bm{\pi}) &= \gamma -  e_i(\bm{\pi}) \beta_I - (i - 1 - e_i(\bm{\pi})) \beta_c, \label{Eqn:omega_i_max}\\
e_i(\bm{\pi}) &= \sum_{j=1}^{i-1} \mathds{1}_{\pi_j = \pi_i},\nonumber
\end{align}
as in (\ref{Eqn:lower_bound}), (\ref{Eqn:weight vector}) and (\ref{Eqn:sum_beta_ji}).

Consider a general $d_I, d_c$ setting,
where each newcomer node is helped by $d_I$ nodes in the same cluster, receiving $\beta_I$ information from each node, and $d_c$ nodes in other clusters, receiving $\beta_c$ information from each node.
Under this setting, the coefficient of $\beta_I$ in (\ref{Eqn:omega_i_max}) cannot exceed $d_I$. Similarly, the coefficient of $\beta_c$ in (\ref{Eqn:omega_i_max}) cannot exceed $d_c$.  Thus, the capacity for general $d_I, d_c$ is expressed as
\begin{equation}\label{Eqn:capacity_general_helper_nodes}
\mathcal{C}(d_I, d_c) = \displaystyle\min_{\bm{s} \in S, \bm{\pi} \in \Pi(\bm{s})} \  L  (d_I, d_c, \bm{s},\bm{\pi})
\end{equation}
where 
\begin{align}
L  (d_I, d_c,\bm{s},\bm{\pi}) &= \sum_{i=1}^{k} \min \{\alpha, \omega_i(d_I, d_c, \bm{s}, \bm{\pi})\}, \nonumber\\
\omega_i(d_I, d_c, \bm{s}, \bm{\pi}) &= \gamma - \min\{d_I, e_i(\bm{\pi})\}\beta_I \nonumber\\
& \quad \quad - \min\{d_c, i-1-e_i(\bm{\pi})\} \beta_c,\label{Eqn:omega_i_general}\\
e_i(\bm{\pi}) &= \sum_{j=1}^{i-1} \mathds{1}_{\pi_j = \pi_i}.\nonumber
\end{align}

Consider arbitrary fixed $\bm{s}, \bm{\pi}$ and $d_c$. 
Since $\gamma$ and $\gamma_c=d_c\beta_c$ are fixed in the basic setting of Proposition \ref{Prop:max_helper_nodes}, only $d_I$ and $\beta_I$ are variables in (\ref{Eqn:omega_i_general}), while other parameters are constants.
Then, (\ref{Eqn:omega_i_general}) can be expressed as
\begin{align}
\omega_i(d_I, d_c, \bm{s}, \bm{\pi}) 
&= C_1 - \min\{d_I, e_i(\bm{\pi})\}\beta_I \nonumber\\
& = C_1 - \frac{\min\{d_I, e_i(\bm{\pi})\}}{d_I}C_2 \label{Eqn:omega_general_constants}
\end{align}
where
$C_1 = \gamma - \min\{d_c, i-1-e_i(\bm{\pi})\} \beta_c $ and
$C_2 = \gamma_I = \gamma - \gamma_c$
are constants. 
Note that
\begin{equation}
\frac{\min\{d_I, e_i(\bm{\pi})\}}{d_I} = 
\begin{cases}
1, & \text{ if } d_I \leq e_i(\bm{\pi}), \\
\frac{e_i(\bm{\pi})}{d_I}, & \text{ otherwise }
\end{cases}
\end{equation}
is a non-increasing function of $d_I$. Thus, $\omega_i(d_I, d_c, \bm{s}, \bm{\pi})$ in (\ref{Eqn:omega_general_constants}) is a non-decreasing function of $d_I$ for arbitrary fixed $\bm{s}, \bm{\pi}, d_c$ and $i \in [k]$.
Since the maximum $d_I$ value is $n_I-1$, we have
\begin{align}
L(d_I,d_c,\bm{s},\bm{\pi}) &= \sum_{i=1}^{k} \min \{\alpha, \omega_i(d_I, d_c, \bm{s}, \bm{\pi})\} \nonumber\\ 
&\leq  \sum_{i=1}^{k} \min \{\alpha, \omega_i(n_I-1, d_c, \bm{s}, \bm{\pi})\} \nonumber\\ 
&= L(n_I-1,d_c,\bm{s},\bm{\pi})
\end{align}
for $d_I \in [n_I-1] $.
In other words, for all $\bm{s}, \bm{\pi}, d_c$, we have
\begin{equation*}
\underset{d_I \in [n_I-1]}{\arg \max}\ L(d_I,d_c,\bm{s},\bm{\pi}) = n_I-1.
\end{equation*}
Similarly, for all $\bm{s}, \bm{\pi}, d_I$, 
\begin{equation*}
\underset{d_c \in [n-n_I]}{\arg \max}\  L(d_I,d_c,\bm{s},\bm{\pi}) = n-n_I
\end{equation*}
holds. Therefore, for all $\bm{s},\bm{\pi}$,
\begin{equation} \label{Eqn:lower_bound_argmax}
\underset{[d_I, d_c]}{\arg \max}\  L(d_I,d_c,\bm{s},\bm{\pi}) = [n_I-1, n-n_I].
\end{equation}
Let
\begin{equation}\label{Eqn:optimal_s_pi}
[\bm{s}^*, \bm{\pi}^*] = 
\underset{\bm{s} \in S, \bm{\pi} \in \Pi(\bm{s})}{\arg \min}\ \  L  (n_I-1, n-n_I, \bm{s},\bm{\pi}).
\end{equation}
Then, from (\ref{Eqn:capacity_general_helper_nodes}), (\ref{Eqn:lower_bound_argmax}) and (\ref{Eqn:optimal_s_pi}),
\begin{align}
\mathcal{C}(d_I, d_c) &= \displaystyle\min_{\bm{s} \in S, \bm{\pi} \in \Pi(\bm{s})} \  L  (d_I, d_c, \bm{s},\bm{\pi}) \nonumber\\
& \leq L  (d_I, d_c, \bm{s}^*,\bm{\pi}^*) \leq L  (n_I-1, n-n_I, \bm{s}^*,\bm{\pi}^*) \nonumber\\
&= \displaystyle\min_{\bm{s} \in S, \bm{\pi} \in \Pi(\bm{s})} \  L  (n_I-1, n-n_I, \bm{s},\bm{\pi}) \nonumber\\
& = \mathcal{C}(n_I-1, n-n_I)
\end{align}
for all $d_I \in [n_I-1] $ and $d_c \in [n-n_I] $.
Therefore, choosing $d_I = n_I -1$ and $d_c = n-n_I$ maximizes storage capacity when the available resources, $\gamma$ and $\gamma_c$, are given.

\subsection{Proof of Proposition \ref{Prop:omega_i_bounded_by_gamma}}\label{Section:proof_of_omega_bound_gamma}

First, we prove (\ref{Eqn:omega_is_bounded_by_gamma}).
Recall $\rho_i$ and $g_m$ defined in (\ref{Eqn:rho_i}) and (\ref{Eqn:g_m}). 
Consider the \textit{support set} $S$, which is defined as
\begin{equation}
S = \{i \in [n_I] : g_i \geq 1  \}.\label{Eqn:Support_Set}
\end{equation}
Then, we have 
\begin{align}
\rho_i = n_I - i &\leq n_I - 1, \label{rho_i_inequality} \\
\sum_{m=1}^{i-1}g_m  &\geq i-1 \label{Eqn:t_ij_inequality}
\end{align}
for  every $i \in S$. 
Therefore, by combining (\ref{Eqn:t_ij_inequality}) and (\ref{Eqn:rho_i}), 
\begin{align}\label{phi_ij_inequality}
n-\rho_i-j-\sum_{m=1}^{i-1}g_m & \leq n- (i-1) - j - \rho_i \nonumber\\
&= n-n_I - (j-1) \leq n-n_I
\end{align}
holds for every $i \in S, j \in [g_i] $.
Combining (\ref{Eqn:gamma}), (\ref{rho_i_inequality}) and (\ref{phi_ij_inequality}) results in 
\begin{equation}\label{Eqn:omega_i_bounded_by_gamma}
\rho_i \beta_I + (n-\rho_i-j-\sum_{m=1}^{i-1}g_m)\beta_c \leq (n_I-1)\beta_I + (n-n_I)\beta_c = \gamma
\end{equation}
for arbitrary $i \in S, j \in [g_i] $.
Since $ [g_i]  = \emptyset$ holds for $i \in [n_I] \setminus S$, we conclude that (\ref{Eqn:omega_i_bounded_by_gamma}) holds for $(i,j)$ with $i \in [n_I]$, $j \in [g_i]$. 

Second, we prove (\ref{Eqn:sum of g is k}).	
Using $q$ and $r$ in (\ref{Eqn:quotient}) and (\ref{Eqn:remainder}),
\begin{equation}\label{Eqn:g_m_simple}
g_i = 
\begin{cases}
q+1, & i \leq r \\
q, & \text{otherwise}
\end{cases}
\end{equation}
Therefore, 
$\sum_{i=1}^{n_I}g_i = (q+1)r + q(n_I-r) = r + qn_I = k,$
where the last equality is from (\ref{Eqn:quotient_remainder_relation}).

\section{Proofs of Lemmas}\label{Section:Proofs_of_Lemmas}

\subsection{Proof of Lemma \ref{Lemma:cap_upper_lower}}\label{Section:proof_of_prop_bound}

Using
(\ref{Eqn:gamma}), $\underline{C}$ in (\ref{Eqn:cap_lower}) can be expressed as 
\begin{align}
\underline{C} &= \frac{k}{2} \left( (n_I-1)\beta_I + (n-n_I)(1+\frac{n-k}{n(1-1/L)}) \beta_c \right) \nonumber\\
&= \frac{k}{2} \left( (n_I-1)\beta_I + (n-n_I)(1+\frac{n-k}{n-n_I}) \beta_c \right) \nonumber\\
&= \frac{k}{2} \left\{ (n_I-1)\beta_I + (2n-n_I-k) \beta_c \right\}. \label{Eqn:cap_lower_revisited}
\end{align} 
According to (\ref{Eqn:lower_bound}) and (\ref{Eqn:capacity_final})  in Appendix B, the capacity can be expressed as
\begin{equation}\label{Eqn:cap}
\mathcal{C} = L(\bm{s}_h, \bm{\pi}_v) = \sum_{i=1}^{k} \min \{\alpha, \omega_i (\bm{\pi}_{v} (\bm{s}_h))\}.
\end{equation}
From  (\ref{Eqn:weight vector}), we have 
\begin{equation}
\omega_i (\bm{\pi}_{v} (\bm{s}_h))  \leq \gamma
\end{equation}
for $i \in [k]$. 
Therefore, when $\alpha = \gamma$, the capacity expression in (\ref{Eqn:cap}) reduces to
\begin{equation}\label{Eqn:cap_bw_limited}
\mathcal{C} = \sum_{i=1}^k \omega_i (\bm{\pi}_{v} (\bm{s}_h)).
\end{equation}
Recall (\ref{Eqn:Lemma_2_proof_4}):
\begin{equation}\label{Eqn:opt_omega_i_revisited}
\omega_i(\bm{\pi}_v(\bm{s}_h)) = (n_I - h_i) \beta_I + (n-i-n_I + h_i) \beta_c.
\end{equation}
From the expression of $(h_i)_{i=1}^k$ in (\ref{Eqn:t_i^*_cases}), 
we have
\begin{equation}\label{Eqn:A_0_part}
\sum_{i=1}^k (n_I - h_i)  = \sum_{s=1}^{n_I} g_s (n_I-s)
\end{equation}
since $\sum_{m=1}^k g_m = k$ from (\ref{Eqn:sum of g is k}), 

Using (\ref{Eqn:opt_omega_i_revisited}) and (\ref{Eqn:A_0_part}), 
the capacity expression in (\ref{Eqn:cap_bw_limited}) is 
expressed as
\begin{align}\label{Eqn:cap_bw_limited_factored}
\mathcal{C} &= \left( \sum_{i=1}^k (n_I-h_i) \right) \beta_I + \left( \sum_{i=1}^k (n-i - (n_I-h_i)) \right) \beta_c \nonumber\\
&= A_0 \beta_I + B_0 \beta_c
\end{align}
where $A_0 = \sum_{s=1}^{n_I} g_s (n_I-s)$ and
\begin{align}
B_0 &= \sum_{i=1}^k (n-i) - A_0. \label{Eqn:B_0}
\end{align}
Similarly, $\underline{C}$ in (\ref{Eqn:cap_lower_revisited}) can be expressed as 
\begin{equation}\label{Eqn:cap_underbar_factored}
\underline{C} = A_0'\beta_I + B_0'\beta_c
\end{equation}
where 
$A_0' = (n_I-1)\frac{k}{2}$ and
\begin{align}
B_0' &= (2n-n_I-k)\frac{k}{2} = \sum_{i=1}^{k}(n-i)-A_0'. \label{Eqn:B_0_prime}
\end{align}

First, we show that $\mathcal{C} \geq \underline{C}$ holds.
From (\ref{Eqn:cap_bw_limited_factored}), (\ref{Eqn:B_0}),  (\ref{Eqn:cap_underbar_factored}) and (\ref{Eqn:B_0_prime}), 
\begin{align}\label{Eqn:c-c'}
\mathcal{C} &- \underline{C} = (A_0 - A_0')\beta_I + (B_0 - B_0')\beta_c \nonumber\\
&= (A_0 - A_0')\beta_I - (A_0 - A_0')\beta_c 	= (A_0-A_0') (\beta_I - \beta_c).
\end{align}
Since we consider the $\beta_I \geq \beta_c $ case, all we need to prove is 
\begin{equation*}
A_0-A_0' \geq 0.
\end{equation*}
Using $(g_i)_{i=1}^k$ expression in (\ref{Eqn:g_m_simple}), 
$A_0$ 
can be rewritten as 
\begin{align*} 
A_0 & = \sum_{s=1}^{n_I}g_s(n_I-s)= q \sum_{s=1}^{n_I} (n_I-s) +  \sum_{s=1}^{r} (n_I-s) \nonumber\\
&= q \frac{n_I(n_I-1)}{2} +\sum_{s=1}^{r} (n_I-s) \nonumber \\
&= q \frac{n_I(n_I-1)}{2} + \frac{r((n_I-1)+(n_I-r))}{2} \nonumber \\
&= \left(qn_I + r\right) \frac{n_I-1}{2} + \frac{r(n_I-r) }{2} \nonumber\\
& = \frac{k(n_I-1)}{2} + \frac{r(n_I-r) }{2} = A_0' + \frac{r(n_I-r) }{2}.
\end{align*}
Since $0 \leq r < n_I$, we have
\begin{equation} \label{A-A'}
A_0-A_0' = \frac{r(n_I-r)}{2} \geq 0.
\end{equation}
Therefore, $\mathcal{C} \geq \underline{C} $ holds.

Second, we prove  $\mathcal{C} \leq \underline{C} + n_I^2(\beta_I-\beta_c)/8$. Note that $A_0-A_0'$ in (\ref{A-A'}) is maximized when $r = \lfloor n_I/2 \rfloor$ holds.
Thus, 
$A_0-A_0' \leq \frac{\lfloor n_I/2 \rfloor (n_I - \lfloor n_I/2 \rfloor)}{2} \leq n_I^2/8.$
Combining with (\ref{Eqn:c-c'}),
\begin{align*}
\mathcal{C}  - \underline{C} &= (A_0-A_0')(\beta_I-\beta_c) \leq n_I^2(\beta_I-\beta_c)/8.
\end{align*} 

%

\subsection{Proof of Lemma \ref{Lemma:Scale}}\label{Section:proof_of_prop_scale}
	
\cmt{Recall that $\xi= \gamma_c/\gamma, \gamma$ and $R = k/n$ value are all fixed.}
The expression for $\underbar{C}$ in (\ref{Eqn:cap_lower}) can be expressed as 
\begin{align}
\underline{C}  &=  \frac{k}{2} ( \gamma + \frac{n-k}{n(1-\frac{1}{L}) }\gamma_c ) =  \frac{nR}{2} ( \gamma + \frac{1-R}{1-\frac{1}{L}}\gamma_c ) \nonumber\\ 
&=  \gamma \frac{nR}{2} ( 1 + \frac{1-R}{(1-\frac{1}{L}) }\xi) \label{Eqn:C_underbar_monotone} 
\end{align}
Note that (\ref{Eqn:C_underbar_monotone}) is a monotonic decreasing function of $L$. Moreover, we consider the $L \geq 2$ case, as mentioned in (\ref{Eqn:L_constraint}). Thus, $\underline{C}$ is upper/lower bounded by expressions for $L=2$ and $L=\infty$, respectively:
\begin{equation}
\gamma \frac{nR}{2} ( 1 + (1-R)\xi) < \underline{C} \leq \gamma \frac{nR}{2} ( 1 + 2(1-R)\xi).
\end{equation}
Therefore, 
$\underline{C} = \Theta(n)$
holds.
Moreover, the expression for $\delta$ in (\ref{Eqn:cap_delta_val}) is 
\begin{align}\nonumber
\delta & = \frac{n_I^2(\beta_I-\beta_c)}{8} = \frac{n_I^2}{8}(\frac{\gamma_I}{n_I-1} - \frac{\gamma_c}{n-n_I}) \nonumber \\
&= \frac{n_I^2}{8}(\frac{\gamma (1-\xi)}{n_I-1} - \frac{ \gamma\xi}{n-n_I}) \nonumber\\
&= \frac{n_I^2 \gamma}{8}  \frac{(1-\xi)(n-n_I) - \xi(n_I-1)}{(n_I-1) (n-n_I)}. \label{Eqn:delta}
\end{align}
Putting $n = n_IL$ into (\ref{Eqn:delta}), we get
\begin{align}
\delta &= \frac{n_I \gamma}{8} \frac{(1 - \xi)n_I(L-1) - \xi(n_I-1)}{(n_I-1)(L-1)} 
= O(n_I). \nonumber
\end{align}


%
%

\ifCLASSOPTIONcaptionsoff
  \newpage
\fi



%
%
%
\bibliographystyle{IEEEtran}
\bibliography{IEEEabrv,IEEE_T_IT2017}

\begin{thebibliography}{10}
\providecommand{\url}[1]{#1}
\csname url@samestyle\endcsname
\providecommand{\newblock}{\relax}
\providecommand{\bibinfo}[2]{#2}
\providecommand{\BIBentrySTDinterwordspacing}{\spaceskip=0pt\relax}
\providecommand{\BIBentryALTinterwordstretchfactor}{4}
\providecommand{\BIBentryALTinterwordspacing}{\spaceskip=\fontdimen2\font plus
\BIBentryALTinterwordstretchfactor\fontdimen3\font minus
  \fontdimen4\font\relax}
\providecommand{\BIBforeignlanguage}[2]{{%
\expandafter\ifx\csname l@#1\endcsname\relax
\typeout{** WARNING: IEEEtran.bst: No hyphenation pattern has been}%
\typeout{** loaded for the language `#1'. Using the pattern for}%
\typeout{** the default language instead.}%
\else
\language=\csname l@#1\endcsname
\fi
#2}}
\providecommand{\BIBdecl}{\relax}
\BIBdecl

\bibitem{sohn2016capacity}
J.~y.~Sohn, B.~Choi, S.~W. Yoon, and J.~Moon, ``Capacity of clustered
  distributed storage,'' in \emph{2017 IEEE International Conference on
  Communications (ICC)}, May 2017, pp. 1--7.

\bibitem{ghemawat2003google}
S.~Ghemawat, H.~Gobioff, and S.-T. Leung, ``The google file system,'' in
  \emph{ACM SIGOPS operating systems review}, vol.~37, no.~5.\hskip 1em plus
  0.5em minus 0.4em\relax ACM, 2003, pp. 29--43.

\bibitem{bhagwan2004total}
R.~Bhagwan, K.~Tati, Y.~Cheng, S.~Savage, and G.~M. Voelker, ``Total recall:
  System support for automated availability management.'' in \emph{NSDI},
  vol.~4, 2004, pp. 25--25.

\bibitem{dabek2004designing}
F.~Dabek, J.~Li, E.~Sit, J.~Robertson, M.~F. Kaashoek, and R.~Morris,
  ``Designing a dht for low latency and high throughput.'' in \emph{NSDI},
  vol.~4, 2004, pp. 85--98.

\bibitem{rhea2003pond}
S.~C. Rhea, P.~R. Eaton, D.~Geels, H.~Weatherspoon, B.~Y. Zhao, and
  J.~Kubiatowicz, ``Pond: The oceanstore prototype.'' in \emph{FAST}, vol.~3,
  2003, pp. 1--14.

\bibitem{shvachko2010hadoop}
K.~Shvachko, H.~Kuang, S.~Radia, and R.~Chansler, ``The hadoop distributed file
  system,'' in \emph{Mass storage systems and technologies (MSST), 2010 IEEE
  26th symposium on}.\hskip 1em plus 0.5em minus 0.4em\relax IEEE, 2010, pp.
  1--10.

\bibitem{huang2012erasure}
C.~Huang, H.~Simitci, Y.~Xu, A.~Ogus, B.~Calder, P.~Gopalan, J.~Li, and
  S.~Yekhanin, ``Erasure coding in windows azure storage,'' in \emph{Presented
  as part of the 2012 USENIX Annual Technical Conference (USENIX ATC 12)},
  2012, pp. 15--26.

\bibitem{muralidhar2014f4}
S.~Muralidhar, W.~Lloyd, S.~Roy, C.~Hill, E.~Lin, W.~Liu, S.~Pan, S.~Shankar,
  V.~Sivakumar, L.~Tang \emph{et~al.}, ``f4: Facebook's warm blob storage
  system,'' in \emph{11th USENIX Symposium on Operating Systems Design and
  Implementation (OSDI 14)}, 2014, pp. 383--398.

\bibitem{borthakur2010hdfs}
D.~Borthakur, R.~Schmidt, R.~Vadali, S.~Chen, and P.~Kling, ``Hdfs raid,'' in
  \emph{Hadoop User Group Meeting}, 2010.

\bibitem{dimakis2010network}
A.~G. Dimakis, P.~B. Godfrey, Y.~Wu, M.~J. Wainwright, and K.~Ramchandran,
  ``Network coding for distributed storage systems,'' \emph{IEEE Transactions
  on Information Theory}, vol.~56, no.~9, pp. 4539--4551, 2010.

\bibitem{ahlswede2000network}
R.~Ahlswede, N.~Cai, S.-Y. Li, and R.~W. Yeung, ``Network information flow,''
  \emph{IEEE Transactions on information theory}, vol.~46, no.~4, pp.
  1204--1216, 2000.

\bibitem{rashmi2009explicit}
K.~Rashmi, N.~B. Shah, P.~V. Kumar, and K.~Ramchandran, ``Explicit construction
  of optimal exact regenerating codes for distributed storage,'' in
  \emph{Communication, Control, and Computing, 2009. Allerton 2009. 47th Annual
  Allerton Conference on}.\hskip 1em plus 0.5em minus 0.4em\relax IEEE, 2009,
  pp. 1243--1249.

\bibitem{rashmi2011optimal}
K.~V. Rashmi, N.~B. Shah, and P.~V. Kumar, ``Optimal exact-regenerating codes
  for distributed storage at the msr and mbr points via a product-matrix
  construction,'' \emph{IEEE Transactions on Information Theory}, vol.~57,
  no.~8, pp. 5227--5239, 2011.

\bibitem{shah2012interference}
N.~B. Shah, K.~Rashmi, P.~V. Kumar, and K.~Ramchandran, ``Interference
  alignment in regenerating codes for distributed storage: Necessity and code
  constructions,'' \emph{IEEE Transactions on Information Theory}, vol.~58,
  no.~4, pp. 2134--2158, 2012.

\bibitem{ford2010availability}
D.~Ford, F.~Labelle, F.~I. Popovici, M.~Stokely, V.-A. Truong, L.~Barroso,
  C.~Grimes, and S.~Quinlan, ``Availability in globally distributed storage
  systems.'' in \emph{OSDI}, 2010, pp. 61--74.

\bibitem{rashmi2013solution}
K.~Rashmi, N.~B. Shah, D.~Gu, H.~Kuang, D.~Borthakur, and K.~Ramchandran, ``A
  solution to the network challenges of data recovery in erasure-coded
  distributed storage systems: A study on the facebook warehouse cluster.'' in
  \emph{HotStorage}, 2013.

\bibitem{ahmad2014shufflewatcher}
F.~Ahmad, S.~T. Chakradhar, A.~Raghunathan, and T.~Vijaykumar,
  ``Shufflewatcher: Shuffle-aware scheduling in multi-tenant mapreduce
  clusters,'' in \emph{2014 USENIX Annual Technical Conference (USENIX ATC
  14)}, 2014, pp. 1--13.

\bibitem{benson2010network}
T.~Benson, A.~Akella, and D.~A. Maltz, ``Network traffic characteristics of
  data centers in the wild,'' in \emph{Proceedings of the 10th ACM SIGCOMM
  conference on Internet measurement}.\hskip 1em plus 0.5em minus 0.4em\relax
  ACM, 2010, pp. 267--280.

\bibitem{vahdat2010scale}
A.~Vahdat, M.~Al-Fares, N.~Farrington, R.~N. Mysore, G.~Porter, and
  S.~Radhakrishnan, ``Scale-out networking in the data center,'' \emph{Ieee
  Micro}, vol.~30, no.~4, pp. 29--41, 2010.

\bibitem{ernvall2013capacity}
T.~Ernvall, S.~El~Rouayheb, C.~Hollanti, and H.~V. Poor, ``Capacity and
  security of heterogeneous distributed storage systems,'' \emph{IEEE Journal
  on Selected Areas in Communications}, vol.~31, no.~12, pp. 2701--2709, 2013.

\bibitem{yu2015tradeoff}
Q.~Yu, K.~W. Shum, and C.~W. Sung, ``Tradeoff between storage cost and repair
  cost in heterogeneous distributed storage systems,'' \emph{Transactions on
  Emerging Telecommunications Technologies}, vol.~26, no.~10, pp. 1201--1211,
  2015.

\bibitem{akhlaghi2010fundamental}
S.~Akhlaghi, A.~Kiani, and M.~R. Ghanavati, ``A fundamental trade-off between
  the download cost and repair bandwidth in distributed storage systems,'' in
  \emph{2010 IEEE International Symposium on Network Coding (NetCod)}.\hskip
  1em plus 0.5em minus 0.4em\relax IEEE, 2010, pp. 1--6.

\bibitem{shah2010flexible}
N.~B. Shah, K.~Rashmi, and P.~V. Kumar, ``A flexible class of regenerating
  codes for distributed storage,'' in \emph{Information Theory Proceedings
  (ISIT), 2010 IEEE International Symposium on}.\hskip 1em plus 0.5em minus
  0.4em\relax IEEE, 2010, pp. 1943--1947.

\bibitem{gaston2013realistic}
B.~Gast{\'o}n, J.~Pujol, and M.~Villanueva, ``A realistic distributed storage
  system that minimizes data storage and repair bandwidth,'' \emph{arXiv
  preprint arXiv:1301.1549}, 2013.

\bibitem{prakash2016generalization}
N.~Prakash, V.~Abdrashitov, and M.~M{\'e}dard, ``A generalization of
  regenerating codes for clustered storage systems,'' in \emph{Communication,
  Control, and Computing (Allerton), 2016 54th Annual Allerton Conference on},
  2016.

\bibitem{prakash2017storage}
------, ``The storage vs repair-bandwidth trade-off for clustered storage
  systems,'' \emph{arXiv preprint arXiv:1701.04909}, 2017.

\bibitem{choi2017secure}
B.~Choi, J.-y. Sohn, S.~W. Yoon, and J.~Moon, ``Secure clustered distributed
  storage against eavesdroppers,'' \emph{arXiv preprint arXiv:1702.07498},
  2017.

\bibitem{hu2017optimal}
Y.~Hu, X.~Li, M.~Zhang, P.~P. Lee, X.~Zhang, P.~Zhou, and D.~Feng, ``Optimal
  repair layering for erasure-coded data centers: From theory to practice,''
  \emph{ACM Transactions on Storage (TOS)}, vol.~13, no.~4, p.~33, 2017.

\bibitem{calis2016architecture}
G.~Calis and O.~O. Koyluoglu, ``Architecture-aware coding for distributed
  storage: Repairable block failure resilient codes,'' \emph{arXiv preprint
  arXiv:1605.04989}, 2016.

\bibitem{ye2017explicit}
M.~Ye and A.~Barg, ``Explicit constructions of optimal-access mds codes with
  nearly optimal sub-packetization,'' \emph{IEEE Transactions on Information
  Theory}, vol.~63, no.~10, pp. 6307--6317, 2017.

\bibitem{tamo2016optimal}
I.~Tamo, D.~S. Papailiopoulos, and A.~G. Dimakis, ``Optimal locally repairable
  codes and connections to matroid theory,'' \emph{IEEE Transactions on
  Information Theory}, vol.~62, no.~12, pp. 6661--6671, 2016.

\bibitem{papailiopoulos2014locally}
D.~S. Papailiopoulos and A.~G. Dimakis, ``Locally repairable codes,''
  \emph{IEEE Transactions on Information Theory}, vol.~60, no.~10, pp.
  5843--5855, 2014.

\bibitem{dimakis2011survey}
A.~G. Dimakis, K.~Ramchandran, Y.~Wu, and C.~Suh, ``A survey on network codes
  for distributed storage,'' \emph{Proceedings of the IEEE}, vol.~99, no.~3,
  pp. 476--489, 2011.

\bibitem{bang2008digraphs}
J.~Bang-Jensen and G.~Z. Gutin, \emph{Digraphs: theory, algorithms and
  applications}.\hskip 1em plus 0.5em minus 0.4em\relax Springer Science \&
  Business Media, 2008.

\bibitem{erdélyi1956asymptotic}
\BIBentryALTinterwordspacing
A.~Erd{\'e}lyi, \emph{Asymptotic Expansions}, ser. Dover Books on
  Mathematics.\hskip 1em plus 0.5em minus 0.4em\relax Dover Publications, 1956.
  [Online]. Available: \url{https://books.google.co.kr/books?id=aedk-OHdmNYC}
\BIBentrySTDinterwordspacing

\bibitem{sohn2018explicit}
J.-y. Sohn and J.~Moon, ``Explicit construction of mbr codes for clustered
  distributed storage,'' \emph{arXiv preprint arXiv:1801.02287}, 2018.

\bibitem{sohn2018class}
J.-y. Sohn, B.~Choi, and J.~Moon, ``A class of msr codes for clustered
  distributed storage,'' \emph{arXiv preprint arXiv:1801.02014}, 2018.

\bibitem{moon2005error}
T.~K. Moon, ``Error correction coding,'' \emph{Mathematical Methods and
  Algorithms. Jhon Wiley and Son}, 2005.

\end{thebibliography}

%

\renewenvironment{IEEEbiography}[1]
{\IEEEbiographynophoto{#1}}
{\endIEEEbiographynophoto}

\begin{IEEEbiography}{Jy-yong Sohn}
	(S'15) received the B.S. and M.S. degrees in electrical engineering from the Korea Advanced Institute of Science and Technology (KAIST), Daejeon, Korea, in 2014 and 2016, respectively. He is currently pursuing the Ph.D. degree in KAIST. His research interests include coding for distributed storage and distributed computing, massive MIMO effects on wireless multi cellular system and 5G Communications. He is a recipient of the IEEE international conference on communications (ICC) best paper award in 2017.
\end{IEEEbiography}

\begin{IEEEbiography}{Beongjun Choi}
	(S'17) received the B.S. and M.S. degrees in mathematics and electrical engineering from the Korea Advanced Institute of Science and Technology (KAIST), Daejeon, Korea, in 2014 and 2017. He is currently pursuing the electrical engineering Ph.D degree in KAIST. His research interests include error-correcting codes, distributed storage system and information theory. He is a co-recipient of the IEEE international conference on communications (ICC) best paper award in 2017.
\end{IEEEbiography}

\begin{IEEEbiography}{Sung Whan Yoon}
	(M'17) received the M.S. and Ph.D. degrees in electrical engineering from the Korea Advanced Institute of Science and Technology (KAIST), Daejeon, South Korea, in 2013 and 2017 respectively. He is currently a postdoctoral researcher in KAIST from 2017.
	His research interests are in the area of coding theory, distributed system and artificial intelligence, with focusing on polar codes, distributed storage system and meta-learning algorithm of neural network. Especially for the area of artificial intelligence, his primary interests include information theoretic analysis and algorithmic development of meta-learning. He was a co-recipient of the IEEE International Conference on Communications best Paper Award in 2017.
\end{IEEEbiography}

\begin{IEEEbiography}{Jaekyun Moon}
	(F'05) received the Ph.D degree in electrical and computer engineering at Carnegie Mellon University, Pittsburgh, Pa, USA. He is currently a Professor of electrical engineering at KAIST. From 1990 through early 2009, he was with the faculty of the Department of Electrical and Computer Engineering at the University of Minnesota, Twin Cities. He consulted as Chief Scientist for DSPG, Inc. from 2004 to 2007. He also worked as Chief Technology Officer at Link-A-Media Devices Corporation. His research interests are in the area of channel characterization, signal processing and coding for data storage and digital communication. Prof. Moon received the McKnight Land-Grant Professorship from the University of Minnesota. He received the IBM Faculty Development Awards as well as the IBM Partnership Awards. He was awarded the National Storage Industry Consortium (NSIC) Technical Achievement Award for the invention of the maximum transition run (MTR) code, a widely used error-control/modulation code in commercial storage systems. He served as Program Chair for the 1997 IEEE Magnetic Recording Conference. He is also Past Chair of the Signal Processing for Storage Technical Committee of the IEEE Communications Society. He served as a guest editor for the 2001 IEEE JSAC issue on Signal Processing for High Density Recording. He also served as an Editor for IEEE TRANSACTIONS ON MAGNETICS in the area of signal processing and coding for 2001-2006. He is an IEEE Fellow.
\end{IEEEbiography}\vfill

\end{document}